\documentclass[12pt]{article}

\usepackage{amsmath, amssymb, amsthm}
\usepackage{graphicx}
\usepackage{enumerate}
\usepackage{natbib}
\usepackage{bm}
\usepackage{dsfont}
\usepackage{url}  
\usepackage{microtype} 
\usepackage{tikz}
\usepackage{mathtools}
\usepackage{paralist}
\usepackage{booktabs} 
\usepackage{wasysym}
\usepackage{subcaption}
\usepackage{soul} 

\newcommand{\prob}{\operatorname{\mathsf{P}}}

\newcommand{\var}{\operatorname{\mathsf{Var}}}
\newcommand{\cov}{\operatorname{\mathsf{Cov}}}
\newcommand{\diff}{\mathrm{d}}
\newcommand{\abs}[1]{\left|{#1}\right|}

\newcommand{\point}{\,\cdot\,}
\newcommand{\eqd}{\stackrel{d}{=}}

\newcommand{\bR}{\bm{R}}
\newcommand{\bu}{\bm{u}}
\newcommand{\bU}{\bm{U}}

\newcommand{\bV}{\bm{V}}
\newcommand{\bw}{\bm{w}}
\newcommand{\bW}{\bm{W}}
\newcommand{\bx}{\bm{x}}
\newcommand{\bX}{\bm{X}}
\newcommand{\bY}{\bm{Y}}
\newcommand{\bZ}{\bm{Z}}
\newcommand{\by}{\bm{y}}
\newcommand{\bz}{\bm{z}}
\newcommand{\bzero}{\bm{0}}
\newcommand{\bone}{\bm{1}}
\newcommand{\binfty}{\bm{\infty}}

\newcommand{\EE}{\mathbb{E}}
\newcommand{\reals}{\mathbb{R}}
\newcommand{\Rd}{\reals^d}
\newcommand{\eps}{\varepsilon}

\newcommand{\symdiff}{\operatorname{\bigtriangleup}}
\newcommand{\norm}[1]{\lVert {#1} \rVert}
\newcommand{\1}{\operatorname{\mathds{1}}}
\renewcommand{\emptyset}{\varnothing}
\newcommand{\LL}{\mathbb{L}}
\newcommand{\dto}{\rightsquigarrow} 
\newcommand{\mBIC}{\operatorname{mBIC}}
\newcommand{\AIC}{\operatorname{AIC}}

\newcommand{\simplex}{\Delta}    
\newcommand{\Sd}{\simplex_{d-1}} 

\newcommand{\cbr}[1]{\left\{ {#1} \right\}}
\newcommand{\rbr}[1]{\left( {#1} \right)}
\newcommand{\sbr}[1]{\left[ {#1} \right]}

\newcommand{\xdf}{r} 
\newcommand{\xcdf}{R} 
\newcommand{\xqdf}{R^{-1}} 
\newcommand{\expmeas}{\Lambda} 
\newcommand{\bexpmeas}{\overline{\expmeas}} 
\newcommand{\nupdf}{\lambda} 
\newcommand{\IMP}{\bm{Z}} 
\newcommand{\IMPm}{Z} 
\newcommand{\Imp}{\IMPm} 
\newcommand{\imp}{\bm{z}} 
\newcommand{\xdfset}{\mathcal{R}} 
\newcommand{\cpdfset}{\mathcal{C}} 

\newcommand{\dset}{[d]}
\newcommand{\Vine}{\mathcal{V}}
\newcommand{\tree}{T}
\newcommand{\nodes}{N}
\newcommand{\edges}{E}
\newcommand{\cunn}{A}
\newcommand{\cndg}{D}
\newcommand{\cndd}{\mathcal{C}}
\newcommand{\iset}{K}

\newcommand{\half}{\tfrac{1}{2}}

\newcommand{\hF}{\widehat{F}}
\newcommand{\hU}{\widehat{U}}
\newcommand{\hbU}{\widehat{\bU}}

\newcommand{\hbu}{\widehat{\bu}}
\newcommand{\hZ}{\widehat{\Imp}}
\newcommand{\hbZ}{\widehat{\IMP}}
\newcommand{\hbz}{\widehat{\imp}}
\newcommand{\lik}{\mathcal{L}}
\newcommand{\htheta}{\hat{\theta}}
\newcommand{\hchi}{\widehat{\chi}}
\newcommand{\htau}{\widehat{\tau}}
\newcommand{\hGamma}{\widehat{\Gamma}}

\newcommand{\HR}{H\"{u}sler--Reiss}

\renewcommand{\ge}{\geqslant}
\renewcommand{\le}{\leqslant}

\renewcommand{\leq}{\leqslant}

\newcommand{\Sigk}{\Sigma^{(k)}}

\newcommand{\mk}{_{\setminus k}}

\newcommand{\pnt}{} 

\usepackage{xcolor}
\definecolor{lightskyblue}{rgb}{0.53, 0.81, 0.98}
\definecolor{lightgreen}{rgb}{0.56, 0.93, 0.56}
\definecolor{lavender(web)}{rgb}{0.9, 0.9, 0.98}
\definecolor{papayawhip}{rgb}{1.0, 0.94, 0.84}
\definecolor{royalblue(web)}{rgb}{0.25, 0.41, 0.88}
\definecolor{seagreen}{rgb}{0.18, 0.55, 0.34}
\definecolor{lightsalmon}{rgb}{1.0, 0.63, 0.48}
\definecolor{orchid}{rgb}{0.85, 0.44, 0.84}
\definecolor{sienna}{rgb}{0.53, 0.18, 0.09}
\definecolor{brown(web)}{rgb}{0.65, 0.16, 0.16}
\definecolor{navajowhite}{rgb}{1.0, 0.87, 0.68}

\graphicspath{{Fig/}}

\definecolor{darkteal}{rgb}{0,0.35,0.35}

\newcommand{\chng}[1]{#1}

\usepackage[utf8]{inputenc}
\usepackage[T1]{fontenc}
\usepackage[a4paper,left=2cm, right = 2cm, top = 2.75cm, bottom = 3.25cm]{geometry}
\usepackage[colorlinks=true,linkcolor=black,citecolor=black,urlcolor=black]{hyperref}
\usepackage{setspace}

\newtheorem{lemma}{Lemma}[section]
\newtheorem{proposition}[lemma]{Proposition}
\newtheorem{theorem}[lemma]{Theorem}
\newtheorem{assumption}[lemma]{Assumption}
\newtheorem{corollary}[lemma]{Corollary}

\theoremstyle{definition}
\newtheorem{definition}[lemma]{Definition}
\newtheorem{example}[lemma]{Example}

\theoremstyle{remark}
\newtheorem{remark}[lemma]{Remark}

\numberwithin{equation}{section}

\begin{document}

\title{X-Vine Models for Multivariate Extremes}

\author{Anna Kiriliouk\thanks{Namur Institute for Complex Systems, University of Namur, Rue Graf\'{e} 2, 5000 Namur, Belgium. E-mail: anna.kiriliouk@unamur.be, jeongjin.lee@unamur.be}
\and Jeongjin Lee\footnotemark[1]
\and Johan Segers\thanks{LIDAM/ISBA, UCLouvain, Voie du Roman Pays 20, 1348 Louvain-la-Neuve, Belgium. \chng{Corresponding author.} E-mail: johan.segers@uclouvain.be}}

\maketitle

\begin{abstract}
Regular vine sequences permit the organisation of variables in a random vector along a sequence of trees. Regular vine models have become greatly popular in dependence modelling as a way to combine arbitrary bivariate copulas into higher-dimensional ones, offering flexibility, parsimony, and tractability. In this project, we use regular vine structures to decompose and construct the exponent measure density of a multivariate extreme value distribution, or, equivalently, the tail copula density. Although these densities pose theoretical challenges due to their infinite mass, their homogeneity property offers simplifications. The theory sheds new light on existing parametric families and facilitates the construction of new ones, called X-vines. Computations proceed via recursive formulas in terms of bivariate model components. We develop simulation algorithms for X-vine multivariate Pareto distributions as well as methods for parameter estimation and model selection on the basis of threshold exceedances. The methods are illustrated by Monte Carlo experiments and a case study on US flight delay data.
\smallskip

\textbf{Key words:} exponent measure,
graphical model,
multivariate Pareto distribution,
pair copula construction,
regular vine,
tail copula
\end{abstract}

\section{Introduction}
\label{sec:intro}

For multivariate extremes, margin-free tail dependence models based on max-stable distributions arise from classical limit theory for sample extremes \citep{de1977limit}. A question of high interest is the construction of such models that are flexible, parsimonious, and computationally tractable, and scale well as the dimension grows \citep{engelke2021sparse}. To do so, we propose a novel approach based on regular vine tree sequences \citep{bedford2001probability,bedford2002vines}, called X-vines. The models can easily be built in arbitrary dimension by combining bivariate components only. The latter can be chosen independently from one another, giving great flexibility. The pairs are grouped in trees, the number of which determines the complexity of the model. Computations proceed by recursive algorithms.

For copula-based dependence modelling, outside the context of extreme value analysis, vine constructions have grown into a versatile and widely applied approach \citep{czado2019analyzing,czado2022vine,rvinecopulib}. Our contribution is to make vine-based methods operational, in theory and practice, for densities of exponent measures of multivariate extreme value distributions. The challenge to overcome is that these densities do not integrate to one but to infinity. A change of margin transforms \chng{these} into tail copula densities, which still have infinite mass, but with structural properties that resemble copula densities more closely.

Figure~\ref{fig:introvine} shows a regular vine sequence $\Vine$ in dimension $d = 5$. The structure consists of four trees $\tree_1,\ldots,\tree_4$, in which the edges of one tree become nodes in the next one. Each of the $d(d-1)/2 = 10$ pairs of variables appears as the leading pair before the semicolon on exactly one edge. The numbers behind the semicolons refer to conditioning variables. The first tree represents a Markov tree, while the subsequent trees add higher-order dependence relations.

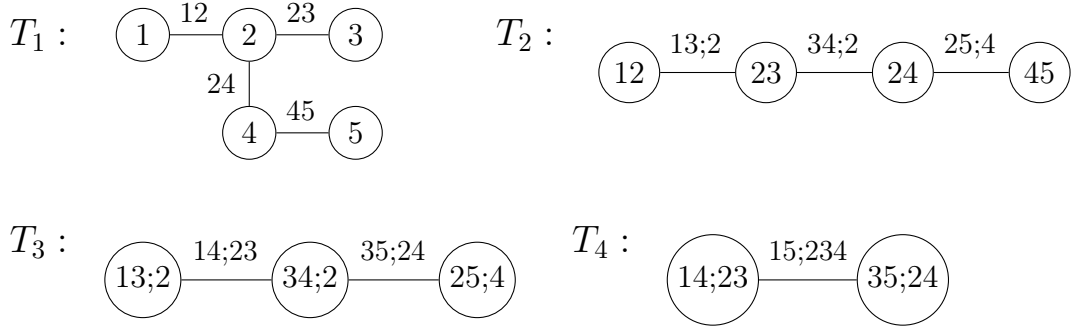
\begin{figure}
	\begin{center}
		\begin{tikzpicture}
			\node [circle,draw=black,fill=white,inner sep=0pt,minimum size=0.7cm] (1) at (0,0) {1};
			\node [circle,draw=black,fill=white,inner sep=0pt,minimum size=0.7cm] (2) at (1.4,0) {2};
			\node [circle,draw=black,fill=white,inner sep=0pt,minimum size=0.7cm] (3) at (2.8,0) {3};
			\node [circle,draw=black,fill=white,inner sep=0pt,minimum size=0.7cm] (4) at (1.4,-1.3) {4};
			\node [circle,draw=black,fill=white,inner sep=0pt,minimum size=0.7cm] (5) at (2.8,-1.3) {5};
			\path [draw,-] (1) edge (2);
			\path [draw,-] (2) edge (3);
			\path [draw,-] (2) edge (4);
			\path [draw] {(4) edge (5)};
			\node [text width=0.7cm] (T1) at (-1.4,0) {\large $\tree_1:$};
			\node [text width=0.7cm] (E1) at (0.85,0.3) {\small 12};
			\node [text width=0.7cm] (E1) at (1.4+0.85,0.3) {\small 23};
			\node [text width=0.7cm] (E1) at (1.2,-0.65) {\small 24};
			\node [text width=0.7cm] (E1) at (1.4+0.85,-1) {\small  45};
			
			\node [circle,draw=black,fill=white,inner sep=0pt,minimum size=0.8cm] (12) at (6+0.4,-0.5) {12};
			\node [circle,draw=black,fill=white,inner sep=0pt,minimum size=0.8cm] (23) at (6+1.8+0.4,-0.5) {23};
			\node [circle,draw=black,fill=white,inner sep=0pt,minimum size=0.8cm] (24) at (6+3.6+0.4,-0.5) {24};
			\node [circle,draw=black,fill=white,inner sep=0pt,minimum size=0.8cm] (45) at (6+5.4+0.4,-0.5) {45};
			\path [draw,-] (12) edge (23);
			\path [draw,-] (23) edge (24);
			\path [draw,] (24) edge (45);
			\node [text width=0.7cm] (T2) at (5,0) {\large $\tree_2:$};
			\node [text width=0.7cm] (E1) at (6.4+0.9,-0.15) {\small 13;2};
			\node [text width=0.7cm] (E1) at (6.4 + 1.8 +0.9,-0.15) {\small 34;2};
			\node [text width =0.7cm] (E1) at (6.4+ 3.6 + 0.9,-0.15) {\small 25;4};
			
			\node [circle,draw=black,fill=white,inner sep=0pt,minimum size=1cm] (13;2) at (0,-3.25) {13;2};
			\node [circle,draw=black,fill=white,inner sep=0pt,minimum size=1cm] (34;2) at (2.2,-3.25) {34;2};
			\node [circle,draw=black,fill=white,inner sep=0pt,minimum size=1cm] (25;4) at (4.4,-3.25) {25;4};
			\path [draw,-] (13;2) edge (34;2);
			\path [draw,-] (34;2) edge (25;4);
			\node [text width=0.7cm] (T3) at (-1.4,-2.75) {\large $\tree_3:$};
			\node [text width=0.7cm] (E1) at (1.025,-2.9) {\small 14;23};
			\node [text width=0.7cm] (E1) at (1.025+2.2,-2.9) {\small 35;24};
			
			\node [circle,draw=black,fill=white,inner sep=0pt,minimum size=1.2cm](14;23) at (7.5,-3.25) {14;23};
			\node [circle,draw=black,fill=white,inner sep=0pt,minimum size=1.2cm] (35;24) at (10,-3.25) {35;24};
			\path [draw,-] (14;23) edge (35;24);
			\node [text width=0.7cm] (T3) at (6,-2.75) {\large $\tree_4:$};    
			\node [text width=0.7cm] (E1) at (8.6,-2.9) {\small 15;234};
		\end{tikzpicture}
	\end{center}
	\caption{\label{fig:introvine} A regular vine sequence $\Vine$ in dimension $d = 5$ consisting of trees $\tree_1,\ldots,\tree_4$, where $\tree_j$ has $d-j+1$ nodes and $d-j$ edges. The nodes of $\tree_1$ are $\cbr{1,\ldots,5}$, while the edges of $\tree_j$ become the nodes of $\tree_{j+1}$.}
\end{figure}

An X-vine specification consists of a regular vine sequence $\Vine$, together with, for each edge in the first tree, a bivariate exponent measure or tail copula density, and, for each edge of the subsequent trees, a bivariate copula density, not necessarily stemming from extreme value theory. For the example in Fig.~\ref{fig:introvine}, four bivariate exponent measure densities ($\tree_1$) and six bivariate copula densities ($\tree_2,\tree_3,\tree_4$) would be required. These bivariate components can be chosen independently from another, without any constraints. Our main result shows how to combine bivariate building blocks into a single multivariate exponent measure or tail copula density. Further, we show that commonly used parametric models in extreme value analysis are actually examples of such X-vine models. We leverage the recursive nature of regular vine sequences both in theory and computations.

The models constructed in this way turn out to be the tail limits of regular vine copulas. For the special case of D-vine copulas, such tail limits were already computed in \citet{joe1996families} and \citet{joe2010tail}, focusing on tail dependence functions rather than densities. 
\chng{\cite{li2013extremal} define tail densities and compute those of D-vine copulas in dimensions three and four.}
In \citet{simpson2021geometric}, tails of D-vine and C-vine copulas were investigated from the perspective of the limiting shapes of sample clouds.

X-vine constructions are related to but different from recently proposed graphical models for extremes, either for directed or undirected graphs \citep{gissibl2018max,engelke2020graphical,engelke2022graphical,engelke2024graphical}. Markov trees \citep{segers2020one,hu2023modelling} are a special case in the intersection of graphical and X-vine models.

After reviewing background on multivariate extremes and regular vine sequences in Section~\ref{sec:background}, we state a version of Sklar's theorem for tail copula densities in Section~\ref{sec:Sklar}. Some parametric examples are worked out in Section~\ref{sec:parametric}, with a focus on the simplifying assumption that the conditional copula densities only depend on the indices of the conditioning variables, but not their actual values. Section~\ref{sec:vine} contains the paper's main results, showing on the one hand how a general tail copula density can be decomposed along any regular vine sequence, and, on the other hand, how to construct a tail copula density from a regular vine sequence and bivariate building blocks. The so-called X-vine models arising in this way are put to work in the subsequent sections, covering simulation algorithms (Section~\ref{sec:simulation}), parameter estimation and model selection (Section~\ref{sec:estimation}), simulation studies (Section~\ref{sec:simulation-study}), and a case study on US flight delay data (Section~\ref{sec:casestudy}) taken from \citet{hentschel2022statistical}. Section~\ref{sec:conc} concludes. The \chng{supplementary material contains a detailed example to illustrate our theory, 
the proofs of the paper's results, 
expressions based on exponent measure densities,
and additional numerical results}. The methods are implemented \chng{in the R \citep{R} package Xvine\footnote{Available from \url{https://github.com/JeongjinLee88/Xvine/tree/main}}}, relying in particular on packages \textsf{graphicalExtremes} \citep{engelke2graphicalextremes} and \textsf{VineCopula} \citep{vinecopula}.

\section{Background}
\label{sec:background}

We write $\dset = \cbr{1,\ldots,d}$ for the index set of the variables. Bold symbols refer to multivariate quantities. For a point $\bx = (x_1,\ldots,x_d) \in \Rd$ and a subset $J \subseteq \dset$, write subvectors as $\bx_J=(x_j)_{j \in J}$ and $\bx_{\setminus J} = (x_j)_{j \in \dset \setminus J}$.
Mathematical operations on vectors such as addition, multiplication and comparison are considered component-wise.

\subsection{Multivariate extreme value theory: tail copulas and their densities}

\paragraph{Tail copulas.}
Classical extreme value theory starts from the assumption that the distribution function $F$ of a random vector $\bX = (X_1,\ldots,X_d)$ is in the max-domain of attraction of a multivariate extreme value distribution \chng{\citep{beirlant2004statistics,de2007extreme}}.
\chng{This assumption concerns the tails of the univariate marginal distribution functions $F_1,\ldots,F_d$ and the tail dependence structure of $\bX$.} 
\chng{Our focus is on 
the latter aspect, that is, \chng{on} probabilities of high values occurring jointly among the variables $X_j$.}
\chng{Let} $C$ be the survival copula of $\bX$: under the standing assumption that $F_1,\ldots,F_d$ are continuous, $C$ is the distribution function of the random vector $\bU = (U_1,\ldots,U_d)$ with uniformly distributed components $U_j = 1-F_j(X_j)$ for $j \in \dset$. Our interest is in high values of $X_j$ and thus in low outcomes of $U_j$.

\begin{definition}[Tail copula and tail copula measure]
\label{def:expmeas}
The \emph{(lower) tail copula} $\xcdf$ of $C$ is the function on $\EE = (0, \infty]^d \setminus \{\binfty\}$ defined by
\begin{equation}
\label{eq:C2Lambda}
	\xcdf(\bx) = \lim_{t \searrow 0} t^{-1} C(t \bx),
	\qquad \bx \in \EE,
\end{equation}	
provided the limit exists. If it does, then the \emph{tail copula measure}, denoted by the same symbol $\xcdf$, is the Borel measure on $\EE$ determined by $\xcdf((\bzero, \bx]) = \xcdf(\bx)$ for $\bx \in \EE$.
\end{definition}

The term `tail copula' stems from \citet{schmidt2006non} in \chng{case $d = 2$. The tail copula measure already appears in \citet{einmahl2001nonparametric} and is closely linked to the exponent measure $\Lambda$ in Eq.~\eqref{eq:Lam2nu} below, introduced in \citet{de1977limit}. The limit~\eqref{eq:C2Lambda} appears in \citet{jaworski2006uniform} and in \citet{joe2010tail}, who call it `tail dependence function'.}

The total mass of the tail copula measure $\xcdf$ equals infinity but $\xcdf(B)$ is finite for Borel sets $B$ contained in $\cbr{ \bx \in \EE : \min \bx \le M }$ for some $M > 0$. Its univariate margins are equal to the one-dimensional Lebesgue measure:
\begin{equation}\label{eq:margin}
	\forall j \in \dset, \; \forall x_j \in (0, \infty), \qquad
	\xcdf\rbr{ \cbr{ \by \in \EE : y_j \le x_j } }
	= x_j.
\end{equation}
Equation~\eqref{eq:C2Lambda} implies that $\xcdf$ is homogeneous of order one, both as a function and as a measure: for $\bx \in \EE$, Borel sets $B \subseteq \EE$, and $s \in (0, \infty)$,
\begin{equation}
\label{eq:Lambdahomo}
	\xcdf(s \bx) = s \, \xcdf(\bx)
	\qquad \text{and} \qquad
	\xcdf(s B) = s \, \xcdf(B).
\end{equation}
Any Borel measure $\xcdf$ on $\EE$ satisfying \eqref{eq:margin} and \eqref{eq:Lambdahomo} is a tail copula measure, i.e., is the limit in \eqref{eq:C2Lambda} for some copula $C$\chng{; see Lemma~\ref{lem:allRok} in the supplement}.

\paragraph{Tail copula densities.}
Throughout, we assume that the tail copula measure $\xcdf$ has no mass on the hyperplanes through infinity, $\xcdf \rbr{ \cbr{ \bx \in \EE : x_j = \infty } } = 0$ for all $j \in \dset$, so $\xcdf$ is supported on $(0, \infty)^d$.
Further, we assume that $\xcdf$ is absolutely continuous with respect to the $d$-dimensional Lebesgue measure with continuous
density $\xdf : (0, \infty)^d \to [0, \infty)$, that is,
$\xcdf(B) = \int_{B} \xdf(\bx) \, \diff \bx$
for all Borel sets $B \subseteq (0, \infty)^d$. 
Choosing $B = \left( \bzero, \bx \right]$ for $\bx \in (0, \infty)^d$, we recover $\xdf$ from $\xcdf$ by $\xdf(\bx) = \frac{\partial^d}{\partial x_1 \cdots \partial x_d} \xcdf(\bx)$ for all $\bx \in (0, \infty)^d$. 
The marginal constraint~\eqref{eq:margin} implies
\begin{equation}
\label{eq:margin:pdf}
	\forall j \in \dset, \; \forall x_j \in (0,\infty), \qquad
	\int_{(0,\infty)^{d-1}} \xdf(\bx) \, \diff \bx_{\setminus j} = 1.
\end{equation}
In view of~\eqref{eq:Lambdahomo}, $\xdf$ is homogeneous of order $1-d$:
\begin{equation}
\label{eq:lambdahomo}
	\forall s \in (0, \infty), \; 
	\forall \bx \in (0, \infty)^d, \qquad
	\xdf(s \bx) = s^{1-d} \, \xdf(\bx).
\end{equation}
Properties~\eqref{eq:margin:pdf} and~\eqref{eq:lambdahomo} characterise the set of $d$-variate \emph{tail copula densities}; see Section~\ref{sec:parametric} for some parametric families.

By Scheffé's lemma, if the copula $C$ has density $c$ and the tail copula measure $\xcdf$ has density $\xdf$, then \eqref{eq:C2Lambda} is implied by $\lim_{t \searrow 0} t^{1-d} c(t \bx) = \xdf(\bx)$ for all $\bx \in (0, \infty)^d$. \chng{\cite{li2013extremal} call $\xdf$ the (lower) tail density function of $C$.} 
A word of caution: for a given copula density $c$, the limit $\lim_{t \searrow 0} t^{1-d} c(t \bx)$ may exist for all $\bx \in (0, \infty)^d$ but not be a tail copula density, as the marginal constraint \eqref{eq:margin:pdf} may fail. A case in point is the independence copula, with density $c \equiv 1$, in which case the said limit is zero.

\paragraph{Margins of tail copulas.}
For non-empty $J \subseteq \dset$, let $\pi_J : (0, \infty)^d \to (0, \infty)^J$ denote the coordinate projection $\bx \mapsto \pi_J(\bx) = \bx_J$ and let $\xcdf_J  = \xcdf \circ \pi_J^{-1}$ denote the $J$-th marginal measure 
\begin{equation}
\label{eq:RJ}
	\xcdf_J(B) 
	= \xcdf\rbr{ \pi_J^{-1}(B) } 
	= \xcdf\rbr{ \cbr{ \bx \in (0, \infty)^d : \bx_J \in B } },
\end{equation}
for Borel sets $B \subseteq (0, \infty)^J$. The choice $B = \prod_{j \in J} \left(0, x_j\right]$ for $x_j \in (0,\infty)$ shows that $\xcdf_J$ is the tail copula measure of the copula $C_J$ of $\bU_J = \pi_J(\bU)$, as in Definition~\ref{def:expmeas}. For $\bx_J \in (0,\infty]^J \setminus \cbr{\binfty}$, we have
$\xcdf_J(\bx_J) = \xcdf(\tilde{\bx}_J)$ with
	$\tilde{x}_j =
	x_j$ if $j \in J$ and $\tilde{x}_j =
	\infty$ if $j \in \dset \setminus J$.
\chng{For $\bx_{J} \in (0, \infty)^J$,} the density $\xdf_J$ of the measure $\xcdf_J$ is
\begin{equation}
\label{eq:lambdaJ}
	\xdf_J ( \bx_J )
	= \frac{\partial^{|J|}}{\prod_{j \in J} \partial x_j} \xcdf_J \left( \prod_{j \in J} \left( 0, x_j \right] \right) 
	= \int_{\bx_{\setminus J} \in (0, \infty)^{d-|J|}} \xdf(\bx) \, \diff \bx_{\setminus J},
\end{equation}
and, provided $J \neq \dset$, is obtained from $\xdf$ by integrating out the variables $x_j$ with indices $j \not\in J$.

\paragraph{Relations between tail copulas and other concepts in multivariate extremes.}

\chng{We review} a number of related terms from multivariate extreme value theory. This paragraph can be skipped at first reading, as the contributions in Sections~\ref{sec:Sklar} to~\ref{sec:vine} do not depend on it.

Equation~\eqref{eq:C2Lambda} is equivalent to the assumption that the random vector $\bV = (V_1,\ldots,V_d)$ with $V_j = 1/U_j$ for $j \in \dset$ satisfies
\begin{equation}
\label{eq:max-stable}
	\lim_{n \to \infty}
	\prob(\bV \le n \bz)^n
	= G(\bz) = \exp \cbr{ - \ell \rbr{1/\bz} }, \qquad \bz \in (0, \infty)^d,
\end{equation}
that is, $\bV$ is in the max-domain of attraction of a \emph{multivariate extreme value} (or \emph{max-stable}) distribution $G$ defined in terms of the \emph{stable tail dependence function} 
$\ell(\bx) = \xcdf \rbr{\EE \setminus [\bx, \binfty]}$ for $\bx \in [0, \infty)^d$.
The margins of $\bV$ are unit-Pareto, $\prob(V_j > x) = 1/x$ for $x \ge 1$, while those of $G$ are unit-Fréchet, $G_j(x) = \exp(-1/x)$ for $x > 0$.
Equivalently, the distribution of $\bV$ is \emph{multivariate regularly varying} \citep{de1977limit, resnick2007heavy} with limit measure
\begin{equation}
	\label{eq:Lam2nu}
	\expmeas(\,\cdot\,) = \xcdf \rbr{ \cbr{ \bx \in \EE : (1/x_1, \ldots, 1/x_d) \in \cdot \, } } \qquad \text{ on } [0, \infty)^d \setminus \cbr{\bzero},
\end{equation}
that is, $t \prob(\bV \in t B) \to \expmeas(B)$ as $t \to \infty$ for Borel sets $B$ contained in $[0, \infty)^d \setminus [0, \eps]^d$ for some $\eps > 0$ and satisfying $\expmeas(\partial B) = 0$, with $\partial B$ the boundary of $B$. The measure $\expmeas$ is called \emph{exponent measure} 
because of the identity $G(\bz) = \exp \cbr{ - \expmeas \rbr{\left[\bzero, \binfty\right) \setminus \sbr{\bzero, \bz}}}$ for $\bz \in (0, \infty)^d$.

Another statement equivalent to~\eqref{eq:C2Lambda} is that the conditional distribution of $\bV/t$ given $\max \bV > t$ is asymptotically \emph{multivariate Pareto} 
that is, we have the weak convergence 
\begin{equation}
\label{eq:Vt2Y}
	\rbr{\bV/t \mid \max \bV > t}
	\dto
	\bY, \qquad t \to \infty,
\end{equation}
where $\bY$ is a random vector supported on
$\LL_{>1} := \cbr{\by \in [0, \infty)^d : \max \by > 1}$ with distribution equal to the restriction of $\expmeas$ in \eqref{eq:Lam2nu} to $\LL_{>1}$ and normalized to a probability measure,  
\begin{equation}
	\label{eq:mgpd}
	\prob(\bY \in B) = \expmeas(B \cap \LL_{>1}) / \expmeas(\LL_{>1})
\end{equation}
for Borel sets $B \subseteq \Rd$ and normalizing constant $\expmeas(\LL_{>1}) = \ell(\bone)$.
For $\by \in \LL_{>1}$, we also have $\prob(\bY \ge \by) = R(1/\by) / \ell(\bone)$.

Up to a location shift, multivariate Pareto distributions are a special case of \emph{multivariate generalised Pareto distributions} introduced in \citet{rootzen2006multivariate}. They play a prominent role in the theory of graphical models for extremes \citep{engelke2020graphical}.

\chng{In practice, it may be convenient to reduce the information contained in the tail copula measure to a dependence coefficient. In this paper, we will encounter the \emph{tail dependence coefficients} $\chi_J = \xcdf(\{\bx \in (0, \infty]^d : \max_{j \in J} x_j < 1\}) = R_J(1,\ldots,1)$ for index sets $J \subseteq \dset$ with two or more elements.}

If the tail copula measure $\xcdf$ is concentrated on $(0, \infty)^d$ and has density $\xdf$, 
the exponent measure and the multivariate Pareto distribution 
have densities too. The exponent measure $\Lambda$ is concentrated on $(0, \infty)^d$ and has density 
\begin{equation}
\label{eq:nupdf}
	\nupdf(\by) = \xdf(1/\by) \prod_{j=1}^d y_j^{-2},
	\qquad \by \in (0, \infty)^d,
\end{equation}
while the multivariate Pareto vector $\bY$ has probability density
$\nupdf(\by) / \ell(\bone)$ 
for $\by \in \LL_{>1}$.

\subsection{Regular vine sequences and the vine telescoping product formula}
\label{sec:Rvine}

Informally, a regular vine sequence on $d$ elements consists of a linked set of $d-1$ trees, where the edges in tree $j-1$ become nodes in tree~$j$, joined by an edge only if they share a common node as edges in tree $j-1$, for $j \in \cbr{2,\ldots,d-1}$ \citep{bedford2002vines}. \emph{Regular vine copulas}, also called \emph{pair copula constructions}, allow for flexible high-dimensional modelling using parametric families of bivariate copulas as building blocks for the edges of each tree.
For recent overviews, see, for example, \chng{\cite{czado2019analyzing} or \cite{czado2022vine}.} Our objective is to do the same for tail copula densities. \chng{All concepts introduced below are illustrated via a five-dimensional example in Section~\ref{app:example} of the supplementary material}. 

A \emph{tree} $\tree = (\nodes, \edges)$ is a connected, acyclic graph comprising a finite node set $\nodes$ and an edge set $\edges \subseteq \cbr{\cbr{x,y}: \; x,y \in \nodes, \, x \ne y}$. Any two distinct nodes in a tree are connected by a unique path.

\begin{definition}[Regular vine \chng{(}tree\chng{)} sequence]\label{def:rvs}
	An ordered set of trees $\Vine = (\tree_1,\ldots,\tree_{d-1})$ is a \emph{regular vine tree sequence} on $d \ge 3$ elements if $\tree_1 = (\nodes_1,\edges_1)$ is a tree with node set $\nodes_1 = \dset$, while for $j \in \{2,\ldots,d-1\}$, the tree $\tree_j = (\nodes_j, \edges_j)$ has node set $\nodes_j = \edges_{j-1}$ and the \emph{proximity condition} holds: 
	for any $\{a,b\} \in \edges_j$, we have $|a \cap b| = 1$, that is, two nodes in $\tree_j$ can be connected only if, as edges in $\tree_{j-1}$, they share a common node.
\end{definition}

\begin{definition}[Complete union, conditioning set, conditioned set]\label{def:ucd}
	Let $\Vine = (\tree_1,\ldots,\tree_{d-1})$ be a regular vine sequence, with $\tree_j = (\nodes_j,\edges_j)$.
	For any edge $e \in \edges_j$, the \emph{complete union} of $e$ is $\cunn_e = e$ if $j = 1$ and $\cunn_e = \cunn_a \cup \cunn_b$ if $e = \{a,b\}$ and $j \in \cbr{2,\ldots,d-1}$. Informally, $\cunn_e$ is the subset of nodes in $\dset$ reachable from $e$ by the membership relation.
	The \emph{conditioning set} of an edge $e = \{a,b\} \in \edges_j$ with $j \ge 2$ is $\cndg_e = \cunn_a \cap \cunn_b$ and the \emph{conditioned set} of $e$ is $\cndd_e = \cndd_{e,a} \cup \cndd_{e,b}$ with $\cndd_{e,a} = \cunn_a \setminus \cndg_e = \cunn_a \setminus \cunn_b$ and $\cndd_{e,b} = \cunn_b \setminus \cndg_e = \cunn_b \setminus \cunn_a$. If $e = \{a,b\} \in \edges_1$, then $\cndd_{e,a} = \{a\}$, $\cndd_{e,b} = \{b\}$, $\cndd_e = e$ and $\cndg_e = \emptyset$.
\end{definition}

\chng{For an edge $e = \cbr{a,b} \in \edges_j$, the sets $\cndd_{e,a}$, $\cndd_{e,b}$ and $\cndg_e$ are disjoint and their union equals $\cunn_e$, a node set with $j+1$ elements. The sets $\cndd_{e,a}$ and $\cndd_{e,b}$ are singletons while $\cndg_e$ has $j-1$ elements. Clearly, $\cndg_e = \cunn_a \setminus \cndd_e = \cunn_b \setminus \cndd_e$.}
The edge $e$ is also written as $e = \rbr{\cndd_{e,a}, \cndd_{e,b}; \cndg_e}$ and, since $\cndd_{e,a}$ and $\cndd_{e,b}$ are singletons, $\{a_e\}$ and $\{b_e\}$, say, we abbreviate $e = \rbr{a_e, b_e; \cndg_e}$, with $a_e, b_e \in \dset$ and the convention that $a_e < b_e$, that is, $a_e = \min \cndd_e$ and $b_e = \max \cndd_e$. Similarly, for $e = \cbr{a, b} \in \edges_1$, the convention is that $a < b$. In Fig.~\ref{fig:introvine}, for instance, the edge $e$ labelled $(14;23)$ in $\tree_3$ has $a_e = 1$, $b_e = 4$, and $\cndg_e = \cbr{2,3}$.
\chng{See Section~\ref{app:proofs:background} in the supplement for some additional properties.}

\chng{For positive scalars $\gamma_1,\ldots,\gamma_d$, the identity $\prod_{j=2}^d (\gamma_j / \gamma_{j-1}) = \gamma_d/\gamma_1$ is a \emph{telescoping product}.} Regular vine sequences enjoy a similar property that will be key to Theorem~\ref{thm:xvine} and which, to the best of our knowledge, has not yet been stated in the literature.

\begin{lemma}[Vine telescoping product]
	\label{lem:Rvine-a}
	Let $d \ge 3$ be an integer and let $\gamma_J \in (0, \infty)$ for every non-empty $J \subseteq \dset$,
	with $\gamma_j := \gamma_{\{j\}} = 1$ for every $j \in \dset$.
	Given a regular vine sequence $\Vine = \rbr{\tree_j = \rbr{\nodes_j,\edges_j}}_{j=1}^{d-1}$ on $\dset$, we have
	\begin{equation}
		\label{eq:Rvine-a}
		\gamma_{\dset}
		=
		\prod_{e \in \edges_1} \gamma_{e} \cdot
		\prod_{j=2}^{d-1} \prod_{e = \{a, b\} \in \edges_j}
		\frac{\gamma_{\cndg_e} \cdot \gamma_{\cunn_e}}{\gamma_{\cunn_a} \cdot \gamma_{\cunn_b}}.
	\end{equation}
	With the notation $\gamma_{I|J} = \gamma_{I \cup J} / \gamma_{J}$, the last factor in \eqref{eq:Rvine-a} is 
	\begin{equation}
		\label{eq:Rvine-a-cond}
		\frac{\gamma_{\cndg_e} \cdot \gamma_{\cunn_e}}{\gamma_{\cunn_a} \cdot \gamma_{\cunn_b}}
		=
		\frac{\gamma_{\{a_e,b_e\}|\cndg_e}}{\gamma_{a_e|\cndg_e} \cdot \gamma_{b_e|\cndg_e}}.
	\end{equation}
\end{lemma}

\chng{Remark~\ref{rem:telescope-f} in Appendix~\ref{app:proofs:background} in the supplement generalises Lemma~\ref{lem:Rvine-a} to factorisations of $\gamma_J$ for certain non-empty subsets $J \subseteq \dset$, while Remark~\ref{rem:Rvine-c} provides a generalisation of \eqref{eq:Rvine-a} without the restriction that $\gamma_j = 1$ for all $j \in \dset$.}

\section{Sklar's theorem for tail copula densities}
\label{sec:Sklar}

\chng{Let $\xdf$ be a $d$-variate tail copula density, so properties~\eqref{eq:margin:pdf} and~\eqref{eq:lambdahomo} hold.} Recall its multivariate margins in \eqref{eq:lambdaJ}. For non-empty, disjoint $I, J \subset \dset$ and for $\bx_I \in (0, \infty)^I$ and $\bx_J \in (0, \infty)^J$ such that $\xdf_J(\bx_J) > 0$, define $\xdf_{I|J}(\bx_I|\bx_J) := \xdf_{I \cup J}(\bx_{I \cup J})/\xdf_J(\bx_J)$ for all $\bx_{I \cup J} \in (0, \infty)^{I \cup J}$.
Viewing the ``conditional'' tail copula density $\xdf_{I|J}$ as an $|I|$-variate probability density, we can decompose it into $|I|$ univariate probability densities and a copula density. 

\begin{proposition}[Sklar's theorem: direct part]
\label{lem:expdensity}
Let $\xdf$ be a $d$-variate tail copula density and let $I, J \subset \dset$ be non-empty and disjoint.
For $\bx_J \in (0, \infty)^J$ such that $\xdf_J(\bx_J) > 0$, the function
$\xdf_{I|J}(\point|\bx_J) : (0, \infty)^I \to \left[0, \infty\right)$, defined via $\bx_I \mapsto \xdf_{I|J}(\bx_I|\bx_J)$, is a probability density on $(0, \infty)^I$. For $i \in I$, it has marginal probability density $x_i \mapsto \xdf_{i|J}(x_i|\bx_J)$ and cumulative distribution function 
$x_i \mapsto \xcdf_{i|J}(x_i|\bx_J) := \int_{0}^{x_i} \xdf_{i|J}(t|\bx_J) \, \diff t$,
	with quantile function $u_i \mapsto \xcdf_{i|J}^{-1}(u_i|\bx_J)$.
	If $|I| \ge 2$, then the copula density of $\xdf_{I|J}(\point|\bx_J)$ is
	\begin{equation}\label{eq:cIJ}
	\forall \bu_I \in (0, \infty)^I, \qquad
	c_{I;J}(\bu_I;\bx_J) := \frac{\xdf_{I|J}(\bx_I|\bx_J)}{\prod_{i\in I} \xdf_{i|J}(x_i|\bx_J)}
	\quad \text{with} \quad x_i = \xcdf_{i|J}^{-1}(u_i|\bx_J). 
	\end{equation}
\end{proposition}

The homogeneity property~\eqref{eq:lambdahomo} induces certain scaling properties.

\begin{lemma}[Scale equi- and invariance]
\label{lem:homo}
	In Proposition~\ref{lem:expdensity}, for $\bx_J \in (0, \infty)^J$ such that $\xdf_J(\bx_J) > 0$, the family $\cbr{\xdf_{I|J}(\point|t\bx_J) : t \in (0,\infty)}$ of probability densities is a scale family: if the random vector $\bm{\xi}_I$ has probability density $\xdf_{I|J}(\point|\bx_J)$, then the random vector $t \bm{\xi}_I$ has probability density $\xdf_{I|J}(\point|t\bx_J)$. As a consequence, for all $t \in (0,\infty)$, $\bx_I \in (0, \infty)^I$, \chng{$\bu_I \in (0, 1)^I$} and $i \in I$, we have $\xdf_{i|J}(tx_i|t\bx_J) 
		= \xdf_{i|J}(x_i|\bx_J)$, \chng{$\xcdf_{i|J}^{-1}(u_i|t\bx_J) = t \cdot \xcdf_{i|J}^{-1}(u_i|\bx_J)$, and}
	\begin{equation}
	\label{eq:cIJhomo}
		c_{I;J}(\bu_I;t\bx_J) = c_{I;J}(\bu_I;\bx_J).
	\end{equation}
	In particular, if $J = \cbr{j}$ for some $j \in \dset$, then $c_{I;j}(\point;x_j) = c_{I;j}(\point;1)$ does not depend on the value of $x_j \in (0,\infty)$.
\end{lemma}

\begin{remark}[Conditional independence]
Let $\bY$ be the multivariate Pareto vector in \eqref{eq:mgpd}, suppose $d \ge 3$, and let $i,j$ be distinct elements in $\dset$, with $J = \dset \setminus \cbr{i,j}$. Provided $\xdf$ is positive and continuous, $Y_i$ and $Y_j$ are conditionally independent given $\bY_{J}$ in the sense of Definition~1 in \citet{engelke2020graphical} if and only if $c_{i,j;J}(u_i,u_j;\bx_J) = 1$ for all $(u_i,u_j) \in (0,1)^2$ and $\bx_J \in (0, \infty)^J$, the density of the independence copula.
This statement follows by combining the factorization property in Proposition~1 in \citet{engelke2020graphical} with equations \eqref{eq:nupdf} and \eqref{eq:cIJ};
see also \eqref{eq:lam2c} in \chng{Appendix~\ref{app:expmeas} in the supplement} for an expression of $c_{I;J}$ in terms of the exponent measure density $\nupdf$.
\end{remark}

\begin{definition}[Simplifying assumption]
\label{def:simplif}
The family of copula densities $c_{I;J}(\point;\bx_J)$ in \eqref{eq:cIJ} satisfies the \emph{simplifying assumption} if the copula densities do not depend on $\bx_J \in (0,\infty)^J$, that is, there exists a single copula density $c_{I;J}$ such that $c_{I;J}(\bu_I;\bx_J) = c_{I;J}(\bu_I)$ for all $\bu_I \in (0, 1)^I$ and $\bx_J \in (0, \infty)^J$.
\end{definition}

If $J$ is a singleton, $J = \cbr{j}$, the scale invariance \eqref{eq:cIJhomo} implies that $\cbr{c_{I;j}(\point;x_j) : x_j \in (0,\infty)}$ always satisfies the simplifying assumption. In Section~\ref{sec:parametric}, we will see that the copula families induced by many commonly used parametric exponent measure models satisfy the simplifying assumption for all $I$ and $J$. 

As a converse to Proposition~\ref{lem:expdensity}, we can combine several tail copula densities along a scale-invariant family of copula densities to form a tail copula density in higher dimensions. With slight abuse of notation, write $i \cup J := \{i\} \cup J$ for $i \in I$. A family of $k$-variate copula densities $\cbr{c(\point;\theta) : \theta \in \Theta}$ indexed by a Borel set $\Theta \subseteq \reals^m$ is \emph{jointly measurable} if the map $(\bu,\theta) \mapsto c(\bu;\theta)$ is Borel measurable on $(0,1)^k \times \Theta$. 

\begin{proposition}[Sklar's theorem: converse part]
	\label{lem:Sklar}
	Let $I \cup J$ be a partition of $\dset$ into two disjoint, non-empty sets, with $d \ge 3$ and $|I| \ge 2$.
	Let $(\xdf_{i \cup J})_{i \in I}$ be an $|I|$-tuple of $(|J|+1)$-variate tail copula densities with common margin $\xdf_J$.
	Let $\cbr{ c_{I;J}(\point;\bx_J) : \bx_J \in (0, \infty)^J }$ be a jointly measurable family of $|I|$-variate copula densities such that~\eqref{eq:cIJhomo} holds. Then the function $\xdf : (0, \infty)^d \to [0, \infty)$ defined by
	\begin{equation}
		\label{eq:Sklarconstructive}
		\xdf(\bx)
		:= \xdf_J(\bx_J) \cdot \prod_{i \in I} \xdf_{i|J}(x_i|\bx_J) \cdot
		c_{I;J} \rbr{ \xcdf_{i|J}(x_i|\bx_J), i \in I ; \bx_J },
	\end{equation}
	for $\bx \in (0, \infty)^d$ such that $\xdf_J(\bx_J) > 0$ and zero otherwise, is a $d$-variate tail copula density with margins $\xdf_{i \cup J}$ for $i \in I$.
\end{proposition}

\section{Parametric families of tail copula densities}
\label{sec:parametric}

We will compute the copula densities $c_{I;J}(\point;\bx_j)$ in \eqref{eq:cIJ} for tail copula densities $\xdf$ arising from two parametric families: the scaled extremal Dirichlet model (Section~\ref{sec:scdiri}) due to \citet{belzile2017extremal}, encompassing the logistic, negative logistic and Dirichlet max-stable models, and the \HR{} model (Section~\ref{sec:hr}). For both families, the simplifying assumption (Definition~\ref{def:simplif}) is satisfied for all choices of $I$ and $J$, and the copula densities $c_{I;J}$ take on known parametric forms. This section can be skipped at first reading.

\subsection{Scaled extremal Dirichlet model}
\label{sec:scdiri}

In dimension $d \ge 2$, consider parameters $\alpha_1,\ldots,\alpha_d > 0$ and $\rho > - \min(\alpha_1,\ldots,\alpha_d)$ with $\rho \ne 0$. 
The angular measure density of the \emph{scaled extremal Dirichlet} model is introduced in \citet[Proposition~5]{belzile2017extremal}.
By Eq.~\eqref{eq:nupdf} and Lemma~\ref{lem:lambdah} \chng{in the supplement}, the corresponding tail copula density is
\begin{equation}
	\label{eq:xdfscdiri}
	\xdf(\bx) = \frac{\Gamma(\bar\alpha+\rho)}{d |\rho|^{d-1} \prod_{j=1}^d \Gamma(\alpha_j)}
	\cdot \cbr{\sum_{j=1}^d c(\alpha_j,\rho)^{1/\rho} x_j^{-1/\rho}}^{-\rho-\bar\alpha}
	\cdot \prod_{j=1}^d	c(\alpha_j,\rho)^{\alpha_j/\rho} x_j^{-\alpha_j/\rho-1},
\end{equation}
for $\bx \in (0, \infty)^d$, where $\bar{\alpha} = \alpha_1+\cdots+\alpha_d$ and $c(\alpha,\rho) = \Gamma(\alpha+\rho)/\Gamma(\alpha)$.
The model is closed under marginalisation: for $J \subseteq \dset$ with $|J| \ge 2$, the marginal tail copula density $r_J$ in \eqref{eq:lambdaJ} is of the same form in dimension $|J|$ and with parameters $(\alpha_j)_{j \in J}$ and $\rho$. The model unites and extends several well-known parametric families in multivariate extreme value analysis: 
\begin{compactitem}
	\item 
	if $\alpha_1 = \ldots = \alpha_d = 1$ and $\rho > 0$, we get the negative logistic model \citep{joe1990families};
	\item 
	if $\alpha_1 = \ldots = \alpha_d = 1$ and $-1 < \rho < 0$, we obtain the logistic model \citep{gumbel1960distributions};
	\item 
	if $\rho = 1$, we get the Coles–Tawn extremal Dirichlet model \citep{coles1991modelling}.
\end{compactitem}

In \citet[Definition~3]{mcneil2010archimedean}, the family of \emph{Liouville} copulas is defined as the collection of survival copulas of random vectors of the form $\bY = S \bW$, where $\bW$ has a Dirichlet distribution on the unit simplex and is independent of the positive random variable~$S$.
The special case where $\bW$ is uniformly distributed on the unit simplex yields the family of \chng{\emph{Archimedean}} copulas \citep{mcneil2009multivariate}.

\begin{proposition}
	\label{prop:scdiri}
	Let $\xdf$ be the scaled extremal Dirichlet tail copula density in \eqref{eq:xdfscdiri} in dimension $d \ge 3$. For every pair of disjoint, non-empty subsets $I, J \subset \dset$ with $|I| \ge 2$, the copula densities $c_{I;J}(\point;\bx_J)$ in \eqref{eq:cIJ} satisfy the simplifying assumption (Definition~\ref{def:simplif}). If $\rho > 0$, then $c_{I;J}$ is equal to \chng{the density of} the $|I|$-variate Liouville copula with Dirichlet parameters $(\alpha_i)_{i \in I}$ and with radial density on $(0, \infty)$ proportional to $s \mapsto s^{\sum_{i \in I} \alpha_i - 1} \cdot \rbr{s+1}^{-\rho-\sum_{k \in I \cup J} \alpha_k}$.
	If $\rho < 0$, then $c_{I;J}$ is equal to the density of the survival copula of the above Liouville copula.
\end{proposition}

In special cases, the copula densities $c_{I;J}$ can be identified more explicitly (calculations in \chng{Appendix~\ref{sec:proof:param} in the supplement}):
\begin{compactitem}
	\item if $\alpha_1 = \ldots = \alpha_d = 1$ and $\rho > 0$, then $c_{I;J}$ is the density of the $|I|$-variate Clayton copula with parameter $\theta / \rbr{1+|J|\theta}$ with $\theta = 1/\rho > 0$;
	\item if $\alpha_1 = \ldots = \alpha_d = 1$ and $-1 < \rho < 0$, then $c_{I;J}$ is the density of the $|I|$-variate Clayton survival copula with parameter $\theta / \rbr{|J|\theta-1}$ with $\theta = -1/\rho > 1$.
\end{compactitem}
In both cases, the dependence \emph{decreases} as the number $|J|$ of conditioning variables increases. 

\subsection{\HR{} model}\label{sec:hr}

Let $\mathcal{D}_{d}$ be the set of symmetric, strictly conditionally negative definite matrices, that is, all matrices $\Gamma = (\Gamma_{ij})_{i,j=1}^d \in \reals^{d \times d}$ of the form 
$\Gamma_{ij} = \var(A_i - A_j)$ for a $d$-variate random vector $\bm{A} = (A_1,\ldots,A_d)$ with positive definite covariance matrix, i.e., $\Gamma$ is a variogram matrix. 
Fix $k \in \dset$ and consider the $(d-1) \times (d-1)$-dimensional positive definite covariance matrix 
\begin{equation}
\label{eq:Sigk}
	\Sigk = \half \rbr{\Gamma_{ik} + \Gamma_{jk} - \Gamma_{ij}}_{i,j\ne k};
\end{equation}
for $\Gamma$ as in the previous sentence, we have $\Sigk_{ij} = \cov(A_i-A_k,A_j-A_k)$.
The $d$-variate \emph{\HR}\ tail copula density is
\begin{equation}
	\label{eq:HR}
	\xdf(\bx;\Sigk) =
	\rbr{\prod_{i:i\ne k} x_{i}^{-1} }
	\phi_{d-1} \rbr{\bar{\bx}_{\setminus{k}};\Sigk},
	\qquad \bx \in (0,\infty)^d,
\end{equation}
where $\phi_{d-1}$ is the centered $(d-1)$-dimensional Gaussian density with the stated covariance matrix and
$\bar{\bx}\mk = \rbr{\log(x_i/x_k) - \half\Gamma_{ik}}_{i:i\ne k}$.
The exponent measure density $\nupdf$ associated to $\xdf$ via \eqref{eq:nupdf} appears in \cite{engelke2015estimation}, where it is shown that it does not depend on the choice of $k \in \dset$. The corresponding max-stable distribution $G$ in \eqref{eq:max-stable} goes back to \citet{huesler1989maxima}, who studied maxima of triangular arrays of Gaussian random vectors.

For a matrix $M$ and for index sets $K$ and $L$, write $M_{KL} = (M_{ij})_{i \in K, j \in L}$. Let $\Phi$ denote the standard normal distribution function and recall that the $m$-variate Gaussian copula density with $(m \times m)$-dimensional positive definite correlation matrix $\bR$ evaluated in $\bu \in (0, 1)^m$ is
\begin{equation}
	\label{eq:copGausspdf}
	c_{\bR}(\bu) = \abs{\bR}^{-1/2} \cdot \exp \cbr{ - \half \bz^\top \rbr{ \bR^{-1} - I_p } \bz } 
	\quad \text{where} \quad z_i = \Phi^{-1}(u_i).
\end{equation}

\begin{proposition}
	\label{prop:hr}
	Let $\Gamma$ be a $d \times d$ variogram with $d \ge 3$ and let $\xdf$ be the \HR{} tail copula density in \eqref{eq:HR}. For every pair of disjoint, non-empty subsets $I, J \subset \dset$ with $|I| \ge 2$, the copula densities $c_{I;J}(\point;\bx_J)$ in \eqref{eq:cIJ} satisfy the simplifying assumption (Definition~\ref{def:simplif}). \chng{The copula density} $c_{I;J}$ is equal to the $|I|$-variate Gaussian copula density \eqref{eq:copGausspdf} with correlation matrix $\bR_{I|J} = (\Delta_{I|J}^{(k)})^{-1/2} \, \Sigk_{I|J} \, (\Delta_{I|J}^{(k)})^{-1/2}$ for any $k \in J$, where $\Delta_{I|J}^{(k)}$ is the diagonal matrix with the same diagonal as
$\Sigk_{I|J}
		= \Sigk_{II} - 
		\Sigk_{IJ} \rbr{\Sigk_{JJ}}^{-1} \Sigk_{JI}$.
\end{proposition}

\chng{The matrix $\bm{R}_{I|J}$ in Proposition~\ref{prop:hr} is the correlation matrix of the conditional Gaussian distribution with covariance matrix $\Sigk_{I|J}$, the corresponding Schur complement. Its} expression depends on the choice of $k \in J$, although the actual matrix does not. 
\chng{It remains to be investigated how to express it in terms of the \HR{} precision matrix $\Theta$ introduced in \citet{hentschel2022statistical}.}

\section{Vine decompositions of tail copula densities}
\label{sec:vine}

\subsection{\chng{X-vines: density decomposition and construction}}

In this section, we show that any tail copula density $\xdf$ in dimension $d \ge 3$ can be decomposed along any regular vine sequence (Definition~\ref{def:rvs}) on $d$ elements into $d-1$ bivariate tail copula densities and $\sum_{i=1}^{d-2} i = \binom{d-1}{2}$ bivariate copula densities. Moreover,  \chng{copula densities with only} a single conditioning variable satisfy the simplifying assumption (Definition~\ref{def:simplif}). Conversely, starting from any regular vine sequence on $d \ge 3$ elements, any collection of $d-1$ bivariate tail copula densities and any collection of $\binom{d-1}{2}$ bivariate copula densities, a $d$-variate tail copula density can be assembled. We coin tail copula densities constructed in this way as \emph{X-vines}.
\chng{Appendix~\ref{app:expmeas} in the supplement states the main results in terms of exponent measure densities.}

\begin{theorem}[Tail copula density decomposition along a regular vine]
\label{thm:xvine}
Let $\xdf$ be a $d$-variate ($d \ge 3$) tail copula density and let $\Vine = (\tree_j)_{j=1}^{d-1}$ with $\tree_j = \rbr{\nodes_j,\edges_j}$ be a regular vine sequence on $\dset$. For $\bx \in (0,\infty)^d$, we have
\begin{equation}
\label{eq:xdfdecomp}
	\xdf(\bx) = \prod_{e \in \edges_1} \xdf_{a_e,b_e}(x_{a_e},x_{b_e}) \cdot
	\prod_{j=2}^{d-1} \prod_{e \in \edges_j}
	c_{a_e,b_e;\cndg_e} \rbr{ 
		\xcdf_{a_e|\cndg_e}(x_{a_e}|\bx_{\cndg_e}), \xcdf_{b_e|\cndg_e}(x_{b_e}|\bx_{\cndg_e});
		\bx_{\cndg_e}
	},
\end{equation}
where, for $e = (a_e,b_e;\cndg_e) \in E_2 \cup \cdots \cup E_{d-1}$, the pair-copula density $c_{a_e,b_e;\cndg_e}$ is
\begin{equation}
\label{eq:caebeDex}	
	c_{a_e,b_e;\cndg_e}\rbr{ u_{a_e}, u_{b_e} ; \bx_{\cndg_e} }
	= \frac{\xdf_{a_e,b_e|\cndg_e}(x_{a_e},x_{b_e}|\bx_{\cndg_e})}{\xdf_{a_e|\cndg_e}(x_{a_e}|\bx_{\cndg_e}) \cdot \xdf_{b_e|\cndg_e}(x_{b_e}|\bx_{\cndg_e})},
\end{equation}
with $x_{a_e} = \xcdf_{a_e|\cndg_e}^{-1} (u_{a_e}| \bx_{\cndg_e})$ and $x_{b_e} = \xcdf_{b_e|\cndg_e}^{-1} (u_{b_e}| \bx_{\cndg_e})$ for $u_{a_e},u_{b_e} \in (0,1)$.
\end{theorem}

\chng{The decomposition~\eqref{eq:xdfdecomp} also applies to marginal tail copula densities $\xdf_J$ for index sets $J = \cunn_f$ for some edge $f$ in the vine; see Remark~\ref{rem:xdfdecomp-J} in the supplement.}
For the pair copula $C_{a_e, b_e; \cndg_e}(\point,\point; \bx_{\cndg_e})$ associated to the edge $e = (a_e, b_e; \cndg_e)$ in the regular vine sequence, consider the first-order partial derivatives
\begin{equation}
\label{eq:Ccond}
\begin{split}
	C_{a_e|b_e; \cndg_e}(u_{a_e} \mid u_{b_e}; \bx_{\cndg_e})
	&:=
	\frac{\partial}{\partial u_{b_e}} C_{a_e, b_e; \cndg_e}(u_{a_e}, u_{b_e}; \bx_{\cndg_e})
	=
	\int_{v=0}^{u_{a_{e}}}
	c_{a_e, b_e; \cndg_e}(v, u_{b_e}; \bx_{\cndg_e})
	\, \diff v,
	\\
	C_{b_e|a_e; \cndg_e}(u_{b_e} \mid u_{a_e}; \bx_{\cndg_e})
	&:=
	\frac{\partial}{\partial u_{a_e}} C_{a_e, b_e; \cndg_e}(u_{a_e}, u_{b_e}; \bx_{\cndg_e}) 
	=
	\int_{v=0}^{u_{b_{e}}}
		c_{a_e, b_e; \cndg_e}(u_{a_e}, v; \bx_{\cndg_e})
	\, \diff v.
\end{split}
\end{equation}

\begin{theorem}[Recursion and uniqueness]
	\label{thm:recuni}
	In the setting of Theorem~\ref{thm:xvine}, for any $e = (a_e,b_e;\cndg_e) \in \edges_2 \cup \cdots \cup \edges_{d-1}$, we have
	\begin{equation}
	\label{eq:xcdfrecur}
	\begin{split}
		\xcdf_{a_e|\cndg_e \cup b_e}(x_{a_e}| \bx_{\cndg_e \cup b_e})
		&=
		C_{a_e|b_e;\cndg_e} \rbr{
			\xcdf_{a_e|\cndg_e}(x_{a_e}|\bx_{\cndg_e})
			\mid
			\xcdf_{b_e|\cndg_e}(x_{b_e}|\bx_{\cndg_e});
			\bx_{\cndg_e}
		}, \\
		\xcdf_{b_e|\cndg_e \cup a_e}(x_{b_e}| \bx_{\cndg_e \cup a_e})
		&=
		C_{b_e|a_e;\cndg_e} \rbr{
			\xcdf_{b_e|\cndg_e}(x_{b_e}|\bx_{\cndg_e}) 
			\mid
			\xcdf_{a_e|\cndg_e}(x_{a_e}|\bx_{\cndg_e});
			\bx_{\cndg_e}
		}.
	\end{split}
	\end{equation}
	As a consequence, $\xdf$ is determined uniquely by the bivariate tail copula densities $\xdf_{a,b}$ for $e = \cbr{a,b} \in \edges_1$ and the bivariate copula densities $c_{a_e,b_e;\cndg_e}(\point;\bx_{\cndg_e})$ for $e \in \edges_2 \cup \cdots \cup \edges_{d-1}$ and $\bx_{\cndg_e} \in (0,\infty)^{\cndg_e}$.
\end{theorem}

Equation~\eqref{eq:xcdfrecur} is effectively a recursive relation, allowing to reduce the number of conditioning variables until there is only one conditioning variable left. 
The reason is that each of the indices $a_e|D_e$ and $b_e|D_e$ in \eqref{eq:xcdfrecur} is itself of the form $a_f|D_f \cup b_f$ or $b_f|D_f \cup a_f$ for an edge $f \in e$ in $\edges_{j-1}$, i.e., one level lower than $e$.

\begin{definition}[X-vine tail copula density]
\label{def:xvine}
	A $d$-variate tail copula density $\xdf$ is an \emph{X-vine} along a regular vine sequence $\Vine$ if for each edge $e \in \edges_2 \cup \cdots \cup \edges_{d-1}$, the pair copula densities $c_{a_e,b_e;\cndg_e}(\point,\point;\bx_{\cndg_e})$ do not depend on the value of $\bx_{\cndg_e} \in (0, \infty)^{\cndg_e}$.
\end{definition}

\begin{example}[Trivariate case]\label{ex:tri}
By Lemma~\ref{lem:homo}, a trivariate tail copula density $\xdf$ is always an X-vine, and this along any of the three possible regular vine sequences on $\cbr{1,2,3}$. \chng{For the vine determined by $\edges_1 = \cbr{\cbr{1,2},\cbr{2,3}}$, for instance, we have}
\begin{align}
\label{eq:decomp3}
	\xdf(x_1,x_2,x_3) 
	&= \xdf_{12}(x_1,x_2) \, \xdf_{23}(x_2,x_3) \cdot 
	c_{13;2} \rbr{\xcdf_{1|2}(x_1|x_2), \, \xcdf_{3|2}(x_3|x_2)}, \\
	\nonumber
	\text{where} \quad 
	c_{13;2}(u_1, u_3) &=
	\frac{\xdf(x_1,1,x_3)}{\xdf_{12}(x_1,1) \, \xdf_{23}(1,x_3)}
	\quad \text{with} \quad
	x_j = \xcdf_{j|2}^{-1}(u_j|1), \quad j \in \cbr{1,3}.
\end{align}
The function $\xdf$ is thus completely specified by the two bivariate tail copula densities $\xdf_{12}$ and $\xdf_{23}$ and one bivariate copula density $c_{13;2}$. \chng{The form~\eqref{eq:decomp3} was already discovered for tail copula densities of D-vine copulas in \cite{li2013extremal}.}
\end{example}

By Propositions~\ref{prop:scdiri} and~\ref{prop:hr}, the scaled extremal Dirichlet model (including the logistic, negative logistic and extremal Dirichlet models) and the Hüsler--Reiss family have conditional copula densities $c_{I;J}(\point;\bx_J)$ that always satisfy the simplifying assumption. As a consequence, they are examples of X-vine tail copula densities too, and this along any regular vine sequence.

\begin{definition}[X-vine specification]
\label{def:xvinespec}
The triplet $\rbr{\Vine, \xdfset, \cpdfset}$ is an \emph{X-vine specification} on $d$ elements ($d \ge 3$) if:
\begin{compactenum}[1.]
\item $\Vine = (\tree_{j})_{j=1}^{d-1}$ with $\tree_j = (\nodes_j, \edges_j)$ is a regular vine sequence on $\dset$;
\item $\xdfset = \cbr{\xdf_{a_e,b_e} : \; e = \cbr{a_e,b_e} \in \edges_1}$ is a family of bivariate tail copula densities;
\item $\cpdfset = \cbr{c_{a_e,b_e;\cndg_e} : \; e = (a_e,b_e;\cndg_e) \in \bigcup_{j \ge 2} \edges_j}$ is a family of bivariate copula densities.
\end{compactenum}
\end{definition}

Fig.~\ref{fig:5d_Xvine_a} in Section~\ref{sec:simulation-study} shows an example of an X-vine specification involving a mix of parametric models for the bivariate (tail) copula densities.

\begin{theorem}[X-vine tail copula density construction]
\label{thm:xvineconstruct}
Let $\rbr{\Vine, \xdfset, \cpdfset}$ be an X-vine specification on $d \ge 3$ elements.
Then the function $\xdf$ defined by 
\begin{equation}
\label{eq:xvine}
	\xdf(\bx) = \prod_{e \in \edges_1} \xdf_{a_e,b_e}(x_{a_e},x_{b_e}) \cdot
\prod_{j=2}^{d-1} \prod_{e \in \edges_j}
c_{a_e,b_e;\cndg_e} \rbr{ 
	\xcdf_{a_e|\cndg_e}(x_{a_e}|\bx_{\cndg_e}), \xcdf_{b_e|\cndg_e}(x_{b_e}|\bx_{\cndg_e})
}
\end{equation}
with the functions $\xcdf_{\point|\point}$ defined recursively by
\begin{equation} 
\label{eq:xvine:recur}
\left.
\begin{aligned}
		\xcdf_{a_e|\cndg_e \cup b_e}(x_{a_e}| \bx_{\cndg_e \cup b_e})
		&=
		C_{a_e|b_e;\cndg_e} \rbr{
			\xcdf_{a_e|\cndg_e}(x_{a_e}|\bx_{\cndg_e})
			\mid
			\xcdf_{b_e|\cndg_e}(x_{b_e}|\bx_{\cndg_e})
		} \\
		\xcdf_{b_e|\cndg_e \cup a_e}(x_{b_e}| \bx_{\cndg_e \cup a_e})
		&=
		C_{b_e|a_e;\cndg_e} \rbr{
			\xcdf_{b_e|\cndg_e}(x_{b_e}|\bx_{\cndg_e}) 
			\mid
			\xcdf_{a_e|\cndg_e}(x_{a_e}|\bx_{\cndg_e})
		}
\end{aligned}
\right\}
\end{equation}
is a $d$-variate tail copula density. \chng{For $e \in \edges_1$, the bivariate margin of $\xdf$ is equal to $\xdf_e \in \xdfset$, while for $e = (a_e,b_e;\cndg_e) \in \edges_2 \cup \cdots \cup \edges_{d-1}$, the pair copula density $c_{a_e,b_e;\cndg_e}(\point,\point;\bx_{\cndg_e})$ in \eqref{eq:caebeDex} is equal to $c_{a_e,b_e;\cndg_e} \in \cpdfset$.} In particular, $\xdf$ is an X-vine.
\end{theorem}

\chng{In Example~\ref{ex:tri} with $d = 3$, the vine $\Vine = (\tree_1,\tree_2)$ is determined by $\edges_1 = \cbr{\cbr{1,2},\cbr{2,3}}$, while $\xdfset = \cbr{r_{12},r_{23}}$ and $\cpdfset = \cbr{c_{13;2}}$, and \eqref{eq:xvine} reduces to \eqref{eq:decomp3}.}

\subsection{\chng{X-vines as limits of regular vine copula densities}}

\chng{A natural question is whether X-vine tail copula densities arise as the lower tail dependence limits of regular vine copula densities, as introduced in \cite{bedford2002vines}. Below, we show that this is indeed the case, provided the pair copula densities at the edges of the first tree have the corresponding bivariate tail copula densities as lower tail dependence limits. In the passage to the limit, the regular vine sequence is preserved and so are the pair copulas at all trees starting from the second one.} 

\chng{Let $\Vine = \rbr{\tree_j}_{j=1}^{d-1}$ with $\tree_j = (\nodes_j, \edges_j)$ be a regular vine sequence on $d$ elements, for $d \ge 3$. For every edge $e = (a_e,b_e;\cndg_e) \in \bigcup_{j=1}^{d-1} \edges_j$, let $c_{a_e,b_e;\cndg_e}$ be a bivariate copula density, with copula $C_{a_e,b_e;\cndg_e}$ and conditional distribution functions $C_{a_e|b_e;\cndg_e}$ and $C_{b_e|a_e;\cndg_e}$. For $e \in \edges_1$, the conditioning set $\cndg_e$ is empty and we simply write $c_{a_e,b_e}$ and so on. Then there is a $d$-variate regular vine copula density $c$ given by
\begin{equation}
\label{eq:PCC}
	c(\bu)
	= \prod_{e \in \edges_1} c_{a_e,b_e}(u_{a_e},u_{b_e})
	\cdot
	\prod_{j=2}^{d-1} \prod_{e \in \edges_j}
	c_{a_e,b_e;\cndg_e} \rbr{
		C_{a_e|\cndg_e} \rbr{u_{a_e}|\bu_{\cndg_e}}, \,
		C_{b_e|\cndg_e} \rbr{u_{b_e}|\bu_{\cndg_e}}
	},
\end{equation}
where, for a random vector $\bU$ with density $c$, the conditional distribution function of $(U_{a_e},U_{b_e})$ given $\bU_{\cndg_e} = \bu_{\cndg_e}$ is
\begin{equation}
\label{eq:PCC:recur}
C_{a_e,b_e|\cndg_e}\rbr{u_{a_e},u_{b_e}|\bu_{\cndg_e}}
	=
	C_{a_e,b_e;\cndg_e} \rbr{
		C_{a_e|\cndg_e} \rbr{u_{a_e}|\bu_{\cndg_e}}, \,
		C_{b_e|\cndg_e} \rbr{u_{b_e}|\bu_{\cndg_e}}
	}.
\end{equation}
\begin{assumption}
	\label{ass:PCClambdae}
	Let $c$ be a $d$-variate regular vine copula density as in~\eqref{eq:PCC}, with $d \ge 3$.
	\begin{compactenum}[(i)]
	\item For every edge $e \in \edges_1$, there exists a bivariate tail copula density $r_{a_e,b_e}$ such that
$\lim_{t \searrow 0} t \cdot c_{a_e,b_e}(tx, ty) 
	= \xdf_{a_e,b_e}(x, y)$, for $(x, y) \in (0, \infty)^2$.
	\item
	For every edge $e \in \bigcup_{j=2}^{d-1} \edges_j$, the pair copula density $c_{a_e,b_e;\cndg_e}$ is continuous.
	\end{compactenum}
\end{assumption}

\begin{proposition}
	\label{prop:c2l}
	Under Assumption~\ref{ass:PCClambdae}, we have $\lim_{t \searrow 0} t^{d-1} \cdot c(t\bx)
	= \xdf(\bx)$ for $\bx \in (0, \infty)^d$,
	where the X-vine tail copula density $\xdf$ is generated by the triple $(\Vine, \xdfset, \cpdfset)$ as in Theorem~\ref{thm:xvineconstruct} with 
	\begin{compactitem}
	\item $\Vine$ the same regular vine sequence as in \eqref{eq:PCC},
	\item $\xdfset = \cbr{ r_{a_e,b_e} : e \in \edges_1 }$ for $r_{a_e,b_e}$ in Assumption~\ref{ass:PCClambdae}(i), and
	\item $\cpdfset$ the same bivariate copula densities for edges $e \in \edges_j$ with $j \ge 2$ as in \eqref{eq:PCC}.
	\end{compactitem}
	In particular, the copula $C$ of $c$ has tail copula $\xcdf$ with tail copula density $\xdf$.
\end{proposition}

Proposition~\ref{prop:c2l} is foreshadowed by Theorem~3.4 and Example~3.5 in \cite{li2013extremal}, who consider D-vine copulas and who require stronger convergence properties.}

\subsection{\chng{Truncated X-vines}}

If the bivariate copula density $c_{a_e,b_e;\cndg_e}$ associated with an edge $e \in E_2 \cup \cdots \cup E_{d-1}$ is equal to the independence one, $c_{a_e,b_e;\cndg_e} \equiv 1$, the corresponding factor drops out in \eqref{eq:xdfdecomp} and the recursive formulas in \eqref{eq:xcdfrecur} and~\eqref{eq:xvine:recur} simplify to
$\xcdf_{a_e|\cndg_e \cup b_e} \rbr{x_{a_e}|\bx_{\cndg_e \cup b_e}} 
	= \xcdf_{a_e|\cndg_e} \rbr{x_{a_e}|\bx_{\cndg_e}}$ and 
	$\xcdf_{b_e|\cndg_e \cup a_e} \rbr{x_{b_e}|\bx_{\cndg_e \cup a_e}} 
	= \xcdf_{b_e|\cndg_e} \rbr{x_{b_e}|\bx_{\cndg_e}}$.
\emph{Sparse} X-vine specifications arise when $c_{a_e,b_e;\cndg_e} \equiv 1$ for many edges $e \in \edges_2 \cup \cdots \cup \edges_{d-1}$. Model selection of sparse vine copulas in high dimensions has been investigated in \citet{muller2019selection} and \citet{nagler2019model}. The case where all pair copulas are equal to the independence one for all edges starting from a given tree is of practical importance \citep{brechmann2012truncated,brechmann2015truncation}.

\begin{definition}
\label{def:truncvine}
A \emph{truncated regular vine tree sequence} $\Vine = (\tree_1,\ldots,\tree_q)$ on $d \ge 3$ elements with truncation level $q \in \cbr{1,\ldots,d-2}$ is an ordered set of trees for which there exists a regular vine tree sequence $\Vine' = (\tree_1',\ldots,\tree_{d-1}')$ on $d$ elements such that $\tree_j = \tree_j'$ for all $j \in \cbr{1,\ldots,q}$.
\end{definition}

\begin{definition}
\label{def:truncXvine}
The triplet $\rbr{\Vine, \xdfset, \cpdfset}$ is a \emph{truncated X-vine specification} on $d$ elements ($d \ge 3$) with truncation level $q \in \cbr{1,\ldots,d-2}$ if:
\begin{compactenum}[1.]
\item $\Vine = (\tree_{j})_{j=1}^{q}$ is a truncated regular vine tree sequence on $\dset$;
\item $\xdfset = \cbr{\xdf_{a,b} : \; e = \cbr{a,b} \in \edges_1}$ is a family of bivariate tail copula densities;
\item $\cpdfset = \{c_{a_e,b_e;\cndg_e} : \; e = (a_e,b_e;\cndg_e) \in \bigcup_{j = 2}^q \edges_j\}$ is a family of bivariate copula densities; in case $q = 1$, we have $\cpdfset = \emptyset$.
\end{compactenum}
\end{definition}

By definition, any truncated X-vine specification can be completed to a full X-vine specification by completing the truncated regular vine tree sequence $\Vine$ to a full one $\Vine'$ as in Definition~\ref{def:truncvine} and by setting $c_{a_e,b_e;\cndg_e} \equiv 1$ for all $e \in \edges_{q+1}' \cup \cdots \cup \edges_{d-1}'$. The resulting X-vine copula density $\xdf$ does not depend on the way in which the truncated regular vine tree sequence $\Vine$ is completed, since the factorisation in \eqref{eq:xdfdecomp} simplifies anyway to
\begin{equation}
\label{eq:xdftrunc}
	\xdf(\bx) = \prod_{e \in \edges_1} \xdf_{a_e,b_e}(x_{a_e},x_{b_e}) \cdot
	\prod_{j=2}^{q} \prod_{e \in \edges_j}
	c_{a_e,b_e;\cndg_e} \rbr{ 
		\xcdf_{a_e|\cndg_e}(x_{a_e}|\bx_{\cndg_e}), \xcdf_{b_e|\cndg_e}(x_{b_e}|\bx_{\cndg_e})
	},
\end{equation}
\chng{The right-hand side of \eqref{eq:xdftrunc} involves $\sum_{j=2}^q (d-j)$ bivariate copula densities.}
The truncation level $q$ allows to tune the trade-off between sparsity and flexibility. 
\chng{If $q = 1$, the second product is empty and the model is a Markov tree as in \cite{engelke2020graphical}, \cite{segers2020one} and \cite{engelke2022structure}}. Increasing $q$ and adding trees (`layers') yields more complex dependence models. Truncated X-vine specifications will be shown at work in Section~\ref{sec:casestudy}.

\section{Sampling from X-vine Pareto distributions}
\label{sec:simulation}

\subsection{Inverted multivariate Pareto distributions}

In Section~\ref{sec:background}, we introduced the multivariate Pareto distribution~\eqref{eq:mgpd} as a limit model for high threshold excesses in \eqref{eq:Vt2Y}. In the context of tail copula measures, it is more convenient to work with their reciprocals.

\begin{definition}
\label{def:invmultpar}
The $d$-variate random vector $\IMP$ has an \emph{inverted multivariate Pareto distribution} if there exists a $d$-variate tail copula measure $\xcdf$ such that
$\prob(\bZ \in B) 
	= \xcdf(B \cap \LL)/\xcdf(\LL)$,
for Borel sets $B$, with $\LL = \cbr{\bx \in (0,\infty]^d : \min \bx < 1}$; equivalently, if there exists a multivariate Pareto random vector $\bY$ such that $\IMP = 1/\bY$.
\end{definition}

If $\xcdf$ is concentrated on $(0, \infty)^d$ and has density $\xdf$, then $\IMP$ has probability density $\bx \mapsto \xdf(\bx) \1(\bx \in \LL) / \xcdf(\LL)$.
In general, for $j \in \dset$, the conditional distribution of $\Imp_j$ given $\Imp_j < 1$ is uniform on $(0, 1)$; this follows from the marginal constraint~\eqref{eq:margin}. For Borel sets $B$, we have
\begin{equation}
\label{eq:Zj}
	\prob(\bZ \in B \mid Z_j < 1) 
	= \xcdf \rbr{B \cap \LL^{(j)}}, 
\end{equation}
where the set $\LL^{(j)} = \cbr{\bx \in (0, \infty]^d : x_j < 1}$ has $\xcdf$-measure one.
Hence, on $\LL^{(j)}$, the probability density $\bx \mapsto \xdf(\bx) \1 \{x_j < 1\}$ of $\rbr{\IMP \mid \IMPm_j < 1}$ coincides with the tail copula density $\xdf$.

For the random vector $\bU$ in Definition~\ref{def:expmeas}, the inverted multivariate Pareto vector $\IMP$ is the weak limit in
\begin{equation}
\label{eq:Ut2Z}
	\rbr{t \bU \mid \min \bU < 1/t}
	\dto \IMP, \qquad t \to \infty.
\end{equation}
For non-empty $J \subseteq \dset$, we have
$\rbr{t \bU_J \mid \min \bU_J < 1/t}
	\dto \IMP_{|J} 
	\eqd \rbr{\IMP_J \mid \min \IMP_J < 1}$, as $t \rightarrow \infty$.
The conditional marginal $\IMP_{|J}$ has a $|J|$-variate inverted multivariate Pareto distribution associated with the marginal tail copula measure $\xcdf_{J}$ in \eqref{eq:RJ}.
In contrast, $\IMP_J$ does \emph{not} necessarily have an inverted multivariate Pareto distribution, as $\min \IMP_J < 1$ is not guaranteed.

\subsection{\chng{Sampling from (inverted) multivariate Pareto distributions}}

Through equations~\eqref{eq:Vt2Y} and~\eqref{eq:Ut2Z}, the (inverted) multivariate Pareto distribution serves as a model for a random vector conditionally on the event that at least one variable takes a value far in the tail of its respective marginal distribution. 
The L-shaped support of the (inverted) multivariate Pareto distribution makes direct random sampling from it a little awkward. 
Lemma~2 in \citet{engelke2020graphical} provides an ingenious algorithm that reduces the task of sampling a multivariate Pareto random vector $\bY$ to sampling the conditional distributions $\rbr{ \bY \mid Y_j > 1 }$ for every $j \in \dset$. Below, we study the equivalent problem of simulating from the conditional distribution of $\rbr{ \IMP \mid \Imp_j < 1 }$ for X-vine tail copula densities.

We will do so by inverting the Rosenblatt transformation \citep{rosenblatt1952remarks}, applying conditional quantile functions with an increasing number of conditioning variables successively to independent uniform random variables.
A judicious choice of the ordering of the variables permits to compute the required conditional quantile functions recursively in terms of the bivariate ingredients of the X-vine specification. This \chng{\emph{sampling order} \citep{cooke2015sampling} is encoded by the permutation constructed in the next lemma.}

\begin{lemma}
	\label{lem:perm}
	Let $\Vine$ be a regular vine sequence on $d$ elements. For all $j \in \dset$, there exists a permutation $\sigma_j$ of $\cbr{1,\ldots,d}$ such that (i) $\sigma_j(1) = j$ and (ii) there exist edges $e_{j,1} \in \edges_{1}, \ldots, e_{j,d-1} \in \edges_{d-1}$ such that $\sigma_j(k) \in \cndd_{e_{k-1}}$ and $\cbr{\sigma_j(i): i = 1,\ldots,k} = \cunn_{e_{j,k-1}}$ for all $k \in \cbr{2,\ldots,d}$.
\end{lemma}

Let $\xdf$ be an X-vine tail copula density as in Definition~\ref{def:xvine}. For each edge $e \in \edges_2 \cup \ldots \cup \edges_{d-1}$ and for fixed $0 < u_{b_e} < 1$, let $u_{a_e} \mapsto C_{a_e|b_e;\cndg_e}^{-1} \rbr{u_{a_e}|u_{b_e}}$ be the inverse of the distribution function $u_{a_e} \mapsto C_{a_e|b_e;\cndg_e} \rbr{u_{a_e}|u_{b_e}}$, with $C_{a_e|b_e;\cndg_e}$ defined in \eqref{eq:Ccond}. Similarly for $C_{b_e|a_e;\cndg_e}^{-1}$. Recall from Proposition~\ref{lem:expdensity} that for $i \in \dset$, non-empty $J \subseteq \dset \setminus \cbr{i}$, and $\bx_J \in (0, \infty)^d$ such that $\xdf_J(\bx_J) > 0$, the quantile function $u_i \mapsto \xqdf_{i|J}(u_i| \bx_J)$ is the inverse of the distribution function $x_i \mapsto \xcdf_{i|J}(x_i|\bx_J)$.
Inverting \eqref{eq:xcdfrecur} yields the recursive relations
\begin{equation}
\label{eq:xqdfrecur}
\begin{split}
	\xqdf_{a_e|\cndg_e \cup b_e}(u_{a_e}|\bx_{\cndg_e \cup b_e})
	&= \xqdf_{a_e|\cndg_e} \rbr{ C_{a_e|b_e;\cndg_e}^{-1} \rbr{u_{a_e} \mid \xcdf_{b_e|\cndg_e}(x_{b_e}| \bx_{D_e})} \mid \bx_{\cndg_e} }, \\
	\xqdf_{b_e|\cndg_e \cup a_e}(u_{b_e}|\bx_{\cndg_e \cup a_e})
	&= \xqdf_{b_e|\cndg_e} \rbr{ C_{b_e|a_e;\cndg_e}^{-1} \rbr{u_{b_e} \mid \xcdf_{a_e|\cndg_e}(x_{a_e}| \bx_{D_e})} \mid \bx_{\cndg_e} }.
\end{split}
\end{equation}

Let $W_1,\ldots,W_d$ be independent random variables, all uniformly distributed on $(0, 1)$. For $j \in \dset$, let $\sigma_j$ be a permutation of $\dset$ satisfying the two requirements in Lemma~\ref{lem:perm}. Define a random vector $\IMP^{(j)} = (\Imp_1^{(j)},\ldots,\Imp_d^{(j)})$ recursively as follows:
\begin{equation}
\label{eq:IMPj}
	\Imp_{j}^{(j)} = W_{j} ~\text{and}~
	\Imp_{\sigma_j(k)}^{(j)} = \xqdf_{\sigma_j(k)|\sigma_j(\cbr{1,\ldots,k-1})} \rbr{ W_{\sigma_j(k)} \mid \IMP_{\sigma_j(\cbr{1,\ldots,k-1})}^{(j)}}, 
	\quad k \in \cbr{2,\ldots,d}.
\end{equation}

\begin{proposition}
\label{prop:simu}
	Let $\IMP$ be an inverted multivariate Pareto random vector associated with the X-vine tail copula density $\xdf$. For $j \in \dset$, the distribution of $\IMP^{(j)}$ in \eqref{eq:IMPj} is equal to the one of $\IMP$ conditionally on $\Imp_j < 1$. For every $k \in \cbr{2,\ldots,d}$, the conditional quantile function $\xqdf_{\sigma_j(k)|\sigma_j(\cbr{1,\ldots,k-1})}$ is of the form $\xqdf_{a_e|\cndg_e \cup b_e}$ or $\xqdf_{b_e|\cndg_e \cup a_e}$ for some edge $e = e_{j,k-1} \in \edges_{k-1}$ and can thus be computed recursively via \eqref{eq:xqdfrecur}.
\end{proposition}

\chng{The simulation algorithm based on Proposition~\ref{prop:simu} relies on structure matrices encoding regular vine sequences as explained in Appendix~\ref{app:example} in the supplement.
In Section~\ref{sec:simulation-study}, we apply the algorithm to assess the estimation methods from Section~\ref{sec:estimation} through Monte Carlo experiments.}

\section{Estimation and model selection}
\label{sec:estimation}

Let $\bX_i = (X_{i,1},\ldots,X_{i,d})$ for $i \in \cbr{1,\ldots,n}$ be an independent random sample from a distribution function $F$ with continuous but unspecified margins $F_1,\ldots,F_d$ and whose survival copula $C$ has lower tail copula $\xcdf$ (Definition~\ref{def:expmeas}). Suppose that the tail copula density $\xdf$ is an X-vine with specification $\rbr{\Vine, \xdfset, \cpdfset}$ (Definitions~\ref{def:xvine} and~\ref{def:xvinespec} and Theorem~\ref{thm:xvineconstruct}). We propose a procedure to estimate $\xdf$ from the excesses over a high multivariate threshold.

The regular vine sequence $\Vine = (\tree_j)_{j=1}^{d-1}$ with trees $\tree_j = (\nodes_j, \edges_j)$ may be known or not. The bivariate tail copulas $r_{a,b}$ for edges $e = \cbr{a,b} \in \edges_1$ in the first tree and the bivariate copula densities $c_{a_e,b_e;\cndg_e}$ for edges $e = (a_e,b_e;\cndg_e) \in \edges_j$ in trees $j \ge 2$ are assumed to belong to prespecified (lists of) parametric families. The X-vine specification may be truncated (Definition~\ref{def:truncXvine}), leading to a simpler model. 

The basis of the method is a link between the conditional copula densities $c_{I;J}$ in Sklar's theorem (Proposition~\ref{lem:expdensity}) on the one hand and the inverted multivariate Pareto distribution (Definition~\ref{def:invmultpar}) on the other hand (Section~\ref{sec:SklarIMP}). In  Section~\ref{sec:estimation1}, we propose parameter estimates, supposing that $\Vine$ is given and that parametric families of (tail) copula densities have been specified for all edges. In Section~\ref{sec:estimation2}, finally, we treat model selection, which comprises the selection of the parametric families of the bivariate model components, the selection of the regular vine sequence $\Vine$, and the selection of the truncation level $q$.

\subsection{Copulas and inverted multivariate Pareto distributions}
\label{sec:SklarIMP}

Let $\xdf$ be a $d$-variate \chng{tail copula} density, not necessarily an X-vine. In Sklar's theorem (Proposition~\ref{lem:expdensity}), suppose that the copula density $c_{I;J}$ satisfies the simplifying assumption (Definition~\ref{def:simplif}). The following proposition shows how to transform an inverted multivariate Pareto random vector $\bZ$ associated to $\xdf$ into a random vector with density $c_{I;J}$.

\begin{proposition}
	\label{prop:Z2c}
	Let $\xdf$ be a $d$-variate tail copula density and let $\IMP$ be an inverted multivariate Pareto random vector associated to $\xdf$. Let $I, J \subset \dset$ be non-empty and disjoint.
	\begin{compactenum}[(i)]
	\item 
	For $\bz_J \in (0, \infty)^J$ such that $\min \bz_J < 1$ and $\xdf_J(\bz_J) > 0$, the conditional density of $\IMP_I$ given $\IMP_J = \bz_J$ is $\xdf_{I|J}(\point|\bz_J)$.
	\item 
	Suppose $|I| \ge 2$. If $c_{I;J}(\point;\point)$ satisfies the simplifying assumption (Definition~\ref{def:simplif}), then, conditionally on the event $\min \IMP_J < 1$, the random vector  $\rbr{\xcdf_{i|J}(\Imp_i|\IMP_J)}_{i \in I}$ is independent of $\IMP_J$ and its density is $c_{I;J}$.
	\end{compactenum}
\end{proposition}

By statement~(ii), the density of $\rbr{\xcdf_{i|J}(\Imp_i|\IMP_J)}_{i \in I}$ given $\IMP_J \in A$ is equal to $c_{I;J}$ for any non-empty set $A \subseteq \cbr{\bz \in (0, \infty)^J : \min \bz < 1}$. For estimation, we will use this property for $A = (0, 1)^J$, requiring in effect that \emph{all} variables (rather than \emph{at least one}) in $J$ exceed a high threshold. We do so in order to avoid a potential bias stemming from including too many non-extreme values in the procedure. \chng{An alternative would be to opt for a censored likelihood approach \citep{ledford1996statistics}.}

\subsection{Sequential maximum likelihood estimation of X-vines}
\label{sec:estimation1}

Let the $d$-variate tail copula density $\xdf$ be an X-vine specified by $(\Vine, \xdfset, \cpdfset)$ as in Theorem~\ref{thm:xvineconstruct}. Assume that the bivariate tail copula densities ($\xdfset$) and bivariate copula densities ($\cpdfset$) belong to specified parametric families.
Let the parameter vector be denoted by $\bm{\theta} = (\theta_{\xdfset}, \theta_{\cpdfset})$: here, $\theta_{\xdfset} = (\theta_e)_{e \in \edges_1}$ contains the parameters (or parameter vectors) $\theta_e \in \Theta_e$ associated with each pairwise tail copula density $\xdf_{a_e,b_e}$ for $e \in \edges_1$, while $\theta_{\cpdfset} = \rbr{ \theta_{\cpdfset,j}}_{j=2}^{d-1}$ for $ \theta_{\cpdfset,j} = (\theta_e)_{e \in \edges_j}$ denotes the parameters (or parameter vectors) $\theta_e \in \Theta_e$ associated with each bivariate copula density $c_{a_e,b_e ; D_e}$ for $e \in \bigcup_{j=2}^{d-1} \edges_j$. 
While it is possible to derive the \chng{full likelihood} of an X-vine (inverted) multivariate Pareto distribution, performing parameter estimation with the full model in high dimensions \chng{is challenging}.
Instead, using the X-vine decomposition into \emph{bivariate} components and recursively defined quantities (Theorem~\ref{thm:xvineconstruct}), we outline a \emph{sequential} procedure for parameter estimation,
\chng{tree by tree}. 
This approach is inspired by the one for regular vine copulas \citep[see, e.g.,][]{czado2019analyzing}, but with suitable adaptations to the extreme value context.

\paragraph{(1) Standardising the margins and selecting sub-samples.} 
\chng{Recall 
that $\bX_1,\ldots,\bX_n$
is an independent random sample from 
$F$, with survival copula $C$, tail copula $\xcdf$ and tail copula density $\xdf$}. 
For $i = 1,\ldots,n$ and $j \in \dset$, let	$\hU_{i,j} = 1 - \hF_j(X_{i,j})$, 
where $\hF_j$ denotes any estimator of the marginal distribution function $F_j(x) = \prob(X_{i,j} \le x)$. One possibility is the empirical distribution function, and to avoid boundary effects, we set 
\begin{equation}\label{eq:margins:ranks}
    \hU_{i,j} = 1 - (\operatorname{rnk}_{i,j} - 0.5)/n,
\end{equation}
where $\operatorname{rnk}_{i,j} = \sum_{s=1}^n \1(X_{s,j} \le X_{i,j})$ is the (maximal) rank of $X_{i,j}$ among $X_{1,j},\ldots,X_{n,j}$. We view the points $\hbU_i = (\hU_{i,1},\ldots,\hU_{i,d})$ as pseudo-observations from the survival copula $C$. 

By Eq.~\eqref{eq:Ut2Z}, for large $t > 0$, the rescaled points $t \hbU_i$ for $i \in \cbr{1,\ldots,n}$ such that $\min \hbU_i < 1/t$ constitute pseudo-observations from a distribution that approximates the inverted multivariate Pareto distribution associated with $\xdf$. We set $t = n/k$ where\chng{, in an asymptotic setting, $k = k_n \in \cbr{1,\ldots,n}$ 
satisfies} $k \to \infty$ and $k/n \to 0$. For $j \in \dset$ \chng{and $\hU_{i,j}$ as in \eqref{eq:margins:ranks}}, let 
\begin{equation}
\label{eq:Nj}
	\iset_j = \cbr{i = 1,\ldots,n : \hU_{i,j} < k/n}
\end{equation}
be the set of indices $i$ corresponding \chng{to the $k$ largest observations for the $j$th component.
Further, put $\iset_J = \bigcap_{j \in J} \iset_j$ and $\iset = \bigcup_{j \in \dset} \iset_j,$
which are, respectively, the set of 
indices $i$ corresponding to large observations in
\emph{all} variables in $J \subseteq \dset$ \emph{simultaneously} and the set of 
indices with large observations in \emph{at least one} variable.
Write $\hbZ_i = (\hZ_{i,j})_{j \in \dset}$ where $\hZ_{i,j} = \rbr{n/k} \hU_{i,j}$. In view of Eq.~\eqref{eq:Ut2Z}, we treat $\{\hbZ_i\}_{i \in \iset}$ as a sample of $|\iset|$ pseudo-observations of the inverted multivariate Pareto distribution associated with $\xdf$. The sample size $|\iset|$ is random, and 
from \eqref{eq:margins:ranks}, we have $k \le |\iset| \le dk$. For non-empty $J \subseteq \dset$ and $i \in \iset$, we write $\hbZ_{i,J} = (\hZ_{i,j})_{j \in J}$.}

\paragraph{(2) Estimating the tail copula parameters $\theta_{\xdfset}$ in $\tree_1$.}
For each edge $e = \cbr{a_e,b_e} \in \edges_1$, we estimate the parameter (vector) $\theta_{e}$ associated with the bivariate tail copula density $r_{a_e,b_e}(\point;\theta_e)$. By Eq.~\eqref{eq:Zj}, for $j \in \cbr{a_e,b_e}$, the conditional density of $(\Imp_{a_e}, \Imp_{b_e})$ given $\Imp_{j} < 1$ is $\xdf_{a_e,b_e}(z_{a_e},z_{b_e};\theta_e) \1(z_j < 1)$.
\chng{We use maximum pseudo-likelihood estimation} to fit this density to the sub-samples 
$(\hZ_{i,a_e}, \hZ_{i,b_e})$ for both $i \in \iset_{a_e}$ and $i \in \iset_{b_e}$, with $\iset_j$ as in Eq.~\eqref{eq:Nj}.
More precisely, we maximise each of the two pseudo-likelihoods 
\begin{equation}
\label{eq:lik}
	\lik_{\xdfset,e} \rbr{ 
		\theta_{e}^{(j)} ; \;
		(\hZ_{i,a_e},\hZ_{i,b_e}), i \in \iset_j 
	} 
	=  \prod_{i \in \iset_j} r_{a_e,b_e} \rbr{ \hZ_{i, a_e}, \hZ_{i, b_e}; \theta_{e}^{(j)} }, 
	\qquad j \in \{a_e,b_e\},
\end{equation}
over $\theta_e^{(j)} \in \Theta_e$, yielding estimates $\htheta_e^{(a_e)}$ and $\htheta_e^{(b_e)}$, respectively.
The final estimate is $\htheta_{a_e,b_e} = \{\htheta_{e}^{(a_e)} + \htheta_{e}^{(b_e)}\} / 2$. The idea of averaging maximum pseudo-likelihood estimators on product spaces has already been \chng{proposed for the \HR{} model \citep{engelke2015estimation,engelke2graphicalextremes}.} 
\chng{Other estimation approaches} include censored likelihoods \citep{ledford1996statistics,de2008parametric} or empirical stable tail dependence functions \citep{einmahl2008method,einmahl2018continuous}.

\paragraph{(3) Estimating the copula parameters $\theta_{\cpdfset}$ in $\tree_2,\ldots,\tree_{d-1}$.}
\chng{Estimation of the parameters associated with edges $e = (a_e, b_e; \cndg_e)$ in $\edges_j$ for $j \in \cbr{2,\ldots,d-1}$ is based on a similar procedure as the one employed in regular vine copulas \citep[Section 7.2]{czado2019analyzing}. The main difference concerns the definition of the pseudo-observations: 
if $\htheta(\edges_{1:(j-1)})$ denotes the parameter estimates associated with the edges in $\edges_{1:(j-1)} = \edges_1 \cup \cdots \cup \edges_{j-1}$,
by Proposition~\ref{prop:Z2c}(i), for $i \in \iset_{\cndg_e}$, pseudo-observations $(\hU_{i,a_e;\cndg_e}, \hU_{i,b_e;\cndg_e})$ from $c_{a_e,b_e; \cndg_e}$ can be defined by}
\begin{equation}
\label{eq:hUiaebe}
	\hU_{i,a_e ; \cndg_e}
	= \xcdf_{a_e | \cndg_e} \rbr{\hZ_{i,a_e} \mid \hbZ_{i,\cndg_e} ; \; \htheta(\edges_{1:j-1})}, \quad
	\hU_{i,b_e ; \cndg_e}  
	= \xcdf_{b_e | \cndg_e} \rbr{\hZ_{i,b_e} \mid \hbZ_{i,\cndg_e} ; \, \htheta(\edges_{1:j-1})}.
\end{equation}
\chng{The full procedure is detailed in Section~\ref{sec:estimationdetail} of the supplement.} 

\subsection{Model selection for X-vines}
\label{sec:estimation2}

\chng{Model selection for X-vines based on a random sample $\bX_1,\ldots,\bX_n$ from $F$ 
involves:}
\begin{compactenum}[(1)]
\item selecting a regular vine sequence $\Vine = (T_1,\ldots,T_{d-1})$;
\item given the regular vine sequence obtained in (1), choosing adequate bivariate parametric (tail) copula families, $\xdfset$ and $\cpdfset$.
\end{compactenum}
In fact, the two procedures are executed together sequentially, progressing from one tree to the next tree.
First, we describe step (2) given the regular vine sequence $\Vine$.

\paragraph{Selecting parametric (tail) copula families given a truncated regular vine sequence.}
We first consider the specification of $\xdfset$ in Definition~\ref{def:xvinespec} in tree $\tree_1$.
Let $\mathcal{B}_{\xdfset,1:T}=(\mathcal{B}_{\xdfset,1},\ldots,\mathcal{B}_{\xdfset,T})$ be the list of candidate bivariate tail copula families. 
For each edge $e = (e_a,e_b) \in \edges_1$, specifying $\xdfset$ involves choosing the bivariate tail copula family among $\mathcal{B}_{\xdfset,1:T}$.
Similar to the idea of averaged maximum pseudo-likelihood estimates in Section~\ref{sec:estimation1}, we use the averaged-AIC value for selecting bivariate tail copula families.
For each edge $e$ and each $\mathcal{B}_{\xdfset,t}$, $t = 1,\ldots,T$, we obtain two maximum pseudo-likelihood estimates 
$\htheta_{e}^{(t,a_e)}$ and $\htheta_{e}^{(t,b_e)}$, derived through the maximisation of log-likelihood functions on product spaces; see \eqref{eq:lik}.
The averaged-AIC value is
$\AIC_e \rbr{\mathcal{B}_{\xdfset,t}} = 2 \nu^{(t)} - \tfrac{1}{2}\left\{\log \mathcal{L}_{\xdfset,e} \rbr{\htheta_{e}^{(t,a_e)}} + \log \mathcal{L}_{\xdfset,e} \rbr{\htheta_e^{(t,b_e)}}\right\}$,
where $\nu^{(t)}$ is the number of parameters in $\mathcal{B}_{\xdfset,t}$.
We select the bivariate tail copula family with the lowest averaged-AIC among $\AIC_e \rbr{\mathcal{B}_{\xdfset,1}}, \ldots, \AIC_e \rbr{\mathcal{B}_{\xdfset,T}}$.

\chng{Similarly, to specify $\cpdfset$ and select the bivariate parametric copula family among a list of candidates for each edge in $\edges_2 \cup \cdots \cup \edges_{d-1}$, we follow common practice in vine copula modelling and choose the family with the lowest AIC; see \citet{brechmann2010truncated} and \citet[Section~8.1]{czado2019analyzing}.}

It is worthwhile to note that when implementing the model selection step for an edge $e$ in $\edges_j$, only the trees $\tree_1,\ldots,\tree_j$ need to have been selected, but not the trees $\tree_{j+1},\ldots,\tree_{d-1}$.
This sequential approach aligns well with the vine learning procedure in the next paragraph.

\paragraph{Selecting the regular vine sequence.}
\citet{morales2010counting} showed that the number of regular vine sequences on $d$ elements is equal to $d! \, 2^{\binom{d-2}{2} - 1}$, making it impossible to go through all possible vine structures. 
We adopt a model selection approach similar to the one in \citet{dissmann2013selecting}, choosing trees sequentially from $\tree_1$ to $\tree_{d-1}$.

\chng{To select the first tree, $\tree_1$, on the node set $\nodes_1 = \dset$, we follow a procedure as in \citet{engelke2022structure} and \citet{hu2023modelling}. For every pair $\cbr{a,b}$ of distinct elements in $\dset$, let $w_{a,b}$ be a nonnegative weight derived from the data. We use an empirical version of the tail dependence coefficient $\chi_{a,b} = R_{a,b}(1,1)$, setting $w_{a,b}$ equal to
$\hchi_{a,b} =
    \frac{1}{k} \sum_{i=1}^n \1 \cbr{\hU_{i,a} \leq k/n, \hU_{i,b} \leq k/n}$,
where $\hU_{i,a}$ and $\hU_{i,b}$ are defined as in \eqref{eq:margins:ranks}. Another possible edge weight could be the empirical extremal variogram as in \citet{engelke2022structure}.
For subsequent trees $\tree_2,\ldots,\tree_{d-1}$, the edge weight is chosen to be the absolute value of the empirical Kendall's tau. All trees are selected as maximum spanning trees; see, for example, \citet[Section 8.3]{czado2019analyzing}.}

\paragraph{Selecting truncated regular vine tree sequences.}
The model selection procedure described above may set all pair-copulas to the independence copula in the subsequent trees from $\tree_{q+1}$ to $\tree_{d-1}$ for a truncation level $q \in \cbr{1,\ldots,d-2}$.
In this case, the resulting model corresponds to the truncated X-vine model in Definition~\ref{def:truncXvine}.
Besides an information criterion, we consider two additional criteria for selecting the independence copula at an edge $e$: when the effective sample size $|\iset_{D_e}|=n_{D_e}$ falls below a certain low value, or when the absolute value of the empirical Kendall's tau is close to zero. \chng{Sparsity induced by a smaller effective sample size, is more likely when the total sample size $n$ is relatively small with respect to the dimension $d$.}

Even when the above criteria are not met, when $d$ is large, it is natural to limit the number of model parameters by considering truncated X-vines, since Dissmann's algorithm captures as much dependence as possible in the first few trees. We use a modified Bayesian information (mBIC) to determine the truncation level, inspired by the one for regular vine copulas in \cite{nagler2019model}.
This modified version adjusts the prior probability in the BIC to penalise dependence copulas more severely in trees at higher levels.
More specifically, assuming that for any edge $e \in \edges_j$ for $j \ge 2$, the parametric family has a single parameter $\theta_{e}$
and that a value of $\theta_e = 0$ corresponds to the independence copula, the mBIC includes independent Bernoulli variables $\1(\theta_e\ne 0)$ with mean $\psi_e=\psi_0^{j-1}$ for $e\in \edges_j$ and a hyperparameter $\psi_0 \in (0, 1)$: set $\mBIC(1) = 0$ and, for $q \in \cbr{2,\ldots,d-1}$, put
\begin{equation}
\label{eq:mBIC}
    \mBIC(q) =  \sum_{j=2}^q \sum_{e \in \edges_j} \sbr{
        \1(\theta_e\ne 0) \cbr{\log {n_{D_e}}-2\log\rbr{\frac{\psi_e}{1-\psi_e}}} 
        - 2 \log \left\{  \frac{\lik_{\cpdfset,e} \rbr{\htheta_{e}}}{(1 - \psi_e )^{-1}}\right\}
    }
\end{equation}
with $\lik_{\cpdfset,e}$ as in \eqref{eq:lik2}.
The selected truncation level $q^*$ is the value of $q$ in $\cbr{1,\ldots,d-1}$ that minimises $\mBIC(q)$. In practice, we will set $\psi_0 = 0.9$ as in \citet{nagler2019model}.

\section{Simulation study}
\label{sec:simulation-study}

We conduct three simulation studies, evaluating the proposed procedures for parameter estimation, selection of bivariate parametric (tail) copula families, and vine truncation.

\subsection{Parameter estimation}
\label{sec:simulation-study:estim}

We first consider the $5$-dimensional X-vine model $\xdf$ in Fig.~\ref{fig:5d_Xvine_a}.
The bivariate tail copula densities $\xdf_{a_e,b_e}$ in the first tree $\tree_1=(\nodes_1,\edges_1)$ are chosen from the \HR{} model, the negative logistic, logistic, and Dirichlet families (Section~\ref{sec:parametric}), 
while the bivariate copula densities $c_{a_e,b_e;\cndg_e}$ in the subsequent trees $\tree_j = (\nodes_j,\edges_j)$ for $j \in \cbr{2,3,4}$ are taken from the Clayton, Gumbel, and Gaussian copula families.
\chng{Fig.~\ref{fig:5d_Xvine_a} shows the values specified for the tail dependence coefficient $\chi = R(1,1)$ or for Kendall's $\tau$ for each edge.} 
The formulas connecting the parameters of the families of tail copula densities with 
$\chi$ are given in \chng{Appendix~\ref{sec:chi} in the supplement}.

\begin{figure}
\centering
  \subfloat[\label{fig:5d_Xvine_a}]{
  \resizebox{0.59\textwidth}{!}{
  \begin{tikzpicture}
	\node[draw] (2) at (0, 3) {$2$};
	\node[draw] (1) at (-3, 3) {$1$};
	\node[draw] (3) at (3, 3) {$3$};
    \node[draw] (4) at (0, 2) {$4$};
    \node[draw] (5) at (3, 2) {$5$};
    \draw (1) -- (2) -- (3);
	\draw (2) -- (4) -- (5);
    \draw (1) -- (2) node[midway, above]{\tiny\fbox{HR;$\chi=0.54$}};
    \draw (2) -- (3) node[midway, above]{\tiny\fbox{NL;$\chi=0.71$}};
    \draw (2) -- (4) node[midway, right]{\tiny\fbox{L;$\chi=0.68$}};
    \draw (4) -- (5) node[midway, below, sloped]{\tiny\fbox{Diri;$\chi=0.62$}};
    \node (T1) at (-5,3) {$T_1$};
	\node[draw] (23) at (-1, 0.5) {$23$};
	\node[draw] (12) at (-4, 0.5) {$12$};
	\node[draw] (24) at (2, 0.5) {$24$};
	\node[draw] (45) at (5, 0.5) {$45$};
    \draw (12) -- (23) -- (24) -- (45);
    \draw (12) -- (23) node[midway, above]{\tiny\fbox{Clay;$\tau=0.50$}};
    \draw (23) -- (24) node[midway, above]{\tiny\fbox{Gum;$\tau=0.60$}};
    \draw (24) -- (45) node[midway, above]{\tiny\fbox{Ga;$\tau=0.49$}};   
    \node (T2) at (-5,0.5) {$T_2$};
    \node[draw] (13) at (-2.5, -1) {$13;2$};
	\node[draw] (34) at (0.5, -1) {$34;2$};
	\node[draw] (25) at (3.5, -1) {$25;4$};
 	\draw (13) -- (34) -- (25);
    \draw (13) -- (34) node[midway, above]{\tiny\fbox{Clay;$\tau=0.17$}};
    \draw (34) -- (25) node[midway, above]{\tiny\fbox{Ga;$\tau=-0.19$}};  
    \node (T3) at (-5,-1) {$T_3$};
    \node[draw] (14) at (-1.5, -2.5) {$14;23$};
	\node[draw] (35) at (2, -2.5) {$35;24$};  
    \draw (14) -- (35);
    \draw (14) -- (35) node[midway, above]{\tiny\fbox{Ga;$\tau=0.06$}};
    \node (T4) at (-5,-2.5) {$T_4$};
    \end{tikzpicture}
    }} 
      \subfloat[\label{fig:boxplot_Fam}]{\includegraphics[width=0.39\textwidth]{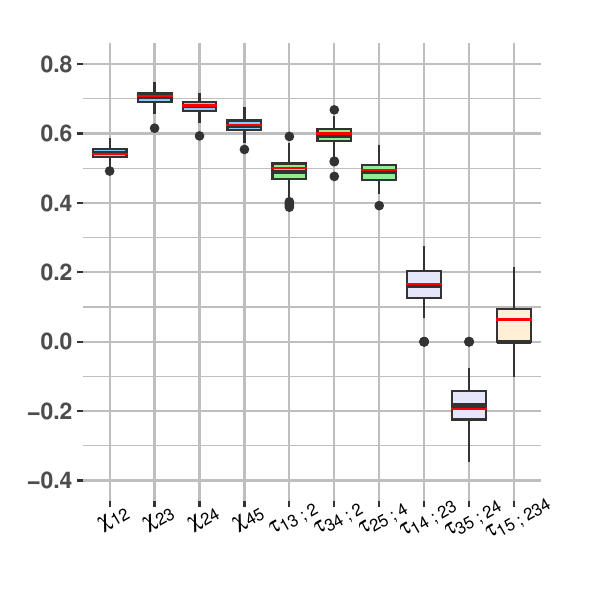}}%
    \caption{\label{fig:5d_Xvine}
    	(a) An X-vine specification $(\Vine, \xdfset, \cpdfset)$ on $d=5$ variables. The bivariate tail copula densities $\xdf_{a_e,b_e}$ in the first tree $\tree_1$ are selected from the \HR{} (HR), negative logistic (NL), logistic (L), and Dirichlet (Diri) parametric families, each with specified tail dependence coefficients $\chi$. The bivariate copula densities $c_{a_e,b_e;\cndg_e}$ in the subsequent trees $\tree_2,\tree_3,\tree_4$ are chosen from the Clayton (Clay), Gumbel (Gum), and Gaussian (Ga) copula families, each with specified Kendall's tau $\tau$. (b)  Box-plots (\textcolor{lightskyblue}{$\blacksquare$} for $\tree_1$
    \textcolor{lightgreen}{$\blacksquare$} for $\tree_2$
    \textcolor{lavender(web)}{$\blacksquare$} for $\tree_3$
    \textcolor{papayawhip}{$\blacksquare$} for $\tree_4$) \chng{of dependence measure estimates from bivariate parametric families selected sequentially from the data for the X-vine specification in (a).
	The red lines indicate the specified parameter values.}}
\end{figure}
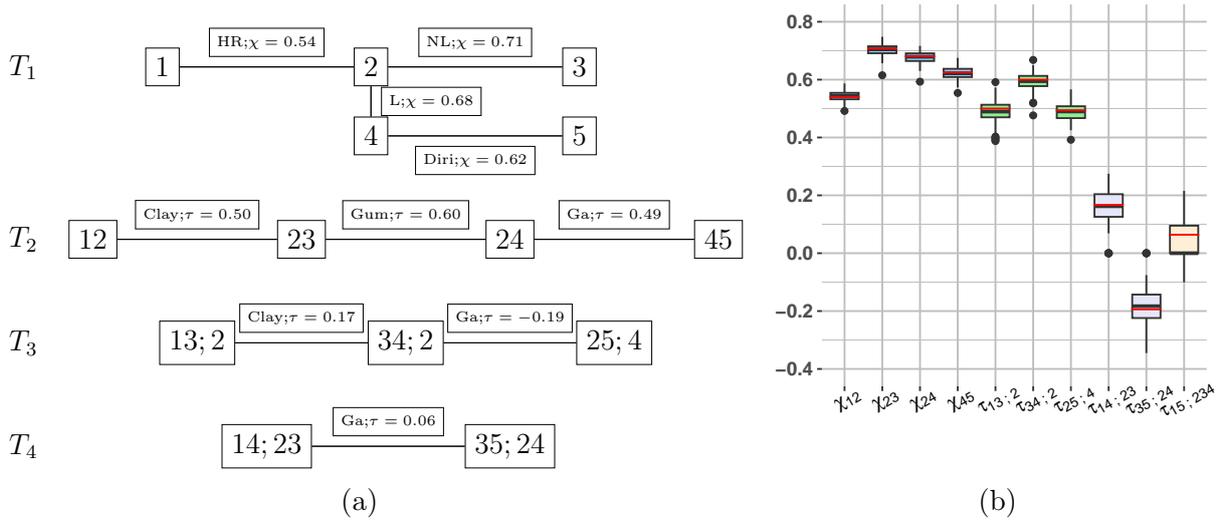

Implementing simulation algorithms as detailed in \chng{Section~\ref{sec:simulation} and Appendix~\ref{app:example} in the supplement}, we generate multivariate inverted Pareto random samples $\bZ_1,\ldots,\bZ_n$ associated with $\xdf$.
\chng{As in Eq.~\eqref{eq:margins:ranks}, 
we transform to $\hbU_1, \ldots, \hbU_n$, with $\hat{U}_{i,j}$ based on the rank of $Z_{i,j}$ among $Z_{1,j}, \ldots, Z_{n,j}$.}
We then take \chng{the sub-sample} $\{\hbZ_i\}_{i\in \iset}$ with $\iset = \bigcup_{j \in \dset} \iset_j$ and $\iset_j$ as in \eqref{eq:Nj}.

We set $(n,k)=(4\,000,200)$ and perform sequential parameter estimation as in Section~\ref{sec:estimation1} over 200 repetitions.
We obtain maximum pseudo-likelihood estimates and \chng{corresponding} dependence measures, $\hchi$ and $\htau$, \chng{using the relations in Appendix~\ref{sec:chi} of the supplement,} 
for each of the ten edges in the vine.
In Fig.~\ref{fig:boxplot_Fam}, box-plots present dependence measure estimates of parametric families based on the X-vine specification in Fig.~\ref{fig:5d_Xvine_a}.
The four left-most box-plots (\textcolor{lightskyblue}{$\blacksquare$}) show the tail dependence coefficient estimates $\hchi$ for the four edges $e \in \edges_1$, and the remaining plots display Kendall's tau estimates $\htau$ for the six edges $e \in \edges_2 \cup \edges_3 \cup \edges_4$.
\chng{The plot supports the validity of the sequential method for estimating tail dependence measures. 
We see that estimation uncertainty becomes larger at higher tree levels. Sections~\ref{sec:simu:addtional:dep}--\ref{sec:boxplotmle} in the supplement include additional box-plots: first, of dependence measures as in Fig.~\ref{fig:boxplot_Fam}, but supposing that the bivariate parametric families are known, second, of dependence measures with varying sample sizes and threshold exceedances, and third}, of maximum likelihood estimates of tail copula densities for a specific edge $e\in \edges_1$ with varying $(n,k)$.

\subsection{Selecting parametric (tail) copula families}
\label{sec:simulation-study:select}

We assess algorithm effectiveness in selecting bivariate parametric families for each edge in each tree, using the \chng{X-vine specification in} Fig.~\ref{fig:5d_Xvine_a}.
While the flexibility of X-vine models allows us to consider any bivariate parametric family, we simplify the process by considering a catalogue of four candidate tail copula models for $\tree_1$: the \HR, logistic, negative logistic, and Dirichlet models, along with a catalogue of nine candidate pair-copula families for $\tree_2,\tree_3,\tree_4$: Independence, Gaussian, Clayton, Survival Clayton, Gumbel, Survival Gumbel, Frank, Joe, and Survival Joe copulas, as implemented in the \chng{R package \textsf{VineCopula}} \citep{vinecopula}.

Using again 200 repetitions with $(n,k)=(4\,000,200)$, we assess the accuracy of the (averaged) AIC in Section~\ref{sec:estimation2} in selecting bivariate parametric families. 
The overall proportion of correctly selected families across all trees is 60\%.
For individual trees, the proportions are 
62\% for $\tree_1$, 78\% for $\tree_2$, 55\% for $\tree_3$, and 10\% for $\tree_4$. 
\chng{The specific proportions for each edge in each tree are }
\begin{center}
        \begin{tabular}{rrrrr}
        \midrule
        $\tree_1$ & 99\% & 40\% & 60\% & 49\% \\ 
        $\tree_2$ &  & 58\% & 93\% & 82\% \\ 
        $\tree_3$ &  &  & 42\% & 68\% \\ 
        $\tree_4$ &  &  &  & 11\% \\ 
        \bottomrule
       \end{tabular}
\end{center}
\bigskip

The lower proportions observed for $e_{23}$ and $e_{45}$ in $\edges_1$ result from the challenge in distinguishing between the logistic and negative logistic models.
Overall, the proportion of accurately selected families decreases with declining effective sample size and Kendall's tau estimate across tree levels.
We observe the lowest proportion for the deepest edge $e_{15;234}$ where the corresponding pair-copula exhibits weak dependence. For this edge, the algorithm selects the independence copula 110 times out of 200.

\chng{We also investigate the effective sample sizes for each tree in Section~\ref{sec:avsamplesize} of the supplement}.

\subsection{Tree selection and truncation}
\label{sec:simulation-study:trunc}

We consider a higher dimensional X-vine model where residual dependence weakens with increasing tree level.
Specifically, for $\Vine$ we consider a $50$-dimensional C-vine, that is, in each tree $\tree_j$, there is a single node $a_j \in \dset$ such that $a_j$ belongs to the conditioned set $\cndd_e$ of all edges $e \in \edges_j$.
The first tree includes the \HR{} models and negative logistic models with randomly assigned parameter values $\theta_{e} \in [1,2]$ for $e \in \edges_1$.
Subsequent trees contain bivariate Gaussian copulas with partial correlations $\rho_{e} = 1.1 - 0.1 j$ for $e \in \edges_{j}$ and $j\in\cbr{2,\ldots,9}$, and $\rho_{e}=0.1$ for $e \in \edges_{j}$ with $j \ge 10$.
This X-vine specification allows us to explore truncated X-vine models by setting pair-copulas with weak dependence to independence copulas.

\chng{We use a single inverted multivariate Pareto sample of size $n=1\,000$ from the X-vine model directly, i.e., without rank transformation and thresholding.}
\chng{As in Sections~\ref{sec:estimation1} and~\ref{sec:estimation2}}, we sequentially select trees $\tree_1, \ldots, \tree_{d-1}$ using $\hchi_e$ for $e\in \edges_1$ and $\htau_e$ for $e\in \edges_j$, $j \ge 2$, as edge weights, \chng{with $\hchi_e$ as in Eq.~\eqref{eq:chiab} in the supplement.}
For each selected tree, we choose the bivariate parametric families with the lowest (averaged) AIC for each edge and estimate the associated parameters.
Additionally, independence copulas are chosen in subsequent trees if either $|\htau_e| < 0.05$ or $n_{D_e}<10$ for each edge.

To explore truncated X-vine models, we use the mBIC in \eqref{eq:mBIC} with $\psi_0=0.9$ as in \cite{nagler2019model}.
\chng{The estimated mBIC-optimal truncation level is $q^*=19$ (dotted line in Fig.~\ref{fig:mBIC50dim}).} 
\chng{We assess the goodness of fit by comparing pairwise $\chi$-values from the true X-vine model with those from the fitted 50-dimensional X-vine model via Monte Carlo simulations 
(\chng{as explained in Section~\ref{sec:taildepcoefgof} in the supplement}) in Fig.~\ref{fig:Chi50dim}. 
For the truncated X-vine model,} the $\chi$-plot in Fig.~\ref{fig:Chi50Trunc} resembles that of the full model but exhibits more variability, particularly for lower $\chi$-values.

\begin{figure}%
    \centering
    \subfloat[\label{fig:mBIC50dim}]{\includegraphics[width=0.33\textwidth]{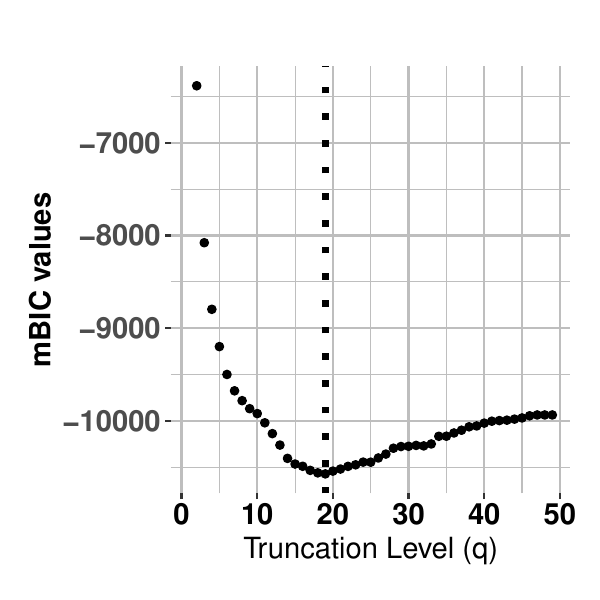}}%
    \subfloat[\label{fig:Chi50dim}]{\includegraphics[width=0.33\textwidth]{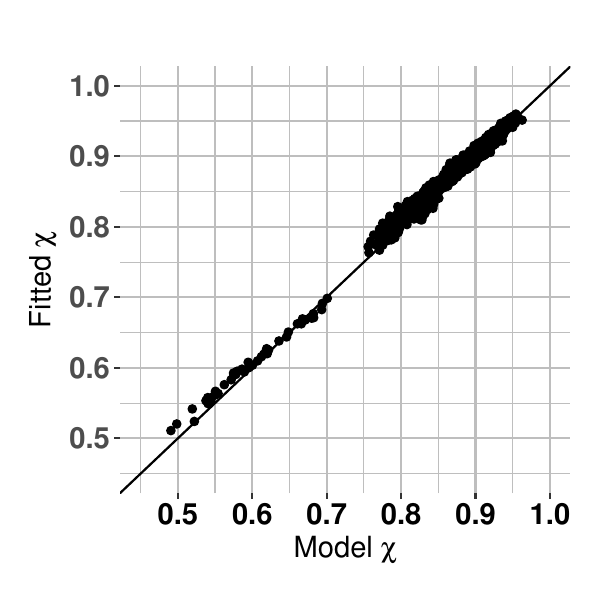}}%
    \subfloat[\label{fig:Chi50Trunc}]{\includegraphics[width=0.33\textwidth]{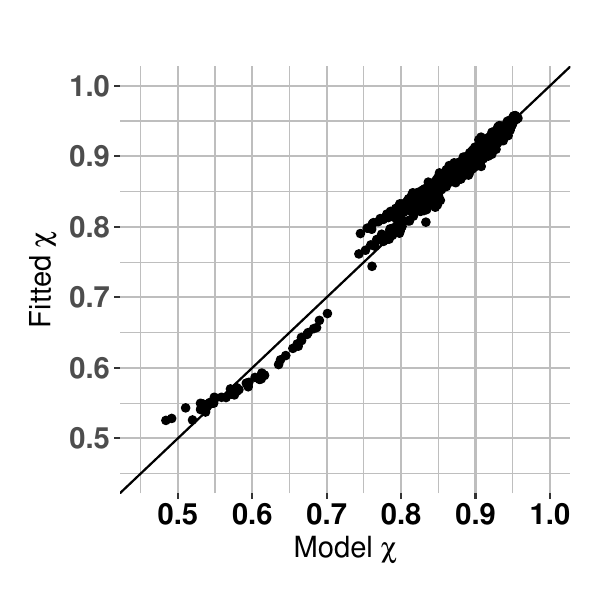}}%
    \caption{\label{fig:chisim} Simulation study: (a) The mBIC plotted across tree levels with a dotted line indicating the selected mBIC-optimal truncation level of $q^*=19$. 
    \chng{(b) $\chi$-plot comparing pairwise $\chi$ values from the true X-vine model to those obtained from the fitted 50-dimensional X-vine model via Monte Carlo simulations. 
    (c) Similar to (b) but for the truncated fitted X-vine model.}}%
\end{figure}

\section{Application: US flight delay data}
\label{sec:casestudy}

We apply X-vine models to investigate extremal dependence among large flight delays in the US flight network.
The raw data set is accessible through the US Bureau of Transportation Statistics\footnote{\url{https://www.bts.dot.gov}}.
These flight delay data were analysed by \cite{hentschel2022statistical} who first selected airports with a minimum of 1\,000 flights per year and applied a $k$-medoids clustering algorithm to identify homogeneous clusters in terms of extremal dependence.
This clustering approach not only makes the analysis suitable for modelling extremal dependence but also reveals shared frequent flight connections between airports and similar geographical characteristics in each cluster.
Focusing on the \HR{} family, \cite{hentschel2022statistical} fitted an extremal graphical model to large flight delays of airports in the Texas cluster to investigate conditional independence.

For the purpose of model comparison, we also analyse daily total delays (in minutes) in the Texas cluster.
The cluster comprises $d = 29$ airports and counts $n = 3603$ days from 2005 to 2020, during which all airports have recorded measurements.
This pre-processed data is available through the R-package \textsf{graphicalExtremes} \citep{engelke2graphicalextremes}.
A graphical representation of the actual flight connections between airports is presented in Fig.~\ref{fig:FlightGraphDomain}.

\begin{figure}[ht]%
    \centering
    \subfloat[\label{fig:FlightGraphDomain}]{\includegraphics[width=0.28\textwidth]{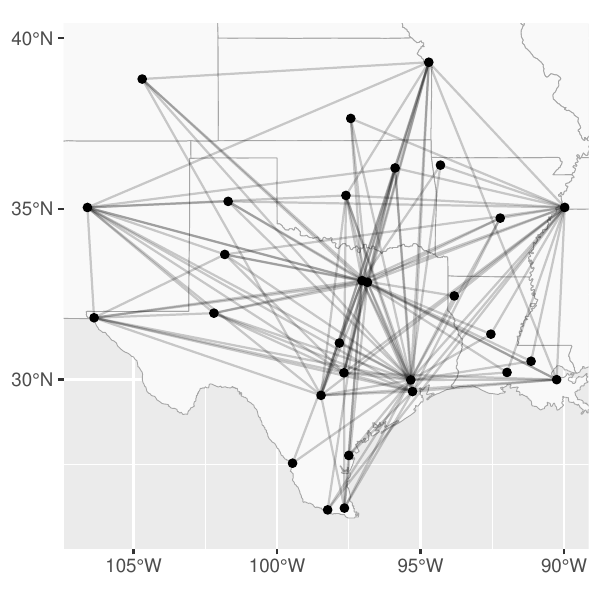}}%
    \quad
    \subfloat[\label{fig:FlightGraphXvineTrunc}]{\includegraphics[width=0.28\textwidth]{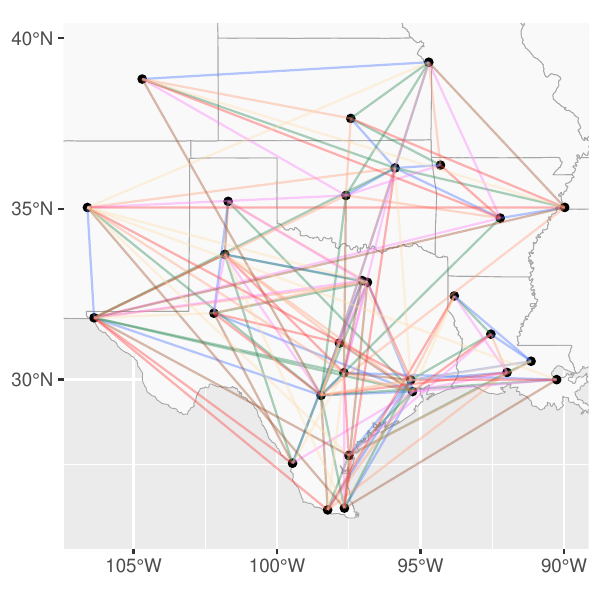}}%
    \quad
    \subfloat[\label{fig:FlightGraphEGlearn}]{\includegraphics[width=0.28\textwidth]{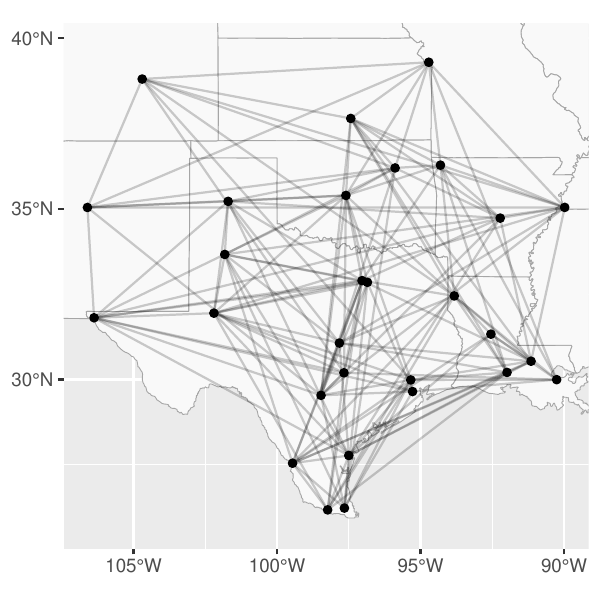}}%
    \\
    \subfloat[\label{fig:ChiFlightDomain}]{\includegraphics[width=0.28\textwidth]{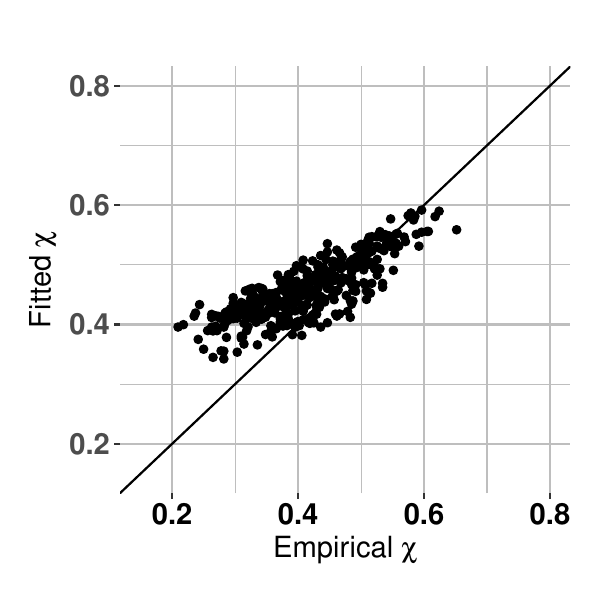}}%
    \quad
    \subfloat[\label{fig:ChiFlightTruncXvine}]{\includegraphics[width=0.28\textwidth]{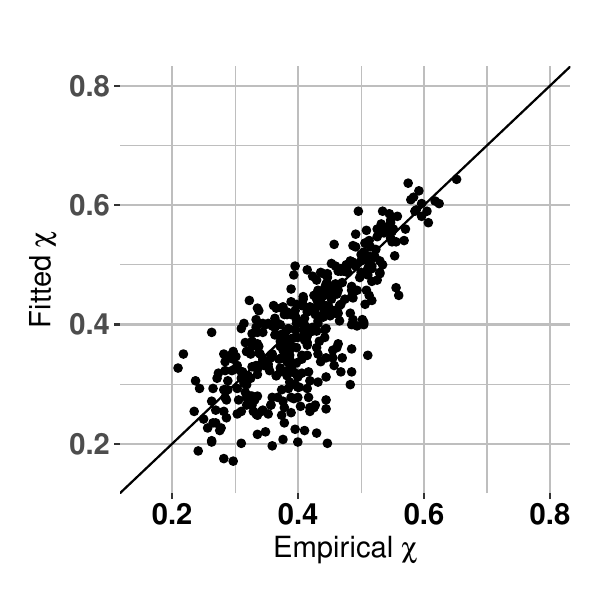}}%
    \quad
    \subfloat[\label{fig:ChiFlightEGlearn}]{\includegraphics[width=0.28\textwidth]{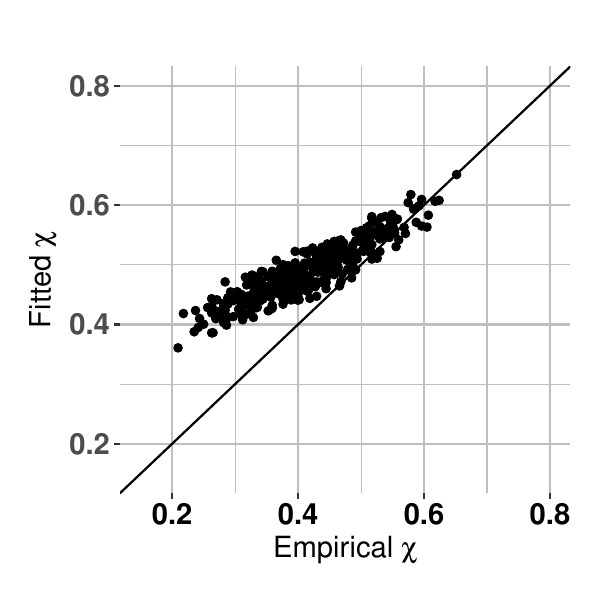}}%
    \caption{
(a) Flight graph illustrating connection flights between airports joined by edges in the Texas cluster. 
(b) The X-vine graph with the first seven trees superimposed, showing for trees $\tree_2$ to $\tree_7$ only the 113 out of 147 edges for which a copula other than the independence copula is selected, for 141 edges in total.
(c) The estimated extremal \HR{} graph structure using the regularised method with a tuning parameter value of $\rho^*=0.1$, totalling 148 edges.
Corresponding $\chi$-plots comparing empirical pairwise tail dependence coefficients with those from the fitted graphical models: (d) the flight graph, (e) the truncated X-vine graph, and (f) the extremal \HR{} graph using the \textsf{EGlearn} algorithm.}%
\end{figure}

Let $\widehat{\bx}_i$, for $i=1,\ldots,n$, denote the measurements.
Following the standardisation process with the rank transformation \eqref{eq:margins:ranks} as in Section~\ref{sec:estimation1}, we take \chng{the sub-sample}
$\hbz_{i} = (n/k)\hbu_{i}$ for $i\in \iset = \bigcup_{j\in\dset} \iset_j$ as in \eqref{eq:Nj}.
We choose $k = 0.13 n$ in order to have a large enough effective sample size \chng{with respect to} the number of variables.

Assigning $\hchi_e$ for $e\in \edges_1$ and $\htau_e$ for $e\in \edges_j$, $j \ge 2$, as edge weights, we select the vine tree structure sequentially as in Section~\ref{sec:estimation2}.
\chng{Fig.~\ref{fig:XvineFlightGrapha} in the supplement shows the first maximum spanning tree. The corresponding estimated tail dependence coefficients vary between $0.42$ and $0.65$, with a mean of $0.55$.}
Out of the $d-1=28$ edges in $\tree_1$, selected bivariate tail copula families include the \HR{} ($8$ edges), negative logistic ($3$ edges) and Dirichlet models ($17$ edges).
Subsequent trees with a total of $378$ edges are then chosen with the following pair-copula families (number of edges in parentheses): Independence (201), Gaussian (31), Clayton (8), Gumbel (26), Frank (34), Joe (36), Survival Clayton (31), Survival Gumbel (9), and Survival Joe (2).
As the tree level rises, we observe a decline in residual dependence and an increase in the number of independence copulas.
The X-vine algorithm sets all pair copulas to the independence copula from \chng{$\tree_{21}$}, that is, the vine is truncated at level $q=20$.

To evaluate the X-vine model's efficacy in capturing extremal dependence for large flight delays between airports, 
\chng{we compare bivariate and trivariate empirical tail dependence coefficients 
with those from the fitted X-vine models in Section~\ref{sec:gofXvine} in the supplement. Fig.~\ref{fig:chifit} indicates a satisfactory fit, both for the full $29$-dimensional model and for the truncated model.}

To further explore truncated X-vine models with a lower truncation level, we use the mBIC in Eq.~\eqref{eq:mBIC}.
Fig.~\ref{fig:mBICFlight} \chng{in the supplement}
illustrates the mBIC-optimal truncation level of $q^*=7$, which corresponds to 113 dependence copulas [Gaussian (17), Clayton (3), Gumbel (23), Frank (16), Joe (25), Survival Clayton (22), Survival Gumbel (6), and Survival Joe (1)] out of the 147 total edges from $\tree_2$ to $\tree_7$. The total number of bivariate model components is thus $28+113=141$.
Additionally, we investigate how the truncation level changes over the threshold range 
in Fig.~\ref{fig:mBICvsQuan}.
It appears that the mBIC-optimal truncation levels are not overly sensitive to the choice of the threshold.
The superimposed graph of the first seven vine trees is shown in Fig.~\ref{fig:FlightGraphXvineTrunc}, showing for $\tree_2$ to $\tree_7$ only the 113 out of 147 edges for which a copula other than the independence one is selected. \chng{The vine tree sequence from $\tree_1$ to $\tree_7$} is presented in Fig.~\ref{fig:XvineFlightGraph}. 

Returning to model comparison, we consider \chng{extremal \HR{} graphical models 
\citep{engelke2021learning,hentschel2022statistical}.
\citet{hentschel2022statistical} obtained a sparse \HR{} graphical model 
using the \textsf{EGlearn} algorithm \citep{engelke2021learning}.}
The tuning parameter $\rho\ge 0$ controls sparsity. The empirical extremal variogram matrix $\hGamma$ is used as edge weight for the \emph{minimum} spanning tree with smaller elements in $\hGamma$ indicating stronger extremal dependence.
Following 
\citet{engelke2graphicalextremes}, the data set is split in half: a training set for model fitting and a test set for tuning parameter selection.
We determine the optimal tuning parameter, $\rho^*=0.1$, by evaluating the \HR{} log-likelihood across tuning parameter values $\rho$ from the test set.
The resulting sparse extremal graph with 148 edges and $\rho^*=0.1$ is shown in Fig.~\ref{fig:FlightGraphEGlearn}.

We assess the goodness of fit using the entire sub-sample and compare empirical tail dependence coefficients to those obtained from the fitted graphical models as explained in \chng{Appendix~\ref{sec:taildepcoefgof} in the supplement}: the \HR{} extremal graphical model for the flight graph in Fig.~\ref{fig:ChiFlightDomain}, the truncated X-vine model with $q^*=7$ in Fig.~\ref{fig:ChiFlightTruncXvine}, and the extremal \HR{} graphical model with $\rho^*=0.1$ in Fig.~\ref{fig:ChiFlightEGlearn}, respectively.

The flexibility of a regular vine structure allows the selection of various (tail) copula families with the lowest AIC.
Recall from Section~\ref{sec:hr} that the \HR{} model arising from the X-vine specification is constructed using bivariate \HR{} tail copula densities in the first tree and bivariate Gaussian copulas in subsequent trees.
In the truncated X-vine model, 8 out of 28 edges prefer \HR{} models in $\tree_1$, while in trees $\tree_j$ for $j\ge 2$, a total of 96 out of 113 edges favour copula families other than the Gaussian one.
However, uncertainty in sequential parameter estimation tends to accumulate across tree levels.
Consequently, the $\chi$-plot of the truncated X-vine model shows a better fit but more variability.
In contrast, the $\chi$-plot of the extremal \HR{} graphical model has less variability but seems biased towards higher tail dependence.
\chng{Fig.~\ref{fig:ChiComparisonHR} in the supplement presents a model comparison similar to Figs.~\ref{fig:ChiFlightDomain}---\ref{fig:ChiFlightEGlearn}, but focusing exclusively on the \HR{} model and using empirical tail dependence coefficients derived from the empirical extremal variogram matrix.}

\section{Discussion}
\label{sec:conc}

We have opened the door to the construction and the use of extremal dependence models based on regular vine tree sequences. For (ordinary) copulas, the methodology has been extensively developed in the literature since its inception more than two decades ago. Our contribution is to deliver the theoretical and methodological advances required to apply vine machinery in the extreme value context too. The key consists of a version of Sklar's theorem applied to tail copula densities, 
together with a telescoping product formula for regular vines.

While the proposed methodology is fully operational, it is clearly open to improvement, while many open questions remain, just as for ordinary regular vine copulas. We name just a few. The vine structure learning method inspired by \cite{dissmann2013selecting} aims to capture dependence by the first several trees in the sequence, but there is no guarantee that it retrieves the true structure. The regular vine atlas of \cite{morales2023chimera} can serve as a test bed for evaluating vine learning approaches. While the parameter estimators were shown to perform well in simulations, their large-sample theory remains to be developed, in line with \cite{hobaek2013parameter}. For the Hüsler–Reiss model, the precise connection between the variogram matrix and the correlation parameters of the bivariate Gaussian copulas remains to be elucidated and the effect of truncation on the Hüsler–Reiss precision matrix \citep{hentschel2022statistical} to be uncovered. The relation between X-vines and graphical models for extremes as in \citet{engelke2020graphical} deserves further investigation, perhaps leveraging results of \citet{zhu2022regular}. Finally, as our approach is limited to dependence structures generated from max-stable distributions, a completely open question is whether it can be extended to other settings in multivariate extreme value analysis, such as the conditional extremal dependence model \citep{heffernan2004conditional} or the geometric sample-cloud approach \citep{nolde2014geometric,simpson2021geometric}.

\section*{Acknowledgments}
\chng{The authors thank Manuel Hentschel for providing the US flight delay data and gratefully acknowledge two reviewers and an associate editor for valuable comments.} 
\chng{This work is supported by funds from the \emph{Fonds de la Recherche Scientifique – FNRS}, Belgium (grant T.0203.21).}

\setlength{\bibsep}{0pt}
\bibliographystyle{chicago}
\bibliography{biblio}

\newpage

\begin{center}
\Large\bfseries
Supplementary material for\\
``X-Vine Models for Multivariate Extremes''
\end{center}

\bigskip

\appendix

\renewcommand\thefigure{\thesection.\arabic{figure}}    

In Section~\ref{app:example}, we provide detailed illustrations of a number of concepts, results and methods for a particular five-dimensional regular vine sequence. Section~\ref{app:proofs} contains the proofs of all results in the paper. 
For readers more familiar with exponent measure densities than with tail copula densities, a number of theoretical results are reformulated in terms of exponent measure densities in Section~\ref{app:expmeas}. 
\chng{Section~\ref{sec:estimationdetail} provides more details on the estimation of the copula parameters in the trees after the first one. Expressions for tail dependence coefficients $\chi_J$ (analytical, empirical, and via Monte Carlo) are provided in Section~\ref{app:chi}.}
Finally, Section~\ref{sec:additional} contains additional numerical results, both for the Monte Carlo simulations in Section~\ref{sec:simulation-study} and the US flight data case study in Section~\ref{sec:casestudy}.


\section{Example in dimension five}
\label{app:example}

\setcounter{figure}{0}

Consider the regular vine sequence $\Vine = (\tree_1,\ldots,\tree_4)$ on $d = 5$ nodes specified by the following trees $\tree_j = (\nodes_j,\edges_j)$ for $j = 1,\ldots,4$:
\begin{align*}
\tree_1: N_1 & = \{1,2,3,4,5\} \\
\edges_1 & = \{\{1,2\}, \{2,3\},\{2,4\},\{4,5\} \} \\
\tree_2: N_2 & = 	 E_1 = \{\{1,2\}, \{2,3\},\{2,4\},\{4,5\} \} \\
\edges_2 & = \{ \{\{1,2\}, \{2,3\}\}, \{\{2,3\}, \{2,4\}\}, \{\{2,4\}, \{4,5\}\} \} \\
\tree_3: N_3 & = 	 E_2 = \{ \{\{1,2\}, \{2,3\}\}, \{\{2,3\}, \{2,4\}\}, \{\{2,4\}, \{4,5\}\} \} \\
\edges_3 & = \{ \{ \{\{1,2\}, \{2,3\}\}, \{\{2,3\}, \{2,4\} \} \}, \{ \{\{2,3\}, \{2,4\}\}, \{\{2,4\}, \{4,5\} \} \} \} \\
\tree_4: N_4 & = 	 E_3 = \{ \{ \{\{1,2\}, \{2,3\}\}, \{\{2,3\}, \{2,4\} \} \}, \{ \{\{2,3\}, \{2,4\}\}, \{\{2,4\}, \{4,5\} \} \} \} \\
\edges_4 & = \{ \{ \{ \{\{1,2\}, \{2,3\}\}, \{\{2,3\}, \{2,4\} \} \}, \{ \{\{2,3\}, \{2,4\}\}, \{\{2,4\}, \{4,5\} \} \} \} \}
\end{align*}

In tree $\tree_3$, consider the edge $e = \{a,b\} = \{ \{\{1,2\}, \{2,3\}\}, \{\{2,3\}, \{2,4\} \} \} \in \edges_3$, which is also a node in tree $\tree_4$. The proximity condition is satisfied since $|a \cap b| =  |\{\{2,3\}\}| = 1$: note that $a \cap b$ is indeed a singleton, the unique element of which is the set $\{2,3\}$. The edge~$e$ has complete union $\cunn_e = \{1,2,3,4\}$, which is the union of the conditioning sets $\cunn_a = \cbr{1,2,3}$ and $\cunn_b = \cbr{2,3,4}$. Further, $e$ has conditioning set $\cndg_e = \{2,3\} = \cunn_a \cap \cunn_b$, and conditioned set $\cndd_e = \cunn_a \symdiff \cunn_b = \{1,4\}$, which is itself the union $\cndd_e = \cndd_{e,a} \cup \cndd_{e,b}$ of $\cndd_{e,a} = \cunn_a \setminus \cunn_b = \cbr{1}$ and $\cndd_{e,b} = \cunn_b \setminus \cunn_a = \cbr{4}$.  Writing this edge in its conditioned form gives better readability: $e =(\cndd_{e,a},\cndd_{e,a};\cndg_e) = (14;23)$. Fig.~\ref{fig:rvineexa} shows the regular vine sequence, where the nodes and edges have been labeled using the conditioned forms.

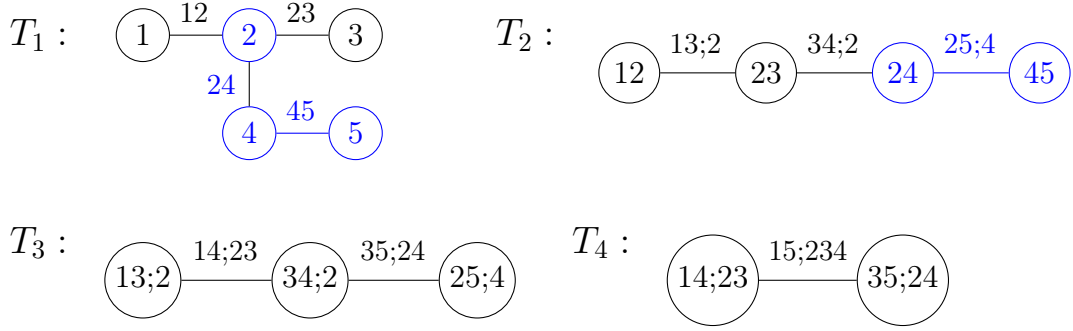
\begin{figure}[ht]
\vspace{0.5cm}
\begin{equation*}	
	\begin{tikzpicture}
	\node [circle,draw=black,fill=white,inner sep=0pt,minimum size=0.7cm] (1) at (0,0) {1};
    \node [blue,circle,draw=blue,fill=white,inner sep=0pt,minimum size=0.7cm] (2) at (1.4,0) {2};
    \node [circle,draw=black,fill=white,inner sep=0pt,minimum size=0.7cm] (3) at (2.8,0) {3};
    \node [blue,circle,draw=blue,fill=white,inner sep=0pt,minimum size=0.7cm] (4) at (1.4,-1.3) {4};
    \node [blue,circle,draw=blue,fill=white,inner sep=0pt,minimum size=0.7cm] (5) at (2.8,-1.3) {5};
    \path [draw,-] (1) edge (2);
    \path [draw,-] (2) edge (3);
    \path [draw,-] (2) edge (4);
    \path [draw,color=blue] {(4) edge (5)};
    \node [text width=0.7cm] (T1) at (-1.4,0) {\large $\tree_1:$};
    \node [text width =0.7cm] (E1) at (0.85,0.3) {\small 12};
    \node [text width =0.7cm] (E1) at (1.4+0.85,0.3) {\small 23};
    \node [text width =0.7cm,blue] (E1) at (1.2,-0.65) {\small 24};
    \node [text width =0.7cm,blue] (E1) at (1.4+0.85,-1) {\small  45};

    \node [circle,draw=black,fill=white,inner sep=0pt,minimum size=0.8cm] (12) at (6+0.4,-0.5) {12};
    \node [circle,draw=black,fill=white,inner sep=0pt,minimum size=0.8cm] (23) at (6+1.8+0.4,-0.5) {23};
    \node [blue,circle,draw=blue,fill=white,inner sep=0pt,minimum size=0.8cm] (24) at (6+3.6+0.4,-0.5) {24};
    \node [blue,circle,draw=blue,fill=white,inner sep=0pt,minimum size=0.8cm] (45) at (6+5.4+0.4,-0.5) {45};
    \path [draw,-] (12) edge (23);
    \path [draw,-] (23) edge (24);
    \path [draw,color=blue] (24) edge (45);
    \node [text width=0.7cm] (T2) at (5,0) {\large $\tree_2:$};
    \node [text width =0.7cm] (E1) at (6.4+0.9,-0.15) {\small 13;2};
    \node [text width =0.7cm] (E1) at (6.4 + 1.8 +0.9,-0.15) {\small 34;2};
    \node [blue, text width =0.7cm] (E1) at (6.4+ 3.6 + 0.9,-0.15) {\small 25;4};
    
    \node [circle,draw=black,fill=white,inner sep=0pt,minimum size=1cm] (13;2) at (0,-3.25) {13;2};
    \node [circle,draw=black,fill=white,inner sep=0pt,minimum size=1cm] (34;2) at (2.2,-3.25) {34;2};
    \node [circle,draw=black,fill=white,inner sep=0pt,minimum size=1cm] (25;4) at (4.4,-3.25) {25;4};
    \path [draw,-] (13;2) edge (34;2);
    \path [draw,-] (34;2) edge (25;4);
    \node [text width=0.7cm] (T3) at (-1.4,-2.75) {\large $\tree_3:$};
    \node [text width =0.7cm] (E1) at (1.025,-2.9) {\small 14;23};
    \node [text width =0.7cm] (E1) at (1.025+2.2,-2.9) {\small 35;24};

    \node [circle,draw=black,fill=white,inner sep=0pt,minimum size=1.2cm](14;23) at (7.5,-3.25) {14;23};
    \node [circle,draw=black,fill=white,inner sep=0pt,minimum size=1.2cm] (35;24) at (10,-3.25) {35;24};
    \path [draw,-] (14;23) edge (35;24);
    \node [text width=0.7cm] (T3) at (6,-2.75) {\large $\tree_4:$};    
    \node [text width =0.7cm] (E1) at (8.6,-2.9) {\small 15;234};
\end{tikzpicture}
  \end{equation*}
    \caption{A graphical representation of the five-dimensional regular vine sequence $\Vine$. The sub-vine $\Vine_f$ induced by $\Vine$ on the node set $A_f = \{2,4,5\}$ for edge $f = (25;4)$ in $\tree_2$ is shown in \textcolor{blue}{blue}.}
    \label{fig:rvineexa}
	\end{figure}

\subsubsection*{Vine telescoping product}
Given are scalars $\gamma_{J} \in (0, \infty)$ for non-empty $J \subseteq \cbr{1,\ldots,5}$ such that $\gamma_j = 1$ for every $j \in \cbr{1,\ldots,5}$. Also, write $\gamma_{I|J} = \gamma_{I \cup J} / \gamma_J$ for disjoint and non-empty $I, J \subset \dset$.
Developing the product formula in Lemma~\ref{lem:Rvine-a} along the given regular vine sequence $\Vine$ gives
\begin{equation}
\label{eq:gam1..5}
\gamma_{12345} = \gamma_{12} \gamma_{23} \gamma_{24} \gamma_{45} \, \cdot \, \frac{\gamma_{13|2} \gamma_{34|2} \gamma_{25|4}}{\gamma_{1|2} \gamma_{3|2} \gamma_{3|2} \gamma_{4|2} \gamma_{2|4} \gamma_{5|2}} \, \cdot \, \frac{\gamma_{14|23} \gamma_{35|24}}{\gamma_{1|23} \gamma_{4|23} \gamma_{3|24} \gamma_{5|24}} \, \cdot \, \frac{\gamma_{15|234}}{\gamma_{1|234} \gamma_{5|234}}.
\end{equation}

As an illustration of Remark~\ref{rem:telescope-f}, consider the edge $f = \{\{2,4\},\{4,5\}\} \in \edges_2$ (or $f = (25;4)$ in its conditioned form). The vine $\Vine_f$ induced by $\Vine$ on the node set $\cunn_f = \{2,4,5\}$ consists of two trees, $\Vine_f = (\tree_{f,1},\tree_{f,2})$, with edge sets $\edges_{f,1} = \{\{2,4\},\{4,5\}\}$ and $\edges_{f,2} = \{\{\{2,4\},\{4,5\}\}\}$; see Fig.~\ref{fig:rvineexa}. The product formula \eqref{eq:telescope-f} yields the identity
\[
	\gamma_{245} = \gamma_{24}\gamma_{25} \cdot \frac{\gamma_{25|4}}{\gamma_{2|4}\gamma_{5|4}}.
\]

\subsubsection*{Tail copula density decomposition}
By Theorem~\ref{thm:xvine}, any 5-variate tail copula density $r$ can be decomposed along the given regular vine sequence $\Vine$ as
\begin{align}
	\nonumber
r(\bx)  & =   r_{12} (x_1,x_2) \, r_{23} (x_2,x_3) \, r_{24} (x_2,x_4) \, r_{45} (x_4,x_5) \\
	\nonumber
& \quad \cdot \, c_{13;2} (R_{1 |2} , R_{3|2}) \, c_{34;2} (R_{3 |2} , R_{4|2})  \, c_{25;4} (R_{2 |4} , R_{5|4}) \\
	\nonumber
& \quad \cdot \, c_{14;23} (R_{1|23}, R_{4|23} ; \bx_{23})  \, c_{35;24} (R_{3|24}, R_{5|24}; \bx_{24}) \\
	\label{eq:r5}
& \quad \cdot \, c_{15;234} (R_{1|234},R_{5|234} ; \bx_{234}),
\end{align} 
where, for brevity, we write $R_{j|D} = R_{j|D} (x_j | \bx_D)$ for $j \in \dset$ and $D \subseteq \dset \setminus \cbr{j}$. The bivariate copula densities in the decomposition are given by formula~\eqref{eq:cIJ} in Sklar's theorm for tail copula densities (Proposition~\ref{lem:expdensity}).
The copulas associated with the edges in tree $\tree_2$ necessarily satisfy the simplifying assumption (see the paragraph right after Definition~\ref{def:simplif}): for instance, the value of $c_{13;2}(u_1,u_3|x_2)$ does not depend on $x_2$, which is why we write $c_{13;2}(R_{1|2}, R_{3|2})$ rather than $c_{13;2}(R_{1|2}, R_{3|2}; x_2)$ on the second line of \eqref{eq:r5}.

\subsubsection*{Recursion}
To completely specify the tail copula density $r$ in terms of bivariate components, we can calculate the arguments $R_{j|D}$ of the pair copulas through the recursive formula of Theorem~\ref{thm:recuni}, as explained also in the paragraph after that theorem. For example, in Eq.~\eqref{eq:r5}, starting from the first argument of the copula $c_{15;234}$ associated with the deepest edge, the conditional distribution function $R_{1|234}$ can be expressed in terms of the conditional distribution functions $R_{1|23}$ and $R_{4|23}$ and the bivariate copula $C_{14;23}$ via
\begin{equation}
\label{eq:R1|234}
	R_{1|234} (x_1 | \bx_{234}) 
	= C_{1|4;23} \rbr{ R_{1|23} (x_1 | \bx_{23}) \mid R_{4|23} (x_4 | \bx_{23}) ; \bx_{23} },
\end{equation}
where $C_{1|4;23}  (u_1 | u_4; \bx_{23}) 
= \int_{0}^{u_1} c_{14;23} (v, u_4;\bx_{23}) \, \diff v$.
Eq.~\eqref{eq:R1|234} is the first identity in Eq.~\eqref{eq:xcdfrecur} applied to $e = \cbr{\cbr{\cbr{1,2},\cbr{2,3}}, \cbr{\cbr{2,3},\cbr{3,4}}}$ with $a_e = 1$, $b_e = 4$, and $\cndg_e = \cbr{2,3}$.
 In turn, the conditional distribution function $R_{1|23}$ can be expressed as (similarly for $R_{4|23}$),
\[
	R_{1|23} (x_1 | \bx_{23}) 
	= C_{1|3;2} \rbr{ 
		R_{1|2} (x_1 |x_2) \mid R_{3|2} (x_3|x_2) 
	},
\]
where $C_{1|3;2} (u_1 | u_3) = \int_{0}^{u_1} c_{13;2} (v, u_3) \, \diff v$ and $R_{1|2} (x_1 | x_2) = \int_{0}^{x_1} r_{12} (t, x_2) \, \diff t$ and $R_{3|2} (x_3 | x_2) = \int_{0}^{x_3} r_{23} (x_2,t) \, \diff t$, i.e., in terms of \emph{bivariate} (tail) copula densities $r_{12}$, $r_{23}$ and $c_{13;2}$ only.

\subsubsection*{Unique-edge property}
Proofs involving regular vine sequences often operate via a recursive argument in which a node $a \in \dset$ and the edges that contain it are peeled off, leaving a regular vine sequence on $d-1$ elements; see for instance the proof of Theorem~\ref{thm:recuni}, which relies on Lemma~\ref{lem:uniqueness} below. Let us illustrate the latter for the regular vine sequence $\Vine$ in Figure~\ref{fig:rvineexa}. The conditioned set of the deepest edge $e = (15;234)$ is $\cndd_e = \cbr{1,5}$. For each of the latter two nodes and for each level $j \in \cbr{1,\ldots,d-1}$, there is a unique edge $e_j \in \edges_j$ that contains this node in its complete union. For instance, for node $1$ these are the edges $(12)$, $(13;2)$, $(14;23)$ and $(15;234)$; moreover, node $1$ belongs to the conditioned set of each of these edges. Removing node $1$ and these four edges leaves a regular vine sequence on the four remaining elements $\cbr{2,\ldots,5}$, which is the sub-vine induced by $\Vine$ on the node set $\cbr{2,\ldots,5}$.

\subsubsection*{Simulation}
Let the random vector $\IMP = (\Imp_1,\ldots,\Imp_5)$ have the inverted multivariate Pareto distribution as in Definition~\ref{def:invmultpar} associated to an X-vine tail copula density $\xdf$ as in \eqref{eq:r5}. By the rejection algorithm in Lemma~2 in \cite{engelke2020graphical}, the task of simulating $\IMP$ can be reduced to the one of simulating the random vector $\IMP^{(j)}$, by definition equal in distribution to $\rbr{\IMP \mid \Imp_j < 1}$, and this for all $j = 1,\ldots,5$. Let $\bW = (W_1,\ldots,W_5)$ be a vector of independent random variables uniformly distributed on $(0, 1)$. We generate $\IMP^{(j)}$ from $\bW$ by applying the inverse Rosenblatt transform \eqref{eq:IMPj} with the five variables $\Imp_1^{(j)},\ldots,\Imp_5^{(j)}$ arranged according to a permutation $\sigma_j$ of $\cbr{1,\ldots,5}$ such that $\sigma_j(1) = j$ and such that for each $k = 2,\ldots,5$, the node set $\cbr{\sigma_j(1),\ldots,\sigma_j(k)}$ is equal to the conditioning set of some edge $e_{j,k-1}$ in tree $\tree_{k-1}$ (Lemma~\ref{lem:perm}).

The permutations $\sigma_j$ can be constructed via a simple algorithm. Given $j \in \dset$:
\begin{enumerate}[(1)]
	\item Put $\sigma_j(1) = j$.
	\item Find $e_{j,1} \in \edges_1$ such that $j \in e_{j,1}$.
	\item Let $\sigma_j(2)$ be the unique element in $e_{j,1} \setminus \cbr{j}$.
	\item For $k = 3,\ldots,d$, do
	\begin{compactenum}[(a)]
		\item Find $e_{j,k-1} \in \edges_{k-1}$ such that $e_{j,k-2} \in e_{j,k-1}$.
		\item Let $\sigma_j(k)$ be the unique element in $\cunn_{e_{j,k-1}} \setminus \cunn_{e_{j,k-2}}$.
	\end{compactenum}
\end{enumerate}
Steps (2) and (4.a) leave freedom for choice, and this is why there may be multiple permutations $\sigma_j$ that satisfy the requirements of Lemma~\ref{lem:perm}. 
For the regular vine tree $\Vine$ in Fig.~\ref{fig:rvineexa}, five possible permutations $\sigma_j$ for $j \in \cbr{1,\ldots,5}$ and the corresponding edge sequences $e_{j,1},\ldots,e_{j,4}$ are listed in Table~\ref{tab:perm}.
Another choice for $\sigma_4$ would have been $(4,2,5,3,1)$, for instance. 

\begin{table}
	\[
	\begin{array}{cc@{\qquad}cccc}
		\toprule
		j & \sigma_j & e_{j,1} & e_{j,2} & e_{j,3} & e_{j,4} \\
		\midrule
		1 & (1,2,3,4,5) & 12 & 13;2 & 14;23 & 15;234 \\
		2 & (2,1,3,4,5) & 12 & 13;2 & 14;23 & 15;234 \\
		3 & (3,2,1,4,5) & 23 & 13;2 & 14;23 & 15;234 \\
		4 & (4,5,2,3,1) & 45 & 25;4 & 35;24 & 15;234 \\
		5 & (5,4,2,3,1) & 45 & 25;4 & 35;24 & 15;234 \\
		\bottomrule
	\end{array}
	\]
	\caption{\label{tab:perm} Five possible permutations $\sigma_j$ and the edge sequences $e_{j,1},\ldots,e_{j,4}$ for $j \in \cbr{1,\ldots,5}$ of Lemma~\ref{lem:perm} for the regular vine tree $\Vine$ in Figure~\ref{fig:rvineexa}.}
\end{table}

Spelling out the inverse Rosenblatt transform \eqref{eq:IMPj} for $j = 4$ and $\sigma_4 = (4,5,2,3,1)$ gives $\Imp_4^{(4)} = W_4$ and
\begin{align*}
	\Imp_5^{(4)} 
	&= \xcdf_{5|4}^{-1}\rbr{W_5|\Imp_4^{(4)}}, &
	\Imp_2^{(4)} 
	&= \xcdf_{2|45}^{-1}\rbr{W_2|\IMP_{45}^{(4)}}, \\
	\Imp_3^{(4)} 
	&= \xcdf_{3|245}^{-1}\rbr{W_3|\IMP_{245}^{(4)}}, &
	\Imp_1^{(4)} 
	&= \xcdf_{1|2345}^{-1}\rbr{W_1|\IMP_{2345}}.
\end{align*}
Thanks to the stated properties of the permutations $\sigma_j$, the conditional quantile functions $\xcdf^{-1}_{k|\sigma_j(\cbr{1,\ldots,k-1})}$ can be computed recursively via \eqref{eq:xqdfrecur} in terms of the bivariate ingredients of the X-vine specification of $\xdf$: for instance, starting from edge $e_{4,2} = (25;4)$ in $\tree_2$, we obtain
\[
	\xcdf_{2|45}^{-1} \rbr{u_2|\bx_{45}}
	= \xcdf_{2|4}^{-1} \rbr{
		C_{2|5;4}^{-1} \rbr{
			u_2 \mid \xcdf_{5|4}(x_5|x_4) 
		} 
		\mid x_4
	},
\]
an expression in terms of the tail copula densities $\xdf_{24}$ and $\xdf_{45}$ of edges in $\tree_1$ and the bivariate copula density $c_{25;4}$ of an edge in $\tree_2$.

\subsubsection*{Structure matrices}
For programming purposes, it is convenient to encode a regular vine sequence in matrix form. Such a regular vine matrix or structure matrix in short is a $d \times d$ upper triangular matrix  $M = (m_{i,j})_{i,j=1}^d$ with elements $m_{i,j}$ in $\dset$ and whose diagonal $(m_{1,1},\ldots,m_{d,d})$ is a permutation of $\dset$; see \citet[Section~5.5]{czado2019analyzing}. A recipe to construct such a matrix is explained in Example~5.11 of \citet{czado2019analyzing}. A recursive algorithm to do so is as follows:
\begin{enumerate}[(1)]
\item Let $e_d$ be the single edge in $\edges_{d-1}$.
\item Choose $m_{d,d} \in \cndd_{e_d}$.
\item For $j \in \cbr{1,\ldots,d-1}$:
	\begin{compactenum}[(a)]
	\item Let $e_j \in \edges_{j}$ be the unique edge that contains $m_{d,d}$ in its complete union and hence in its conditioned set (see Lemma~\ref{lem:uniqueness}).
	\item Let $m_{j,d}$ be the unique element in $\cndd_{e_j} \setminus \cbr{m_{d,d}}$.
	\end{compactenum}
\item Remove node $m_{d,d}$ and the edges $e_1,\ldots,e_{d-1}$ in Step~(3) from the regular vine sequence.
\item Proceed recursively with the resulting regular vine sequence on the $d-1$ elements $\dset \setminus \cbr{m_{d,d}}$ to fill up the $(d-1)$th column $m_{1,d-1},\ldots,m_{d-1,d-1}$; and so on.
\end{enumerate}
Note that step~(3a) involves the unique-edge property described in a previous paragraph.
Step~(2) of the algorithm involves an arbitrary choice, and this is why the same regular vine sequence can be encoded by several structure matrices. Given an element $j \in \dset$, it is always possible to choose the diagonal elements $m_{2,2},\ldots,m_{d,d}$ to be different from $j$, thus ensuring $m_{1,1} = j$. The diagonal of the resulting structure matrix $M$ then forms a permutation $\sigma_j$ meeting the requirements of Lemma~\ref{lem:perm}.

For the regular vine sequence in Fig.~\ref{fig:rvineexa}, some possible structure matrices $M$ are
\begin{equation*}
	\begin{bmatrix}
		1 & 1 & 2 & 2 & 4 \\
		  & 2 & 1 & 3 & 2 \\
		  &   & 3 & 1 & 3 \\
		  &   &   & 4 & 1 \\
		  &   &   &   & 5
	\end{bmatrix},
	\begin{bmatrix}
		2 & 2 & 2 & 2 & 4 \\
		  & 1 & 1 & 3 & 2 \\
		  &   & 3 & 1 & 3 \\
		  &   &   & 4 & 1 \\
		  &   &   &   & 5
	\end{bmatrix},
	\begin{bmatrix}
		3 & 3 & 2 & 2 & 4 \\
		  & 2 & 3 & 3 & 2 \\
		  &   & 1 & 1 & 3 \\
		  &   &   & 4 & 1 \\
		  &   &   &   & 5
	\end{bmatrix},
	\begin{bmatrix}
		4 & 4 & 4 & 2 & 2 \\
		  & 5 & 5 & 4 & 3 \\
		  &   & 2 & 5 & 4 \\
		  &   &   & 3 & 5 \\
		  &   &   &   & 1
	\end{bmatrix},
	\begin{bmatrix}
		5 & 5 & 4 & 2 & 2 \\
		  & 4 & 5 & 4 & 3 \\
		  &   & 2 & 5 & 4 \\
		  &   &   & 3 & 5 \\
		  &   &   &   & 1
	\end{bmatrix}.
\end{equation*}
The diagonals of these matrices are equal to the permutations $\sigma_j$ in Table~\ref{tab:perm}.

\subsubsection*{Truncated X-vine models}
To encode a truncated regular vine sequence $\Vine = (\tree_1,\ldots,\tree_q)$ on $d$ elements with truncation level $q \in \cbr{1,\ldots,d-2}$ (Definition~\ref{def:truncvine}), one possible way would be to complete it to it a (full) regular vine sequence $\Vine' = (\tree_1,\ldots,\tree_{d-1})$ and use a structure matrix of the latter, imputing independence copulas for the resulting X-vine model. This would be wasteful, however, especially if $d$ is large and $q$ is small. Instead, we avoid such arbitrary and unnecessary completions and work with truncated structure matrices instead. The form of such a matrix is the same as of a structure matrix in the previous paragraph, but with zeroes on the unspecified elements above the diagonal.

Consider for instance the truncated regular vine sequence $(\tree_1,\tree_2)$ on $d=5$ elements for $\tree_1$ and $\tree_2$ as in Fig.~\ref{fig:5d_Xvine}. One possible structure matrix would be
\[
	\begin{bmatrix}
	1 	& 1 & 2 & 2 & 4 \\
		& 2 & 1 & 3 & 2 \\
		&   & 3 & 0 & 0 \\
		&   &   & 4 & 0 \\
		&   &   &   & 5
	\end{bmatrix},
\]
which is the first of the five structure matrices in the previous paragraph but with zeroes on positions that correspond to edges in trees $\tree_j$ with $j > q = 2$. Such a matrix can be directly obtained from the truncated regular vine sequence by an algorithm similar to the one in the previous paragraph:
\begin{enumerate}[(1)]
\item Choose $m_{d,d} \in \dset$ such that there is only a single edge $e_{q} \in \edges_q$ with $m_{d,d} \in \cunn_{e_q}$ and hence $m_{d,d} \in \cndd_{e_q}$.
\item For $j \in \cbr{1,\ldots,q}$:
	\begin{compactenum}[(a)]
	\item Let $e_j \in \edges_j$ be the uniqe edge that contains $m_{d,d}$ in its complete union and hence in its conditioned set.
	\item Let $m_{j,d}$ be the unique element in $\cndd_{e_j} \setminus \cbr{m_{d,d}}$.
	\end{compactenum}
\item Remove node $m_{d,d}$ and the edges $e_1,\ldots,e_q$ from Step~(2) from the truncated regular vine sequence.
\item Proceed recursively with the resulting regular vine sequence, truncated or not.
\end{enumerate}

\subsubsection*{Recursive estimation}
We illustrate the recursive nature of the sequential estimation procedure in Section~\ref{sec:estimation1} for the regular vine sequence in Fig.~\ref{fig:rvineexa}.
For the edge $e = (25;4)$ in $\edges_2$, we consider pseudo-observations
\begin{align*}
	\hU_{i,2;4}
	&= \xcdf_{2|4} \rbr{
		\hZ_{i,2} \mid \hZ_{i,4} ; \; \htheta_{24}
	}, &
	\hU_{i,5;4}
	&= \xcdf_{5|4} \rbr{
		\hZ_{i,5} \mid \hZ_{i,4} ; \; \htheta_{45}
	}
\end{align*}
for $i \in N_4$, i.e., for those observations with a large value in variable $j=4$. Here, the tail copula parameters $\htheta_{24}$ and $\htheta_{45}$ have already been estimated in a previous step based on the pairs $(\hZ_{i,2}, \hZ_{i,4})$ for $i \in N_2 \cup N_4$ and the pairs $(\hZ_{i,4}, \hZ_{i,5})$ for $i \in N_4 \cup N_5$, respectively, using the averaging procedure based on the pseudo-likelihoods \eqref{eq:lik}. For instance, we have $\htheta_{24} = (\htheta_{24}^{(2)}+\htheta_{24}^{(4)})/2$ where $\htheta_{24}^{(2)}$ is estimated based on the pairs $(\hZ_{i,2},\hZ_{i,4})$ for $i \in N_2$, while $\htheta_{24}^{(4)}$ is based on the same pairs but for $i \in N_4$.

Next, consider the edge $e = (35;24)$ in $\edges_3$, which itself joins the edges $f = (34;2)$ and $g = (25;4)$ in $\edges_2$. In view of Proposition~\ref{prop:Z2c}(ii) with $I = \cbr{3,5}$ and $J = \cbr{2,4}$, we define pseudo-observations $(\hU_{i,3;24}, \hU_{i,5;24})$ for $i \in N_{24} = N_2 \cap N_4$ from $c_{3,5;24}$ by
\begin{align*}
	\hU_{i,3;24} 
	&= \xcdf_{3|24} \rbr{
		\hZ_{i,3} \mid \hbZ_{i,24} ; \;
		\htheta(\edges_{1:2})
	} \\
	&= C_{3|4;2} \rbr{ 
		\xcdf_{3|2} \rbr{
			\hZ_{i,3} \mid \hZ_{i,2} ; \; 
			\htheta_{23}
		} \mid 
		\xcdf_{4|2} \rbr{
			\hZ_{i,4} \mid \hZ_{i,2} ; \;
			\htheta_{24}
		} ; \; \htheta_{34;2}
	}
\end{align*}
and
\begin{align*}
	\hU_{i,5;24} 
	&= \xcdf_{5|24} \rbr{
		\hZ_{i,5} \mid \hbZ_{i,24} ; \;
		\htheta(\edges_{1:2})
	} \\
	&= C_{5|2;4} \rbr{ 
		\xcdf_{5|4} \rbr{
			\hZ_{i,5} \mid \hZ_{i,4} ; \;
			\htheta_{45}
		} \mid 
		\xcdf_{2|4} \rbr{
			\hZ_{i,2} \mid \hZ_{i,4} ; \;
			\htheta_{24}
		} ; \; \htheta_{25;4}
	}.
\end{align*}
Note that for the pair $\cbr{2,4}$, both conditional distribution functions $\xcdf_{2|4}$ and $\xcdf_{4|2}$ are required.

\section{Proofs}
\label{app:proofs}

\subsection{Proofs for Section~\ref{sec:background}}
\label{app:proofs:background}

\chng{Recall $\EE = (0, \infty]^d \setminus \{\binfty\}$ for some dimension $d \ge 1$.
\begin{lemma}\label{lem:allRok}
Any Borel measure $\xcdf$ on $\EE$ satisfying \eqref{eq:margin} and \eqref{eq:Lambdahomo} is a tail copula measure, i.e., is the limit in \eqref{eq:C2Lambda} for some copula $C$. 
\end{lemma}
\begin{proof}[Proof of Lemma \ref{lem:allRok}\pnt]
Given $\xcdf$, define the probability measure $P$ on $\LL = \cbr{\bx \in (0,\infty]^d : \min \bx < 1}$ by $P(B) = \xcdf(B \cap \LL)/\xcdf(\LL)$ for Borel sets $B$. (Note that $P$ is an inverted multivariate Pareto distribution as in Definition~\ref{def:invmultpar}, although we do not use that fact in this proof.) For $j \in \dset$ and $x_j \in [0, 1]$, we have, by \eqref{eq:margin},
\[
	F_j(x_j) 
	= P(\{ \by \in \EE : y_j \le x_j \}) 
	= x_j / \xcdf(\LL).
\]
Note that $\xcdf(\LL) \ge 1$ by \eqref{eq:margin}.
For $u_j \in [0, 1/\xcdf(\LL)]$, we have therefore
\[
	F_j^{-1}(u_j) = u_j \, \xcdf(\LL).
\] 
Let $C$ be the copula of $P$. For $\bu \in [0, 1/\xcdf(\LL)]^d$, we find, by Sklar's theorem,
\[
	C(\bu) 
	= \xcdf(\bu \, \xcdf(\LL)) / \xcdf(\LL)
	= \xcdf(\bu)
\]
in view of the homogeneity of $\xcdf$ in \eqref{eq:Lambdahomo}. The same equation then implies that $t^{-1} \cdot C(t \bx) = \xcdf(\bx)$ for all $\bx \in (0, \infty)^d$ and all sufficiently small $t > 0$. If some (but not all) coordinates of $\bx$ are infinity, we repeat the same argument on the appropriate lower-dimensional margins of $C$. We conclude that \eqref{eq:C2Lambda} holds, so that $\xcdf$ is indeed a tail copula measure.
\end{proof}}

\chng{By \eqref{eq:Lambdahomo} and \eqref{eq:Lam2nu}, the exponent measure $\Lambda$ is homogeneous of order $-1$ and thus determined by the \emph{angular measure} $H$ on the unit simplex $\Sd = \cbr{ \bx \in [0, 1]^d : \norm{\bx}_1 = 1}$ via
\[
\expmeas\rbr{\cbr{ 
		\bx \in [0, \infty)^d : \; 
		\norm{\bx}_1 > t, \;
		\bx / \norm{\bx}_1 \in B
	}
} = t^{-1} H(B),
\]
for $t > 0$ and Borel sets $B \subseteq \Sd$, with $\norm{\bx}_1 = \sum_{j=1}^d |x_j|$ for $\bx \in \Rd$.
Under the above smoothness assumption, the angular measure $H$ is concentrated on the relative interior of the unit simplex $\Sd$ and its density $h$ satisfies $h(\bw) = \nupdf(\bw)$ for $\bw \in \Sd$ \citep{dombry2016asymptotic}.
In Lemma~\ref{lem:lambdah} below, we give a proof of this fact and link the angular measure density $h$ to the tail copula density $\xdf$.}

\begin{lemma}\label{lem:lambdah}
	We have $h(\bw) = \nupdf(\bw)$ for $\bw \in \Sd$ and thus also
	\[
	\xdf(\bx) 
	= \nupdf(1/\bx) \prod_{j=1}^d x_j^{-2} 
	= \norm{1/\bx}_1^{-d-1} \cdot h \rbr{\frac{1/\bx}{\norm{1/\bx}_1}} 
	\prod_{j=1}^d x_j^{-2},
	\qquad \bx \in (0, \infty)^d.
	\]
\end{lemma}

\begin{proof}[Proof of Lemma \ref{lem:lambdah}\pnt]
Let $V(\by) = \expmeas([\bzero,\binfty) \setminus [\bzero,\by])$ for $\by \in (0, \infty)^d$, the exponent measure function. For $V_{1:d}(\by) := \partial^d V/(\partial y_1 \cdots \partial y_d)$, \citet{coles1991modelling} show that 
\[
	V_{1:d}(\by) 
	= - \norm{\by}_1^{-d-1} h \rbr{ \by/\norm{\by}_1 }.
\]
On the other hand, by the inclusion-exclusion formula we also have
\begin{align*}
	V_{1:d}(\by)
	&= \sum_{\emptyset \ne J \subseteq \dset}
	(-1)^{|J|-1}
	\frac{\partial^d \nu(\{\bz \ge \bzero : \forall j \in J, z_j > y_j \})}{\partial y_1 \cdots \partial y_d} \\
	&= (-1)^{d-1} (-1)^d \nupdf(\by)
	= - \nupdf(\by).
\end{align*}
We find 
\[
	\nupdf(\by)
	= \norm{\by}_1^{-d-1} h \rbr{ \by/\norm{\by}_1 },
	\qquad \by \in (0, \infty)^d.
\]
If $\norm{\by}_1 = 1$, this specializes to $\nupdf(\by) = h(\by)$. The connection with the tail copula density $\xdf$ follows from \eqref{eq:nupdf}, which can be inverted to
\[
	\xdf(\bx)
	= \nupdf(1/\bx) \prod_{j=1}^d x_j^{-2},
	\qquad \bx \in (0, \infty)^d.
	\qedhere
\]
\end{proof}

\paragraph{Properties of regular vine sequences.}
\chng{Recall Definition~\ref{def:ucd} on the complete union $\cunn_e$, conditioning set $\cndg_e$ and conditioned sets $\cndd_e = \cndd_{e,a} \cup \cndd_{e,b}$ of an edge $e \in \bigcup_{j=2}^d \edges_j$ in a regular vine $\Vine = (\tree_1,\ldots,\tree_{d-1})$ on $d$ elements. If $j \ge 2$, then $\cndg_a \cup \cndg_b \subseteq \cndg_e$; indeed, by proximity, $a \cap b = \cbr{k}$ with $k \in \nodes_{j-1}$, and $\cndg_a \subset \cunn_k \subset \cunn_a \cap \cunn_b = \cndg_e$; similarly $\cndg_b \subset \cndg_e$. Further, the conditioned sets are related by $\cndd_{e,a} = \cndd_e \cap \cndd_a$ and $\cndd_{e,b} = \cndd_e \cap \cndd_b$. To see this, note that, by proximity, we have $a = \cbr{a_1, k}$ and $b = \cbr{b_1, k}$ for three distinct elements $a_1, b_1, k \in \nodes_{j-1}$, and 
$\cndd_{e,a} 
= \cunn_a \setminus \cunn_b 
= \rbr{\cunn_{a_1} \cup \cunn_{k}} \setminus \rbr{\cunn_{b_1} \cup \cunn_{k}}
\subseteq \cunn_{a_1} \setminus \cunn_{k} 
= \cndd_{a,a_1}$. 
Since both the left- and right-hand side are singletons, we must have $\cndd_{e,a} = \cndd_{a,a_1} = \cndd_e \cap \cndd_a$; similarly $\cndd_{e,b} = \cndd_{b,b_1} = \cndd_e \cap \cndd_b$.}

The following lemma provides the basis for recursive arguments concerning regular vines.

\begin{lemma}
	\label{lem:uniqueness}	
	Let $\Vine = (\tree_j)_{j=1}^{d-1}$ be a regular vine sequence on $d$ elements, with $d \ge 2$. Let $e$ be the unique edge in $\edges_{d-1}$. For every $a \in \cndd_e$ and every $j \in \cbr{1,\ldots,d-1}$, there is a unique edge $e_j \in \edges_{j}$ such that $a \in \cunn_{e_j}$, and then in fact $a \in \cndd_{e_j}$. 
\end{lemma}

\begin{proof}[Proof of Lemma~\ref{lem:uniqueness}\pnt]
	The existence is clear:	at every level $j \in \cbr{1,\ldots,d-1}$, the union $\bigcup_{e' \in \edges_j} \cunn_{e'}$ of the complete unions is equal to $\dset$. Indeed, this holds for $\edges_1$ since $\tree_1$ is a tree on $\dset$ and it holds for the other edge sets by the recursive definition of the complete unions.
	
	The uniqueness is specific to the two nodes in the conditioned set of the single edge $e$ in $\edges_{d-1}$. Suppose that for some $a \in \cndd_e$ there is an edge set $\edges_j$, with $j \le d-2$, for which there are two different edges, say $f$ and $g$, that both contain $a$ in their complete unions. If $f$ and $g$ are linked up by some edge $h$ in the next tree $\tree_{j+1}$, then $a \in \cndg_h$ and thus, eventually $a \in \cndg_e$ (see the paragraph following Definition~\ref{def:ucd}), a contradiction. Otherwise, there exist two different edges $f'$ and $g'$ in $\edges_{j+1}$ that contain $a$ in their complete unions, and the argument can be repeated, if needed until the last tree has been reached. Finally, since the conditioning sets are nested and since $a$ does not belong to the conditioning set of $e$, it does not belong to the conditioning set of another edge either.
\end{proof}

\begin{proof}[Proof of Lemma \ref{lem:Rvine-a}\pnt]
Identity \eqref{eq:Rvine-a-cond} follows from
\[
	\frac{\gamma_{\cndg_e} \cdot \gamma_{\cunn_e}}%
	{\gamma_{\cunn_a} \cdot \gamma_{\cunn_b}}
	=
	\frac{\dfrac{\gamma_{\cunn_e}}{\gamma_{\cndg_e}}}%
	{\dfrac{\gamma_{\cunn_a}}{\gamma_{\cndg_e}} \cdot \dfrac{\gamma_{\cunn_b}}{\gamma_{\cndg_e}}}
\]
together with $\cunn_a = \cbr{a_e} \cup \cndg_e$, $\cunn_b = \cbr{b_e} \cup \cndg_e$, and $\cunn_e = \cbr{a_e,b_e} \cup \cndg_e$.

The proof of \eqref{eq:Rvine-a} is by induction on $d$. For $d = 3$, we can relabel the indices so that the first edge set $\edges_1 = \cbr{ \cbr{1,2}, \cbr{2,3} }$ while $\edges_2$ contains the single edge $(13;2)$. The formula then states that
	\[
	\gamma_{123} = \gamma_{12} \cdot \gamma_{23} \cdot \frac{\gamma_2 \cdot \gamma_{123}}{\gamma_{12} \cdot \gamma_{23}},
	\]
which is obviously true since $\gamma_2 = 1$ by assumption.
	
Let $d \ge 4$ and suppose the formula is true for regular vine sequences of dimension up to $d-1$. 
Let $e = (a_e,b_e;\cndg_e)$ be the single edge in the deepest tree $\tree_{d-1} = (\nodes_{d-1}, \edges_{d-1})$ of the vine. By Lemma~\ref{lem:uniqueness}, there exist unique edges $e_j \in \edges_j$ for $j \in \cbr{1,\ldots,d-1}$ with $e_{d-1} = e$ such that
\[
	a_e \in e_1 \in e_2 \in \ldots \in e_{d-1}.
\]
Moreover, the node $a_e$ belongs to the conditioned set of each edge $e_j$, so we can write $\cndd_{e_j} = \cbr{a_e, b_j}$ for $j \in \cbr{1,\ldots,d-1}$ and some $b_j \in \dset \setminus \cbr{a_e}$.
\chng{In particular, $a_e$ is a leaf node in $\tree_1$.}
For $j \in \{2,\ldots,d-1\}$, we have $e_j = \cbr{e_{j-1}, f_{j-1}}$ for some edge $f_{j-1} \in \edges_{j-1}$.
	
\chng{Thanks to the stated properties of $a_e$,} the regular vine sequence $\Vine$ induces a 
regular vine sequence $\bar{\Vine}$ on $\bar{\nodes}_1 = \dset \setminus \cbr{a_e}$ with edge sets $\bar{E}_j = E_j \setminus \cbr{e_j}$ for $j \in \cbr{1,\ldots,d-2}$; this is the sub-vine induced by $\Vine$ on $\dset \setminus \cbr{a_e}$. 
By the induction hypothesis, we have
\[
	\gamma_{\dset \setminus \cbr{a_e}}
	= \prod_{\bar{e} \in \bar{\edges}_1} \gamma_{\bar{e}} \cdot
	\prod_{j=2}^{d-2} \prod_{\bar{e} = \cbr{\bar{f}, \bar{g}} \in \bar{\edges}_j}
	\frac{\gamma_{\cndg_{\bar{e}}} \cdot \gamma_{\cunn_{\bar{e}}}}{\gamma_{\cunn_{\bar{f}}} \cdot \gamma_{\cunn_{\bar{g}}}}.
\]
The product of the additional factors in \eqref{eq:Rvine-a} is
\[
	\gamma_{e_1} 
	\cdot 
	\prod_{j=2}^{d-1} 
	\frac{\gamma_{\cndg_{e_j}} \cdot \gamma_{\cunn_{e_j}}}%
	{\gamma_{\cunn_{e_{j-1}}} \cdot \gamma_{\cunn_{f_{j-1}}}}
	=
	\rbr{ 
		\gamma_{e_1} 
		\cdot 
		\prod_{j=2}^{d-1}
		\frac{\gamma_{\cunn_{e_j}}}{\gamma_{\cunn_{e_{j-1}}}}
	} 
	\cdot
	\prod_{j=2}^{d-1} 
	\frac{\gamma_{\cndg_{e_j}}}{\gamma_{\cunn_{f_{j-1}}}}.
\]
It is sufficient to show that this product is equal to $\gamma_{\dset} / \gamma_{\dset \setminus \cbr{a_e}}$. On the one hand, since $\cunn_{e_1} = e_1$ and $\cunn_{e_{d-1}} = \dset$, we have the telescoping product
\[
	\gamma_{e_1} \cdot
	\prod_{j=2}^{d-1} \frac{\gamma_{\cunn_{e_j}}}{\gamma_{\cunn_{e_{j-1}}}}
	= \gamma_{\cunn_{e_{d-1}}} 
	= \gamma_{\dset}.
\]
On the other hand, we have
\[
	\prod_{j=2}^{d-1} 
	\frac{\gamma_{\cndg_{e_j}}}{\gamma_{\cunn_{f_{j-1}}}}
	=
	\gamma_{\cndg_{e_2}}
	\cdot
	\frac{\prod_{j=3}^{d-1} \gamma_{\cndg_{e_j}}}{\prod_{j=2}^{d-2} \gamma_{\cunn_{f_{j-1}}}}
	\cdot
	\frac{1}{\gamma_{\cunn_{f_{d-2}}}}.
\]
The first factor is $\gamma_{\cndg_{e_2}} = 1$ since $\cndg_{e_2}$ is a singleton, while the last factor is $1/\gamma_{\cunn_{f_{d-2}}} = 1/\gamma_{\dset \setminus \cbr{a_e}}$. It remains to show that the middle factor is equal to $1$ too. To do so, write it as
\[
	\prod_{j=3}^{d-1} \frac{\gamma_{\cndg_{e_j}}}{\gamma_{\cunn_{f_{j-2}}}},
\]
which is indeed equal to $1$, since $\cndg_{e_j} = \cunn_{f_{j-2}}$ for $j \in \cbr{3,\ldots,d-1}$.
Indeed, to see the latter identity, write 
\[
	e_j 
	= \cbr{e_{j-1}, f_{j-1}} 
	= \cbr{\cbr{e_{j-2}, f_{j-2}}, f_{j-1}}. 
\]
Recall $\cndd_{e_j} = \cbr{a_e, b_j}$. Then $\cndg_{e_j} = \cunn_{e_j} \setminus \cbr{a_e, b_j}$ as well as $\cunn_{e_j} = \cunn_{e_{j-1}} \cup \cbr{b_j}$ and $\cunn_{j-1} = \cunn_{f_{j-2}} \cup \cbr{a}$, the unions being disjoint. It follows that $\cunn_{f_{j-2}} = \cunn_{e_j} \setminus \cbr{a, b_j} = \cndg_{e_j}$. The proof is complete.
\end{proof}

\begin{remark}[Vine telescoping product on sub-vines]
	\label{rem:telescope-f}
	Given an edge $f \in \edges_j$ for some $j \in \cbr{1,\ldots,d-1}$ in a regular vine sequence $\Vine$ on $\dset$, we can consider the sub-vine $\Vine_f$ induced by $\Vine$ on the node set $\cunn_f$ in the natural way: we have $\Vine_f = \rbr{\tree_{f,i} = \rbr{\nodes_{f,i}, \edges_{f,i}}}_{i=1}^j$ where $\nodes_{f,1} = \cunn_f$ and $\edges_{f,i} = \cbr{e \in \edges_i : \cunn_e \subseteq \cunn_f}$ for $i \in \cbr{1,\ldots,j}$. Applying Lemma~\ref{lem:Rvine-a} to $\Vine_f$, we obtain the seemingly more general formula 
		\begin{equation}
			\label{eq:telescope-f}
			\gamma_{\cunn_f} 
			= \prod_{e \in \edges_1: \cunn_e \subseteq \cunn_f} \gamma_e \cdot
			\prod_{i=2}^{j} \prod_{e \in \edges_i: \cunn_e \subseteq \cunn_f} \frac{\gamma_{\{a_e,b_e\}|\cndg_e}}{\gamma_{a_e|\cndg_e} \cdot \gamma_{b_e|\cndg_e}}.
	\end{equation}
\end{remark} 

\chng{\begin{remark}[Vine telescoping product when $\gamma_j$ is not necessarily $1$.]\label{rem:Rvine-c}
	Another generalisation of the vine telescoping product in Lemma~\ref{lem:Rvine-a} concerns removing the assumption that $\gamma_j = 1$ for all $j \in \dset$. If $\gamma_J \in (0, \infty)$ for every non-empty $J \subseteq \dset$ but the scalars $\gamma_j$ are not necessarily equal to one, the vine telescoping product in Lemma~\ref{lem:Rvine-a} becomes
	\begin{equation}
	\label{eq:Rvine-c}
		\gamma_{\dset} =
		\prod_{j \in \dset} \gamma_j^{1-\deg(j)}
		\cdot \prod_{e \in \edges_1} \gamma_e
		\cdot \prod_{j=2}^{d-1} \prod_{e = \cbr{a,b} \in \edges_j}
		\frac{\gamma_{\cndg_e} \cdot \gamma_{\cunn_e}}{\gamma_{\cunn_a} \cdot \gamma_{\cunn_b}},
	\end{equation}
	where $\deg(j)$ is the degree of $j$ in $\tree_1$, that is, the number of edges in that tree to which it belongs, or, equivalently, its number of neighbours in that tree. To obtain \eqref{eq:Rvine-c}, apply \eqref{eq:Rvine-a} to the coefficients $\gamma_J' = \gamma_J / \prod_{j \in J} \gamma_j$ for non-empty $J \subseteq \dset$ and simplify.
\end{remark}}

\subsection{Proofs for Section~\ref{sec:Sklar}}

\begin{proof}[Proof of Proposition~\ref{lem:expdensity}\pnt]
The fact that $\xdf_{I|J}(\point|\bx_J)$ is a probability density function on $(0, \infty)^I$ follows from the fact that integrating a tail copula density over certain coordinates while keeping the other coordinates fixed produces the marginal tail copula density in the remaining coordinates, see \eqref{eq:lambdaJ}. The statements about the univariate margins follow in the same way. Finally, the expression of the copula density follows from Sklar's theorem for probability densities applied to $\xdf_{I|J}(\point|\bx_J)$.
\end{proof}

\begin{proof}[Proof of Lemma \ref{lem:homo}\pnt]
By the homogeneity of tail copula densities as stated in \eqref{eq:lambdahomo}, we have
\begin{align*}
	\xdf_{I|J}(\bx_I|t\bx_J)
	= \frac{\xdf_{I\cup J}(\bx_I, t\bx_J)}{\xdf_J(t\bx_J)}
	= \frac{t^{1-|I \cup J|} \cdot \xdf_{I \cup J}(t^{-1} \bx_I, \bx_J)}%
	{t^{1-|J|} \cdot \xdf_J(\bx_J)} 
	= t^{-|I|} \cdot \xdf_{I|J}(t^{-1}\bx_I|\bx_J).
\end{align*}
If the random vector $\bm{\xi}_I$ has probability density function $\bx_I \mapsto \xdf_{I|J}(\bx_I|\bx_J)$, then the probability density function of the random vector $t\bm{\xi}_I$ is indeed $\bx_I \mapsto t^{-|I|} \cdot \xdf_{I|J}(t^{-1}\bx_I|\bx_J) = \xdf_{I|J}(\bx_I|t\bx_J)$. The statements on the cumulative distribution and quantile functions follow easily. Finally, the copula of $t\bm{\xi}_I$ does not depend on $t > 0$, yielding \eqref{eq:cIJhomo}.
\end{proof}

\begin{proof}[Proof of Proposition \ref{lem:Sklar}\pnt]
	First, we show that $\xdf$ is a tail copula density. To this end, we need to show two properties: the normalization constraint~\eqref{eq:margin:pdf} and the homogeneity relation~\eqref{eq:lambdahomo}.
	
	To show the homogeneity relation~\eqref{eq:lambdahomo}, let $t \in (0, \infty)$, and let $\bx \in (0, \infty)^d$ be such that $\xdf_J(\bx_J) > 0$. Then
	\begin{align*}
		\xdf(t\bx)
		&= \xdf_J(t\bx_J) \cdot 
		\prod_{i \in I} \xdf_{i|J}(tx_i|t\bx_J) \cdot
		c_{I;J} \rbr{ \xcdf_{i|J}(tx_i|t\bx_J), i \in I ; t\bx_J } \\
		&= \rbr{\xdf_J(t\bx_J)}^{1-|I|} \cdot 
		\prod_{i \in I} \xdf_{i \cup J}(tx_i, t\bx_J) \cdot
		c_{I;J} \rbr{ \xcdf_{i|J}(tx_i|t\bx_J), i \in I ; t\bx_J }.
	\end{align*}
	By Lemma~\ref{lem:homo}, the last factor does not depend on $t$ at all. By \eqref{eq:lambdahomo}, the scalar $t$ can be taken out of the tail copula densities yielding a factor $(t^{1-|J|})^{1-|I|} \cdot (t^{1 - (1+|J|)})^{|I|}$. The exponent of $t$ is, as required,
	\[
	(1-|J|)(1-|I|) - |I| \cdot |J| 
	= 1 - |I| - |J| 
	= 1 - d.
	\]
	
	To show the normalization constraint~\eqref{eq:margin:pdf}, let $k \in \cbr{1,\ldots,d}$ and $x_k \in (0, \infty)$. We need to show that $\int_{(0, \infty)^{d-1}} \xdf(\bx) \, \diff \bx_{\setminus k} = 1$, where the integral is over all coordinates except for the $k$th one. There are two cases: $k \in I$ or $k \in J$.
	
	First, if $k \in I$, the integral of interest is
	\begin{equation*}
		\int_{(0, \infty)^{J}} 
		\xdf_{J}(\bx_J)
		\cbr{ 
    		\int_{(0, \infty)^{I \setminus k}} 
    		\prod_{i \in I} \xdf_{i|J}(x_i|\bx_J) \cdot
    		c_{I;J} \left( \xcdf_{i|J}(x_i|\bx_J), i \in I ; \bx_J \right)
    		\diff \bx_{I \setminus k} 
    		}
		\diff \bx_{J}.
	\end{equation*}
	From the inner integral, the factor $\xdf_{k|J}(x_k|\bx_J)$ can be taken out. The function that remains in the inner integral has integral equal to one: substitute $u_i = \xcdf_{i|J}(x_i|\bx_J)$ for $i \in I \setminus k$ and use the fact that the copula density $c_{I;J}$ has uniform margins. The integral simplifies to
	\[
	\int_{(0, \infty)^{J}}
	\xdf_{J}(\bx_J) \cdot \xdf_{k|J}(x_k|\bx_J) \,
	\diff \bx_J
	= 
	\int_{(0, \infty)^{J}}
	\xdf_{k \cup J}(x_k, \bx_J) \,
	\diff \bx_J
	=
	1,
	\]
	as $\xdf_{k \cup J}$ is a tail copula density.
	
	Second, if $k \in J$, the integral of interest is
	\begin{equation*}
		\int_{(0, \infty)^{J \setminus k}} 
		\xdf_{J}(\bx_J)
		\cbr{
    		\int_{(0, \infty)^I} 
    		\prod_{i \in I} \xdf_{i|J}(x_i|\bx_J) \cdot
    		c_{I;J} \rbr{ \xcdf_{i|J}(x_i|\bx_J), i \in I ; \bx_J }
    		\diff \bx_I 
    		}
		\diff \bx_{J \setminus k}.
	\end{equation*}
	The inner integral is equal to one, as the integrand is a probability density function, by Sklar's theorem. The integral of interest thus simplifies to
	\[
	\int_{(0, \infty)^{J \setminus k}} 
	\xdf_{J}(\bx_J) \, 
	\diff \bx_{J \setminus k}
	= 1,
	\]
	since $\xdf_{J}$ is a tail copula density.
	
	The proof that $\xdf$ has the required margins $\xdf_{i \cup  J}$ for $i \in I$ is similar to the case $k \in I$ in the previous part of the proof.
	
	Finally, statement \eqref{eq:cIJ} easily follows by applying Sklar's theorem to $\xdf_{I|J}(\point|\bx_J)$.
\end{proof}

\begin{remark}
In the special case that $J$ is a singleton, $J = \cbr{j}$, we have $I = \dset \setminus j$, the $(d-1)$ tail copula densities $(\xdf_{i,j})_{i\in \dset \setminus j}$ are bivariate, and the common-margin assumption is automatic by~\eqref{eq:margin}; in particular, $\xdf_{i|j}(x_i|x_j) = \xdf_{i,j}(x_i,x_j)$. Further, the scale-invariance property \eqref{eq:cIJhomo} enforces the simplifying assumption, so that the construction~\eqref{eq:Sklarconstructive} becomes
\[
\xdf(\bx) := c \rbr{ \xcdf_{i|j}(x_i|x_j), i \in \dset \setminus j } \cdot \prod_{i \in \dset \setminus j} \xdf_{i,j}(x_i,x_j),
\]
with $c$ an arbitrary $(d-1)$-variate copula density. Any $(d-1)$-tuple of bivariate tail copula densities can thus be combined along any $(d-1)$-variate copula density to form a $d$-variate tail copula density. 
The X-vine construction in Section~\ref{sec:vine} is in the same spirit but without requiring a common pivot variable~$j$ and with a further break-down of $c$ into pair copulas.
\end{remark}

\subsection{Proofs for Section~\ref{sec:parametric}}
\label{sec:proof:param}

\begin{proof}[Proof of Proposition~\ref{prop:scdiri}\pnt]
	Since the lower dimensional margins of the scaled extremal Dirichlet model \eqref{eq:xdfscdiri} are of the same form as $\xdf$ itself, we find for the conditional probability density in the variables $\bx_I \in (0, \infty)^I$ the expressions
	\[
	r_{I|J}(\bx_I|\bx_J)
	= K \cdot
	\frac{\cbr{\sum_{i \in I} c(\alpha_i,\rho)^{1/\rho} x_i^{-1/\rho} + s(\bx_J)}^{-\rho-\sum_{k\in I \cup J} \alpha_k}}{s(\bx_J)^{-\rho-\sum_{j\in J}\alpha_J}} \cdot
	\prod_{i \in I} x_i^{-\alpha_i/\rho-1},
	\]
	where $K$ is a normalising constant that depends on the parameters $\alpha_1,\ldots,\alpha_d,\rho$ and the sets $I$ and $J$ but not on the arguments $\bx_{I \cup J}$, while $s(\bx_J) = \sum_{j \in J} c(\alpha_j,\rho)^{1/\rho} x_j^{-1/\rho}$. For a well-chosen scalar $t > 0$ depending on $\bx_J$, we have $s(t \bx_J) = 1$. By Lemma~\ref{lem:homo}, a scalar multiplication of the vector of conditioning variables $\bx_J$ does not affect the copula density $c_{I;J}(\point;\bx_J)$, i.e., $c_{I;J}(\point;t\bx_J) = c_{I;J}(\point;\bx_J)$. We obtain that $c_{I;J}(\point;\bx_J) \equiv c_{I;J}(\point)$ is equal to the copula density of the $|I|$-variate probability density function proportional to
	\[
	\bx_I \mapsto 
	\prod_{i \in I} x_i^{-\alpha_i/\rho-1} \cdot
	\cbr{\sum_{i \in I} c(\alpha_i,\rho)^{1/\rho} x_i^{-1/\rho} + 1}^{-\rho-\sum_{k\in I \cup J} \alpha_k}.
	\]
	This fact already implies that the simplifying assumption holds.
	
	To identify $c_{I;J}$, we can incorporate the constants $c(\alpha_i, \rho)$ into the variables $x_i$ by marginal scale transformations without changing the copula. Hence, $c_{I;J}$ is equal to the copula density of the $|I|$-variate probability density function proportional to
\begin{equation}
	\label{eq:pdfIJscdiri}
	\bx_I \mapsto 
	\prod_{i \in I} x_i^{-\alpha_i/\rho-1} \cdot
	\rbr{\sum_{i \in I} x_i^{-1/\rho} + 1}^{-\rho-\sum_{k\in I \cup J} \alpha_k}.
\end{equation}

Consider two cases: $\rho > 0$ or $\rho < 0$. Suppose first $\rho > 0$. Applying the component-wise \emph{decreasing} transformations $x_i \mapsto y_i = x_i^{-1/\rho}$, we see that the copula density $c_{I;J}$ is equal to the \emph{survival} copula density of the $|I|$-variate probability density function proportional to
\[
	\by_I \mapsto
	\prod_{i \in I} y_i^{\alpha_i-1} \cdot \rbr{\sum_{i \in I} y_i + 1}^{-\rho-\sum_{k\in I \cup J} \alpha_k}.
\]
Let $\bY = (Y_i)_{i \in I}$ be a random vector on $(0,\infty)^I$ with the latter density function and consider the random vector $t(\bY) = (\bW,S)$ with $S = \sum_{i \in I} Y_i$ and $W_i = Y_i / S$ for $i \in I$; note that $\sum_{i \in I} W_i = 1$, so we consider $\bW$ as a random vector on the $(|I|-1)$-dimensional unit simplex $\Delta_{I} = \cbr{\bw \in [0, 1]^I : \sum_{i \in I} w_i}$, where one of the coordinates, say $w_{i_0}$ for some fixed $i_0 \in I$, is written as $w_{i_0} = 1 - \sum_{i \in I \setminus \cbr{i_0}} w_i$. The inverse transformation is $t^{-1}(\bw;s) = (sw_i)_{i \in I}$. By the multivariate changes-of-variables formula, the probability density of $t(\bY)$ is
\[
	\prod_{i \in I} w_i^{\alpha_i-1} \cdot s^{\sum_{i \in I} \alpha_i - 1} (s + 1)^{-\rho-\sum_{k \in I \cup J} \alpha_k}
\]
for $(\bw, s) \in \Delta_{I} \times (0, \infty)$. The random vector $\bW$ thus has a Dirichlet distribution with parameter vector $(\alpha_i)_{i \in I}$ and is independent of the positive random variable $S$. Therefore, $\bY = S \bW$ has a Liouville distribution and its survival copula is a Liouville copula.

Second, suppose $\rho < 0$. We can then apply exactly the same component-wise transformations $x_i \mapsto y_i = x_i^{-1/\rho}$ as in the case $\rho > 0$, but now these transformations are \emph{increasing}, so that $c_{I;J}$ is equal to the copula density of the random vector $\bY$. As the survival copula of $\bY$ is a Liouville copula, $c_{I;J}$ itself is a survival Liouville copula.
\end{proof}

\begin{proof}[Calculating $c_{I;J}$ for the negative logistic model\pnt]
If $\alpha_1 = \ldots = \alpha_d = 1$ and $\rho > 0$, then, writing $\theta = 1/\rho > 0$, the probability density in \eqref{eq:pdfIJscdiri} is proportional to the function
\[
	\bx_I \mapsto 
	\prod_{i \in I} x_i^{-\theta-1} \cdot
	\rbr{1 + \sum_{i \in I} x_i^{-\theta}}^{-1/\theta-|I|-|J|},
\] 
for $\bx_I \in (0, \infty)^I$.
We have $1/\theta+|J| = 1/\theta_J$ with $\theta_J = \theta / \rbr{|J|\theta+1}$.
The copula density does not change if we apply the component-wise increasing transformations $(0, \infty) \to (0, 1) : x_i \mapsto u_i = (x_i^{-\theta}+1)^{-1/\theta_J}$ to the $|I|$ variables. Since 
$x_i^{-\theta} = u_i^{-\theta_J}-1$ and 
$\diff x_i / \diff u_i 
= (\theta_J/\theta) x_i^{1+\theta} u_i^{-\theta_J-1}$,
the resulting probability density is proportional to
\[
	\bu_I \mapsto 
	\prod_{i \in I} u_i^{-\theta_J-1} \cdot
	\rbr{1 - |I| + \sum_{i \in I} u_i^{-\theta_J}}^{-1/\theta_J-|I|},
\]
which is proportional to the density of the $|I|$-variate Clayton copula with parameter $\theta_J$.
\end{proof}

\begin{proof}[Calculating $c_{I;J}$ for the logistic model\pnt]
If $\alpha_1 = \ldots = \alpha_d$ and $\rho = -1/\theta$ with $\theta > 1$, then we seek the copula density of the probability density proportional to
\[
	\bx_I \mapsto 
	\prod_{i \in I} x_i^{\theta-1} \cdot
	\rbr{1 + \sum_{i \in I} x_i^{\theta}}^{1/\theta-|I|-|J|},
	\qquad \bx_I \in (0, \infty)^I.	
\] 
This time, we have $1/\theta - |J| = -1/\vartheta_J$ with $\vartheta_J = \theta / \rbr{|J|\theta-1}$.
We apply the component-wise \emph{decreasing} transformations $(0, \infty) \to (0, 1) : x_i \mapsto u_i = (x_i^{\theta}+1)^{-1/\vartheta_J}$. Since $x_i^\theta = u_i^{-\vartheta_J} - 1$ and 
$\diff x_i / \diff u_i 
= - (\vartheta_J/\theta) x_i^{1-\theta} u_i^{-\vartheta_J-1}$, 
the copula density we seek is equal to the one of the \emph{survival} copula of the probability density proportional to
\[
	\bu_I \mapsto 
	\prod_{i \in I} u_i^{-\vartheta_J-1} \cdot
	\rbr{1 - |I| + \sum_{i \in I} u_i^{-\vartheta_J}}^{-1/\vartheta_J-d},
	\qquad \bu_I \in (0, 1)^I.
\]
Again, this function is proportional to the density of the $|I|$-variate Clayton copula with parameter $\vartheta_J$.
\end{proof}

\begin{proof}[Proof of Proposition~\ref{prop:hr}\pnt]
	Write $K = I \cup J$. The \HR{} tail copula density $\xdf(\point,\Sigk)$ in \eqref{eq:HR} is marginally closed in the sense that the $|K|$-variate tail copula density $\xdf_K(\point;\Sigk)$ resulting from marginalization as in \eqref{eq:lambdaJ} is of the same parametric form but now with respect to the $(|K| \times |K|)$-dimensional variogram $\Gamma_{KK}$. The covariance matrix \eqref{eq:Sigk} with $\Gamma$ replaced by $\Gamma_{KK}$ is $\Sigk_{KK}$. We obtain
	\[
	\xdf_{K} \rbr{\bx; \Sigk}
	= \xdf \rbr{\bx; \Sigk_{KK}}
	= \rbr{\prod_{i \in K \setminus \cbr{k}} x_i^{-1}} \phi_{|K|-1} \rbr{ \bar{\bx}_{\setminus k}; \Sigk_{KK} }, 
	\]
	for $\bx \in (0, \infty)^K$, with $\bar{\bx}_{\setminus k} = \rbr{\log(x_i/x_k) - \half \Gamma_{ik}}_{i \in K \setminus \cbr{k}}$. For fixed $\bx_J \in (0, \infty)^J$, the function $\bx_I \mapsto \xdf_{I \cup J}(\bx_{I \cup J}; \Sigk)$ is proportional to an $|I|$-variate log-normal density with covariance matrix $\Sigk_{I|J}$ and some mean vector depending on $\bx_J$; this follows from a well-known property of conditional distributions of the Gaussian distribution and in particular from the block-matrix inversion formula of
	\begin{equation*}
		\Sigk_{KK} =
		\begin{bmatrix}
			\Sigk_{II} & \Sigk_{IJ} \\
			\Sigk_{JI} & \Sigk_{JJ}
		\end{bmatrix}.
	\end{equation*}
	Since the copula of a random vector is invariant with respect to component-wise increasing transformations, the copula density of the probability density $r_{I|J}(\point|\bx_J;\Sigk)$ is equal to the stated Gaussian copula density.
\end{proof}

\subsection{Proofs for Section~\ref{sec:vine}}

\begin{proof}[Proof of Theorem~\ref{thm:xvine} and Remark~\ref{rem:xdfdecomp-J}\pnt]
Applying the regular vine telescoping equations~\eqref{eq:Rvine-a} and~\eqref{eq:Rvine-a-cond} to $\gamma_J := \xdf_J(\bx_J)$, we find
\[
	\xdf(\bx)
	= \prod_{e \in \edges_1} \xdf_{a_e,b_e}(x_{a_e},x_{b_e}) \cdot
	\prod_{j=2}^{d-1} \prod_{e \in \edges_j}
	\frac{\xdf_{a_e,b_e|\cndg_e}(x_{a_e},x_{b_e}|\bx_{\cndg_e})}{\xdf_{a_e|\cndg_e}(x_{a_e}|\bx_{\cndg_e}) \cdot \xdf_{b_e|\cndg_e}(x_{b_e}|\bx_{\cndg_e})}.
\]
By Eq.~\eqref{eq:cIJ} in Sklar's theorem for tail copula densities (Proposition~\ref{lem:expdensity}) applied to $I = \cbr{a_e,b_e}$ and $J = \cndg_e$, the last fraction is equal to $c_{a_e,b_e;\cndg_e}\rbr{ u_{a_e}, u_{b_e} ; \bx_{\cndg_e} }$ with $u_{a_e} = R_{a_e|\cndg_e}(x_{a_e}|\bx_{\cndg_e})$ and $u_{b_e} = R_{b_e|\cndg_e}(x_{b_e}|\bx_{\cndg_e})$. Formula~\eqref{eq:xdfdecomp-J} in Remark~\ref{rem:xdfdecomp-J} follows in the same way from the more general product formula~\eqref{eq:telescope-f} in Remark~\ref{rem:telescope-f}.
\end{proof}

\begin{remark}
	\label{rem:xdfdecomp-J}
	The decomposition~\eqref{eq:xdfdecomp} is based on the vine telescoping product formulas \eqref{eq:Rvine-a}--\eqref{eq:Rvine-a-cond} in Lemma~\ref{lem:Rvine-a}. By applying the more general formula~\eqref{eq:telescope-f} in Remark~\ref{rem:telescope-f}, we find, for every edge $f \in \edges_j$ with $j \in \cbr{1,\ldots,d-1}$, a decomposition of the marginal tail copula density on the complete union $J = \cunn_f \subseteq \dset$, with cardinality $|J| = j+1$: for $\bx_J \in (0, \infty)^J$, we have
	\begin{equation}
		\label{eq:xdfdecomp-J}
		\xdf_J(\bx_J) = \prod_{\substack{e \in \edges_1\\ \cunn_e \subseteq J}} \xdf_{a_e,b_e}(x_{a_e},x_{b_e}) \cdot
		\prod_{i=2}^{|J|-1} \prod_{\substack{e \in \edges_i\\ \cunn_e \subseteq J}}
		c_{a_e,b_e;\cndg_e} \rbr{ 
			\xcdf_{a_e|\cndg_e}(x_{a_e}|\bx_{\cndg_e}), \xcdf_{b_e|\cndg_e}(x_{b_e}|\bx_{\cndg_e});
			\bx_{\cndg_e}
		}.
	\end{equation}
	The decomposition \eqref{eq:xdfdecomp-J} is true only for subsets $J \subseteq \dset$ of the form $J = \cunn_f$ for some edge $f$ in the regular vine sequence.
\end{remark}

\begin{proof}[Proof of Theorem~\ref{thm:recuni}\pnt]
Let $j_0 \in \cbr{2,\ldots,d-1}$ and $e_0 \in \edges_{j_0}$. The complete union $J_0 = \cunn_{e_0} \subset \dset$ has cardinality $j_0+1$. We can restrict the regular vine sequence $\Vine$ to $J_0$ in the obvious way to obtain a regular vine sequence $\Vine_0 = (\tree_{0,j})_{j=1}^{j_0}$ on $j_0+1$ elements: the first tree $\tree_{0,1}$ has node set $\nodes_{0,1} = J_0$ and edge set $\edges_{0,1} = \cbr{e \in \edges_1 : e \subset J_0}$, and so on. The edge $e_0$ is the single one in the deepest tree $\tree_{0,j_0}$ of $\Vine_0$. By applying Theorem~\ref{thm:xvine} to $\xdf_{J_0}$, we obtain the same decomposition as in \eqref{eq:xdfdecomp}, with the same bivariate (tail) copula densities, but restricted to $\Vine_0$. This means that, without loss of generality, we can assume that $e_0$ is the deepest edge of the original regular vine sequence $\Vine$, that is, $e_0$ is the unique element, say $e$, in $\edges_{d-1}$.

We have $e = (a_e,b_e;\cndg_e) = \cbr{f,g}$ with $f, g \in \nodes_{d-1} = \edges_{d-2}$. The complete union of $e$ is $\cunn_e = \dset$. Furthermore, upon switching the names of $f$ and $g$ if necessary, we have $\cunn_f = \cndg_e \cup \cbr{a_e} =  \dset \setminus \cbr{b_e}$ and $\cunn_g = \cndg_e \cup \cbr{b_e} = \dset \setminus \cbr{a_e}$. We can apply the decomposition~\eqref{eq:xdfdecomp} to both $\xdf$ and its $(d-1)$-dimensional marginal density $\xdf_{\cunn_g}$, as in the previous paragraph; see also Eq.~\eqref{eq:xdfdecomp-J}. The two decompositions are the same, except that the one for $\xdf$ has some additional factors, namely the ones that involve node $a_e$. For each tree, there is only one such additional factor, since in each tree $\tree_j$, there is only a single edge, say $e_j$, that contains $a_e$ in its complete union (Lemma~\ref{lem:uniqueness}). For these edges $e_j$, we have in fact $a_e \in \cndd_{e_j}$, so we can write $e_j = (a_e,b_{e_j};\cndg_{e_j})$. Node that $e_{d-1} = e$ and $b_{e_{d-1}} = b_e$, while $e_{d-2} = f$. It follows that
\begin{align*}
	r_{a_e|\cndg_e \cup b_e} \rbr{x_{a_e}|\bx_{\cndg_e \cup b_e}}
	&= \frac{r(\bx)}{r_{\cunn_g}(\bx_{\cunn_g})} \\
	&= r_{a_e,b_{e_1}} (x_{a_e}, x_{b_{e_1}})
	\cdot \prod_{j=2}^{d-1} c_{a_e,b_{e_j}|\cndg_{e_j}}
	\rbr{
		R_{a_e|\cndg_{e_j}} (x_{a_e}|\bx_{\cndg_{e_j}}),
		R_{b_e|\cndg_{e_j}} (x_{b_e}|\bx_{\cndg_{e_j}});
		\bx_{\cndg_{e_j}}
	}.
\end{align*}
Recall that $\cunn_f = \cndg_e \cup \cbr{a_e}$. Applying the same argument to edge $f$ via the regular vine sequence induced by $\Vine$ on $\cunn_f$ as in the first paragraph of the proof, we find
\[
	r_{a_e|\cndg_e} \rbr{x_{a_e}|\bx_{\cndg_e}} 
	= r_{a_e,b_{e_1}} (x_{a_e}, x_{b_{e_1}})
	\cdot \prod_{j=2}^{d-2} c_{a_e,b_{e_j}|\cndg_{e_j}}
	\rbr{
		R_{a_e|\cndg_{e_j}} (x_{a_e}|\bx_{\cndg_{e_j}}),
		R_{b_e|\cndg_{e_j}} (x_{b_e}|\bx_{\cndg_{e_j}});
		\bx_{\cndg_{e_j}}
	},
\]
where the product is empty if $d = 3$. Comparing the two identities, we get
\[
	r_{a_e|\cndg_e \cup b_e} \rbr{x_{a_e}|\bx_{\cndg_e \cup b_e}}
	= r_{a_e|\cndg_e} \rbr{x_{a_e}|\bx_{\cndg_e}}
	\cdot c_{a_e,b_{e}|\cndg_{e}}
	\rbr{
		R_{a_e|\cndg_{e}} (x_{a_e}|\bx_{\cndg_{e}}),
		R_{b_e|\cndg_{e}} (x_{b_e}|\bx_{\cndg_{e}});
		\bx_{\cndg_{e}}
	}.
\]
Since the partial derivative of $R_{a_e|\cndg_{e}} (x_{a_e}|\bx_{\cndg_{e}})$ with respect to $x_{a_e}$ is $r_{a_e|\cndg_e} \rbr{x_{a_e}|\bx_{\cndg_e}}$, the first of the two equalities in \eqref{eq:xcdfrecur} follows by integration. The second equality in \eqref{eq:xcdfrecur} is shown in the same way.
\end{proof}

\begin{proof}[Proof of Theorem~\ref{thm:xvineconstruct}\pnt]
The proof is by induction on $d$. For $d = 3$, we can, upon renumbering the nodes, assume the regular vine sequence has edge set $\edges_1 = \cbr{\cbr{1,2},\cbr{2,3}}$, and \eqref{eq:xvine} becomes
\[
	r(x_1,x_2,x_3)
	= r_{12}(x_1,x_2) \, r_{23}(x_2,x_3) \cdot
	c_{13;2} \rbr{R_{1|2}(x_1|x_2),R_{3|2}(x_3|x_2)}.
\]
By Proposition~\ref{lem:Sklar}, this function is a valid tail copula density with margins $r_{12}$ and $r_{23}$. Furthermore,
\[
	r_{1|23}(x_1|x_2,x_3)
	= \frac{r(x_1,x_2,x_3)}{r_{23}(x_1,x_2)}
	= r_{12}(x_1,x_2) \cdot
	c_{13;2} \rbr{R_{1|2}(x_1|x_2),R_{3|2}(x_3|x_2)}.
\]
Integrating over $x_1$ yields, after the change of variables $u_1 = R_{1|2}(x_1|x_2)$, the identity
\[
	R_{1|23}(x_1|x_2,x_3)
	= C_{1|3;2} \rbr{R_{1|2}(x_1|x_2), R_{3|2}(x_3|x_2)}.
\]
The proof of the recursive formula for $R_{3|12}(x_3|x_1,x_2)$ is entirely similar.

For the induction step, let $d \ge 4$ and assume the stated properties are true for X-vine specifications on node sets of cardinality up to $d-1$. To show that $\xdf$ defined in \eqref{eq:xvine} is a $d$-variate tail copula density, we check the marginal constraint~\eqref{eq:margin:pdf} and the homogeneity~\eqref{eq:lambdahomo}.

To show \eqref{eq:lambdahomo}, let $\bx \in (0, \infty)^d$ and $s \in (0, \infty)$.
Since $\xdf_{a_e,b_e}$ is a bivariate tail copula density for each $e \in \edges_1$, we have, by \eqref{eq:lambdahomo},
\[
	\prod_{e \in \edges_1} \xdf_{a_e,b_e}(sx_{a_e},sx_{b_e})
	=
	\prod_{e \in \edges_1} 
	\rbr{ s^{-1} \xdf_{a_e,b_e}(x_{a_e},x_{b_e}) }
	=
	s^{-(d-1)} 	\prod_{e \in \edges_1} \xdf_{a_e,b_e}(x_{a_e},x_{b_e}),
\]
as there are $d-1$ edges in the first tree. Further, by the induction hypothesis and by scale equivariance (Lemma~\ref{lem:homo}), we have 
\chng{$R_{a_e|\cndg_e}(sx_{a_e}|s\bx_{\cndg_e}) = R_{a_e|\cndg_e}(x_{a_e}|\bx_{\cndg_e})$} and similarly for $a_e$ replaced by $b_e$. It follows that $\xdf$ defined in \eqref{eq:xvine} satisfies \eqref{eq:lambdahomo}.

Next, we show that $\xdf$ defined in \eqref{eq:xvine} satisfies \eqref{eq:margin:pdf}. Let $e = \cbr{f,g} = (a_e,b_e;\cndg_e)$ be the single edge in the deepest tree $\tree_{d-1}$, with $f, g \in \edges_{d-2}$ and where the labels are such that $a_e \in \cndd_f$ and $b_e \in \cndd_g$. The regular vine sequence $\Vine$ on $\dset$ induces regular vine sequences $\Vine_f$ and $\Vine_g$ on the node sets 
\[
	\cunn_f 
	= \dset \setminus \cbr{b_e} 
	= \cndg_e \cup \cbr{a_e} 
	\quad \text{and} \quad
	\cunn_g 
	= \dset \setminus \cbr{a_e} 
	= \cndg_e \cup \cbr{b_e},
\]
respectively. The X-vine specification $(\Vine, \xdfset, \cpdfset)$ can be restricted in similar ways, by keeping only those bivariate (tail) copula densities associated to edges with complete unions contained in $\cunn_f$ and $\cunn_g$, respectively, yielding X-vine specifications $(\Vine_f, \xdfset_f, \cpdfset_f)$ and $(\Vine_g, \xdfset_g, \cpdfset_g)$. By the induction hypothesis, these two X-vine specifications produce $(d-1)$-dimensional tail copula densities $\xdf_{\cunn_f}$ and $\xdf_{\cunn_g}$ as in \eqref{eq:xvine} but with factors restricted to those edges whose complete unions are contained in $\cunn_f$ and $\cunn_g$, respectively. Let $b_f$ be the other node in $\cndd_f$ besides $a_e$, i.e., $\cndd_f = \cbr{a_e, b_f}$; similarly, let $a_g$ be the other node in $\cndd_g$ besides $b_e$, i.e., $\cndd_g = \cbr{a_g, b_e}$. We have
\[
	\cndg_e 
	= \cndg_f \cup \cbr{b_f} 
	= \cndg_g \cup \cbr{a_g}.
\]
By the induction hypothesis applied to $(\Vine_f, \xdfset_f, \cpdfset_f)$, we can obtain $\xcdf_{a_e|\cndg_e} = \xcdf_{a_e|\cndg_f \cup b_f}$ by the recursive formula \eqref{eq:xvine:recur}; the same is true for $\xcdf_{b_e|\cndg_e} = \xcdf_{b_e|\cndg_g \cup a_g}$.

For each $j \in \cbr{1,\ldots,d-1}$, there exists a unique edge $e_j \in \edges_j$ such that $a_e \in \cunn_{e_j}$ and then in fact $a_e \in \cndd_{e_j}$ (Lemma~\ref{lem:uniqueness}); clearly $e_{d-1} = e$ and $e_{d-2} = f$. For $j \in \cbr{1,\ldots,d-1}$, let $b_j = \cndd_{e_j} \setminus \cbr{a_e}$ be the other node in the conditioned set of $e_j$ besides $a_e$; we have $b_{d-1} = b_e$ and $b_{d-2} = b_f$. Isolating in \eqref{eq:xvine} the factors indexed by edges $e_1,\ldots,e_{d-1}$, we obtain
\chng{
\[
	\xdf(\bx) = \xdf_{\cunn_g}(\bx_{\cunn_g})
	\cdot r_{a_e,b_1}(x_{a_e},x_{b_1})
	\cdot \prod_{j=2}^{d-1} c_{a_e,b_j;\cndg_{e_j}} \rbr{
		\xcdf_{a_e|\cndg_{e_j}}(x_{a_e}|\bx_{\cndg_{e_j}}),
		\xcdf_{b_j|\cndg_{e_j}}(x_{b_j}|\bx_{\cndg_{e_j}})
	}.
\] 
}
By the chain rule, the derivative of the function
\[
	x_{a_e}	\mapsto 
	C_{a_e|b_e;\cndg_e} \rbr{
		\xcdf_{a_e|\cndg_e}(x_{a_e}|\bx_{\cndg_e}), \,
		\xcdf_{b_e|\cndg_e}(x_{b_e}|\bx_{\cndg_e})
	}
\]
is equal to the function
\[
	x_{a_e} \mapsto
	c_{a_e,b_e;\cndg_e} \rbr{
		\xcdf_{a_e|\cndg_e}(x_{a_e}|\bx_{\cndg_e}), \,
		\xcdf_{b_e|\cndg_e}(x_{b_e}|\bx_{\cndg_e})
	}
	\cdot
	\frac{\partial}{\partial x_{a_e}}
	\xcdf_{a_e|\cndg_e}(x_{a_e}|\bx_{\cndg_e}).
\]
Recall $\cndg_e \cup \cbr{a_e} = \cunn_f$, $a_e \in \cndd_f$, and $f = e_{d-2}$.
Applying the induction hypothesis to $(\Vine_f, \xdfset_f, \cpdfset_f)$, we find
\begin{align*}
	\frac{\partial}{\partial x_{a_e}}
	\xcdf_{a_e|\cndg_e}(x_{a_e}|\bx_{\cndg_e}) 
	&= \frac{\xdf_{a_e \cup \cndg_e}(\bx_{a_e \cup \cndg_e})}{\xdf_{\cndg_e}(\bx_{\cndg_e})}
	\\
	&= \xdf_{a_e,b_1}(x_{a_e},x_{b_1})
	\cdot \prod_{j=2}^{d-2} c_{a_e,b_j;\cndg_{e_j}} \rbr{
		\xcdf_{a_e;\cndg_{e_j}}(x_{a_e}|\bx_{\cndg_{e_j}}),
		\xcdf_{b_j;\cndg_{e_j}}(x_{b_j}|\bx_{\cndg_{e_j}})
	}.
\end{align*}
Combining these identities, we get
\begin{equation}
\label{eq:rrg}
	\xdf(\bx) = \xdf_{\cunn_g}(\bx_{\cunn_g})
	\cdot \frac{\partial}{\partial x_{a_e}}
	C_{a_e|b_e;\cndg_e} \rbr{
		\xcdf_{a_e|\cndg_e}(x_{a_e}|\bx_{\cndg_e}), \,
		\xcdf_{b_e|\cndg_e}(x_{b_e}|\bx_{\cndg_e})
	}.
\end{equation}
It follows that
\[
	\int_0^\infty r(\bx) \, \diff x_{a_e}
	= \xdf_{\cunn_g}(\bx_{\cunn_g}).
\]
Since $\xdf_{\cunn_g}$ is a $(d-1)$-variate tail copula density, we can fix variable $x_i \in (0, \infty)$ with $i \in \cunn_g$ and integrate further over all $(d-2)$ remaining variables in $\cunn_g$, to find that $\int_{(0, \infty)^{d-1}} r(\bx) \, \diff \bx_{\setminus i} = 1$. Repeating the same argument with the roles of $f$ and $g$ interchanged leads to
\[
	\xdf(\bx) = \xdf_{\cunn_f}(\bx_{\cunn_f})
	\cdot \frac{\partial}{\partial x_{b_e}}
	C_{b_e|a_e;\cndg_e} \rbr{
		\xcdf_{b_e|\cndg_e}(x_{b_e}|\bx_{\cndg_e}), \,
		\xcdf_{a_e|\cndg_e}(x_{a_e}|\bx_{\cndg_e})
	}
\]
and thus
\[
	\int_0^\infty r(\bx) \, \diff x_{b_e}
	= \xdf_{\cunn_f}(\bx_{\cunn_f}),
\]
from which also $\int_{(0, \infty)^{d-1}} r(\bx) \, \diff \bx_{\setminus i} = 1$ for $i \in \cunn_f$ and $x_i \in (0, \infty)$. As $\cunn_f \cup \cunn_g = \cunn_e = \dset$, we have shown \eqref{eq:margin:pdf}.

Finally, we show the recursive formulas \eqref{eq:xvine:recur}. From $\cndg_e \cup \cbr{b_e} = \cunn_g$, it follows that
\begin{align*}
	\frac{\partial}{\partial x_{a_e}}
	\xcdf_{a_e|\cndg_e \cup b_e}(x_{a_e}|\bx_{\cndg_e \cup b_e})
	&= \xdf_{a_e|\cndg_e \cup b_e}(x_{a_e}|\bx_{\cndg_e \cup b_e}) 
	= \frac{\xdf(\bx)}{\xdf_{\cunn_g}(\bx_{\cunn_g})} \\
	&= \frac{\partial}{\partial x_{a_e}}
	C_{a_e|b_e;\cndg_e} \rbr{
		\xcdf_{a_e|\cndg_e}(x_{a_e}|\bx_{\cndg_e}), \,
		\xcdf_{b_e|\cndg_e}(x_{b_e}|\bx_{\cndg_e})
	},
\end{align*}
where we used \eqref{eq:rrg} in the last step. Integrating over $x_{a_e}$ yields the first identity in \eqref{eq:xvine:recur}. The second identity follows in the same way. The proof of Theorem~\ref{thm:xvineconstruct} is complete.
\end{proof}

\begin{proof}[Proof of Proposition~\ref{prop:c2l}\pnt]
Let $\bx \in (0, \infty)^d$.
As there $d-1$ edges in the first tree, we have
\[
	t^{d-1} \cdot c(t \bx)
	= \prod_{e \in \edges_1} 
	\rbr{t \cdot c_{a_e,b_e}(t x_{a_e}, t x_{b_e})} 
	\cdot
	\prod_{j=2}^{d-1} \prod_{e \in \edges_j}
	c_{a_e,b_e;\cndg_e} \rbr{ 
		C_{a_e|\cndg_e}(tx_{a_e}|t\bx_{\cndg_e}), C_{b_e|\cndg_e}(tx_{b_e}|t\bx_{\cndg_e})
	}.
\]
The first product on the right-hand side converges to $\prod_{e \in \edges_1} \xdf_{a_e,b_e}(x_{a_e},x_{b_e})$ as $t \to 0$ by Assumption~\ref{ass:PCClambdae}(i). It remains to show that for every $j \in \cbr{2,\ldots,d-1}$, every $e \in \edges_j$ and every $i \in \cndd_e = \cbr{a_e,b_e}$, we have
\begin{equation}
\label{eq:c2l:toshow}
	\lim_{t \searrow 0} 
	C_{i|\cndg_e} \rbr{t x_i | t \bx_{\cndg_e}} 
	= \xcdf_{i|\cndg_e} \rbr{x_i | \bx_{\cndg_e}}.
\end{equation}
By properties of regular vine sequences as explained in Section~\ref{sec:Rvine}, there exists an edge $f = (a_f, b_f ; \cndg_f) \in \edges_{j-1}$ such that $f \in e$ and $i \in \cndd_f = \cbr{a_f,b_f}$ and $\cunn_f = \cndg_e \cup \cbr{i}$. The proof of \eqref{eq:c2l:toshow} is by recursion on $j$.

If $j = 2$, then $f \in \edges_1$ so that $\cndg_f = \varnothing$ while $\cndg_e$ is a singleton, say $\cndg_e = \cbr{d_e}$, and thus $f = \cbr{i, d_e}$. But then \eqref{eq:c2l:toshow} follows from Assumption~\ref{ass:PCClambdae}(i) and Lemma~\ref{lem:frompdftocdf2} below.

Let $j \ge 3$ and assume \eqref{eq:c2l:toshow} has been shown for edges in $\edges_1 \cup \cdots \cup \edges_{j-1}$. Let $k$ be the other node in $\cndd_f$ but $i$, that is, $\cbr{a_f,b_f} = \cndd_f = \cbr{i, k}$. Then by \eqref{eq:PCC:recur} applied to $f$, we have
\[
	C_{i,k|\cndg_f} \rbr{tx_i, tx_k | t \bx_{\cndg_f}}
	=
	C_{i,k;\cndg_f} \rbr{
		C_{i|\cndg_f} \rbr{tx_i | t\bx_{\cndg_f}},
		C_{k|\cndg_f} \rbr{tx_k | t\bx_{\cndg_f}}
	}
\]
and thus
\begin{align*}
	C_{i|\cndg_e} \rbr{tx_i|t\bx_{\cndg_e}}
	&= 
	C_{i|\cndg_f \cup k} \rbr{tx_i|t\bx_{\cndg_f \cup k}}
	\\
	&=
	C_{i|k;\cndg_f} \rbr{
		C_{i|\cndg_f} \rbr{tx_i|t\bx_{\cndg_f}},
		C_{k|\cndg_f} \rbr{tx_k|t\bx_{\cndg_f}}
	}.
\end{align*}
By the induction hypothesis, the two arguments of $C_{i|k;\cndg_f}$ converge to $R_{i|\cndg_f}(x_i|\bx_{\cndg_f})$ and $R_{k|\cndg_f}(x_k|\bx_{\cndg_f})$, respectively. 
As $f \in \edges_{j-1}$ and $j \ge 3$, the pair copula density $c_{i,k;\cndg_f}$ is continuous by Assumption~\ref{ass:PCClambdae}(ii). In view of Corollary~\ref{cor:cScheffe} below, $C_{i|k;\cndg_f}$ is continuous in its two arguments too, and thus
\[
	\lim_{t \searrow 0} C_{i|\cndg_e} \rbr{tx_i|t\bx_{\cndg_e}}
	=
	C_{i|k;\cndg_f} \rbr{
		\xcdf_{i|\cndg_f} \rbr{x_i|\bx_{\cndg_f}},
		\xcdf_{k|\cndg_f} \rbr{x_k|\bx_{\cndg_f}}
	}.
\]
By the recursive property~\eqref{eq:xvine:recur}, the last expression is equal to $\xcdf_{i|\cndg_f \cup k} \rbr{x_i | \bx_{\cndg_f \cup k}}$. As $\cndg_f \cup k = \cndg_e$, the convergence~\eqref{eq:c2l:toshow} follows.
\end{proof}

\begin{lemma}
	\label{lem:frompdftocdf2}
	Let $c$ be a bivariate copula density and $\xdf$ a bivariate tail copula density. Assume that
	\[
	\lim_{t \searrow 0} t \cdot c(tx, ty) 
	= \xdf(x, y), \qquad (x, y) \in (0, \infty)^2.
	\]
	The copula $C$ of $c$ and the tail copula $\xcdf$ of $\xdf$ then satisfy
	\[
	\begin{split}
		\lim_{t \to 0} C_{1|2}(tx|ty) &= \xcdf_{1|2}(x|y), \\
		\lim_{t \to 0} C_{2|1}(ty|tx) &= \xcdf_{2|1}(y|x).
	\end{split}
	\]
\end{lemma}

\begin{proof}[Proof of Lemma \ref{lem:frompdftocdf2}\pnt]
	By symmetry, it is sufficient to show the first relation.
	Fix $y \in (0, \infty)$ and define $f_t(x) = t \cdot c(tx, ty)$ for $x \in (0, \infty)$. By assumption, $\lim_{t \to 0} f_t(x) = \xdf(x,y)$ for all $x \in (0, \infty)$. Moreover, for all $t \in (0, 1/y)$, we have
	\[ 
	\int_0^{\infty} f_t(x) \, \diff x
	= \int_0^{\infty} c(u, ty) \, \diff u
	= 1 
	= \int_0^{\infty} \xdf(x,y) \, \diff x.
	\]
	By Scheffé's lemma applied to the functions $f_t$ and their point-wise limit $x \mapsto f(x) = \xdf(x, y)$, it follows that 
	\[
	\lim_{t \to \infty}
	\int_0^{\infty} 
	\abs{ t \cdot c(tx, ty) - \xdf(x, y) }
	\, \diff x = 0.
	\]
	But then
	\begin{align*}
		\abs{ C_{1|2}(tx|ty) - \xcdf_{1|2}(x|y) }
		&= \abs{ 
			\int_{0}^{tx} c(s, ty) \, \diff s 
			-
			\int_{0}^x \xdf(s, y) \, \diff s
		} \\
		&\le \int_0^{x} 
		\abs{ t \cdot c(ts, ty) - \xdf(x, y) }
		\, \diff s \to 0, \qquad t \to 0.
		\qedhere
	\end{align*}
\end{proof}

\begin{lemma}
	\label{lem:fScheffe}
	Let $I$ and $J$ be non-empty open intervals and let $f : I \times J \to [0, \infty)$ be a continuous function such that $f_2 : J \to [0, \infty) : y \mapsto \int_{I} f(x, y) \, \diff x$ is continuous. Then the function 
	\[ 
		(x, y) \mapsto 
		\int_{I} f(s, y) \1{\{s \le x\}} \, \diff s
	\]
	is continuous on $\overline{I} \times J$ too, where $\overline{I}$ is the closure of $I$ in $[-\infty,+\infty]$. 
\end{lemma}

\begin{proof}[Proof of Lemma~\ref{lem:fScheffe}\pnt]
	For $y \in J$, define $f_y : I \to [0, \infty) : x \mapsto f_y(x) := f(x, y)$. Then $\lim_{y' \to y} f_{y'}(x) = f_{y}(x)$ for all $(x, y) \in I \times J$ and $\int_I f_{y'}(x) \, \diff x = f_2(y') \to f_2(y) = \int_{I} f_{y}(x) \, \diff x$ as $y' \to y$ in $J$. By Scheff\'e's lemma, it follows that
	\begin{equation}
	\label{eq:y'toyscheffe}
		\lim_{y' \to y} \int_{I} 
		\abs{ f_{y'}(x) - f_y(x) }
		\, \diff x = 0.
	\end{equation}
	But then, as $(x', y') \to (x, y)$ in $\overline{I} \times J$, we have
	\begin{align*}
    \lefteqn{\abs{
			\int_{I} f(s, y') \1{\{s \le x'\}} \, \diff s
			- 
			\int_{I} f(s, y) \1{\{s \le x\}} \, \diff s
		}} \\
    	&\le
		\int_{I} \abs{ f(s, y') - f(s, y) } \, \diff s
		+
		\int_{I} f(s, y) 
		\abs{ \1{\{s \le x\}} - \1{\{s \le x'\}} }
		\, \diff s
		\to 0.
    \end{align*}
	The first term on the right-hand side converges to zero by \eqref{eq:y'toyscheffe} and the second term vanishes asymptotically by the dominated convergence theorem.
\end{proof}

\begin{corollary}
\label{cor:cScheffe}
	Let $c$ be a bivariate copula density and let $C_{1|2}(u|v) = \int_{0}^{u} c(w,v) \, \diff w$ for $(u, v) \in [0, 1] \times (0, 1)$. If $c$ is continuous, then so is $C_{1|2}(\point|\point)$. 
\end{corollary}

\begin{proof}[Proof of Corollary~\ref{cor:cScheffe}\pnt]
	Apply Lemma~\ref{lem:fScheffe} to $f = c$ and $I = J = (0, 1)$.
\end{proof}

\subsection{Proofs for Section~\ref{sec:simulation}}
\label{sec:proof:simu}

\begin{proof}[Proof of Lemma~\ref{lem:perm}\pnt]
	The proof is by induction on $d$. If $d = 3$, we can, after relabeling the nodes if necessary, assume that $\edges_1 = \cbr{\cbr{1,2},\cbr{2,3}}$. Then $\edges_2 = \cbr{e_2}$ with $e_2 = (13;2)$. The three permutations are then as follows:
	\begin{itemize}
		\item For $j = 1$, we have $\sigma_1 = (1,2,3)$ with $e_{1,1} = \cbr{1,2}$.
		\item For $j = 2$, we have $\sigma_2 = (2,1,3)$ with $e_{2,1} = \cbr{1,2}$ or $\sigma_2 = (2,3,1)$ with $e_{2,1} = \cbr{2,3}$.
		\item For $j = 3$, we have $\sigma_3 = (3,2,1)$ with $e_{3,1} = \cbr{2,3}$.
	\end{itemize}
	
	Next, let $d \ge 4$ and assume the statement is true for regular vine sequence on $d-1$ elements.
	Let $e = (a_e,b_e;\cndg_e) = \cbr{f,g}$ be the single edge in $\edges_{d-1}$, with $f, g \in \edges_{d-2}$ labeled in such a way that $a_e \in \cndd_f$ and $b_e \in \cndd_g$. 
	The regular vine sequence $\Vine$ induces regular vine sequences $\Vine_f$ and $\Vine_g$ on the node sets $\cunn_f = \dset \setminus \cbr{b_e}$ and $\cunn_g = \dset \setminus \cbr{a_e}$, both of which have $d-1$ elements; the edges $\bar{e}$ of the two regular vine sequences $\Vine_f$ and $\Vine_g$ are those of $\Vine$ that satisfy $\cunn_{\bar{e}} \subseteq \cunn_f$ or $\cunn_{\bar{e}} \subseteq \cunn_g$, respectively.
	
	Let $j \in \dset$.
	If $j \neq b_e$, then we can apply the induction hypothesis to $\Vine_f$ to find a bijection $\sigma_{f,j} : \cbr{1,\ldots,d-1} \to \cunn_f = \dset \setminus \cbr{b_e}$ such that (i) $\sigma_{f,j}(1) = j$ and (ii) $\sigma_{f,j}(k)$ belongs to the conditioned set of some edge $e_{k-1}$ in the $(k-1)$th tree of $\Vine_f$ as well as $\cbr{\sigma_{f,j}(i): i = 1,\ldots,k} = \cunn_{e_{k-1}}$, for all $k \in \cbr{2,\ldots,d-1}$.
	Extend $\sigma_{f,j}$ to a map $\sigma_j$ on $\dset$ by $\sigma_{j}(d) = b_e$. Then $\sigma_j$ meets all the requirements.
	
	If $j = b_e$, then $j \ne a_e$, and we can apply the same argument to $\Vine_g$.
\end{proof}

\begin{proof}[Proof of Proposition~\ref{prop:simu}\pnt]
The distribution of $Z_j^{(j)}$ is that of $Z_j$ given that $Z_j < 1$, and this is the uniform distribution on $(0, 1)$. The fact that the distribution of $\bZ^{(j)}$ defined in \eqref{eq:IMPj} is a consequence of the inverse Rosenblatt transform.

Further, the permutation $\sigma_j$ in Lemma~\ref{lem:perm} is such that for every $k \in \cbr{2,\ldots,d}$, there exists an edge $e_{j,k-1} \in \edges_{k-1}$ such that $\sigma_j(k) \in \cndd_{e_{j,k-1}} = \cbr{a_{e_{j,k-1}}, b_{e_{j,k-1}}}$ and $\cbr{\sigma_j(i) : i = 1,\ldots,k} = \cunn_{e_{j,k-1}} = \cbr{a_{e_{j,k-1}},b_{e_{j,k-1}}} \cup \cndg_{e_{j,k-1}}$. These properties imply the stated representation of the conditional quantile function.
\end{proof}

\subsection{Proofs for Section~\ref{sec:estimation}}
\label{sec:proofs7}

\begin{proof}[Proof of Proposition~\ref{prop:Z2c}\pnt]
\emph{(i)} By Definition~\ref{def:invmultpar}, the density of $\IMP$ is $\bz \mapsto \xdf(\bz) \1_{\LL}(\bz) / \xcdf(\LL)$ with $\LL = \cbr{ \bx : \min(\bz) < 1 }$. Let $K \subset \dset$ be non-empty and let $\bz_K \in (0, \infty)^K$ be such that \chng{$\min \bz_{K} < 1$}. The marginal density of $\IMP_K = (\Imp_k)_{k \in K}$ evaluated in $\bz_K$ is
\[
	\int_{(0,\infty)^{\dset \setminus K}} 
		\frac{\xdf(\bz) \1_{\LL}(\bz)}{\xcdf(\LL)} \,
	\diff \bz_{\dset \setminus K}.
\]
Since $\min \bz_K < 1$, we also have $\min \bz < 1$ for $\bz_{\dset \setminus K} \in (0,\infty)^{\dset \setminus K}$. As a consequence, the indicator function in the integral is always equal to one and, by Eq.~\eqref{eq:lambdaJ}, the density of $\bZ_K$ evaluated in $\bz_K$ simplifies to $\xdf_K(\bz_K) / \xdf(\LL)$.
	
Applying this identity to $K = J$ and $K = I \cup J$, we find, for $\bz_I \in (0, \infty)^I$ and for $\bz_J \in (0, \infty)^J$ such that $\min \bz_J < 1$ and $\xdf_J(\bz_J) > 0$, the conditional density of $\IMP_I$ given $\IMP_J = \bz_J$ evaluated in $\bz_I$: 
\[
	\frac{\xdf_{I \cup J}(\bz_{I \cup J}) / \xcdf(\LL)}{\xdf_{J}(\bz_{J}) / \xcdf(\LL)}
	= \frac{\xdf_{I \cup J}(\bz_{I \cup J})}{\xdf_{J}(\bz_{J})}
	= \xdf_{I|J}(\bz_I | \bz_J).
\]
	
\emph{(ii)} Let $\bz_J \in (0, \infty)^J$ be such that $\min \bz_J < 1$ and $\xdf_J(\bz_J) > 0$. By statement~(i) and by Eq.~\eqref{eq:cIJ} in Sklar's theorem (Proposition~\ref{lem:expdensity}), the conditional density of $\rbr{\xcdf_{i|J}(\Imp_i|\IMP_J)}_{i \in I}$ given $\IMP_J = \bz_J$ is $c_{I;J}$. Since, by the simplifying assumption, the latter copula density does not depend on $\bz_J$, the random vector $\rbr{\xcdf_{i|J}(\Imp_i|\IMP_J)}_{i \in I}$ is independent of $\IMP_J$, conditionally on $\min \IMP_J < 1$.
\end{proof}

\section{Exponent measure densities as X-vines}
\label{app:expmeas}

Recall the exponent measure $\expmeas$ in \eqref{eq:Lam2nu} and its density $\nupdf$ in \eqref{eq:nupdf}. To facilitate linking up our results with the literature on multivariate extreme value theory, we present the main formulas in the paper in terms of $\nupdf$ too.

For a proper, non-empty subset $J \subset \dset$, the marginal exponent measure density is
\begin{equation*}
	\nupdf_J(\by_J)
	= \int_{\by_{\setminus J} \in (0,\infty)^{d-|J|}}
	\nupdf(\by) \, \diff \by_{\setminus J},
	\qquad \by_J \in (0, \infty)^{|J|}.
\end{equation*}
For non-empty, disjoint subsets $I, J \subset \dset$, we write
\[
	\nupdf_{I|J}(\by_I|\by_J)
	= \frac{\nupdf_{I \cup J}(\by_{I \cup J})}{\nupdf_J(\by_J)},
\]
for $\by_I \in (0, \infty)^I$ and $\by_J \in (0, \infty)^J$ such that $\nupdf_J(\by_J) > 0$. If $I = \cbr{i}$, we consider the conditional tail function
\[
	\bexpmeas_{i|J}(y_i|\by_J)
	= \int_{y_i}^{\infty} \nupdf_{i|J}(t|\by_J) \, \diff t
	= \xcdf_{i|J}(1/y_i \mid 1/\by_J)
\]
as well as its inverse $\bexpmeas_{i|J}^{-1}$ defined by
\[
	\bexpmeas_{i|J}^{-1}(u_i|\by_J) = y_i
	\iff 
	\bexpmeas_{i|J}(y_i|\by_J) = u_i
\]
for $u_i \in (0, 1)$. The copula density in \eqref{eq:cIJ} becomes
\begin{equation}
\label{eq:lam2c}
	c_{I;J}(\bu_I;1/\by_J)
	= \frac{\nupdf(\by_I|\by_J)}{\prod_{i \in I} \nupdf_{i|J}(y_i|\by_J)}
	\quad \text{with} \quad
	y_i = \bexpmeas_{i|J}^{-1}(u_i|\by_J).
\end{equation}

Under the simplifying assumption, the exponent measure density $\nupdf$ corresponding to the X-vine tail copula density $\xdf$ in \eqref{eq:xdfdecomp} is
\[
	\nupdf(\by) = 
	\prod_{j \in \dset} y_j^{2\deg_1(j) - 2}
	\cdot 
	\prod_{e \in \edges_1} \nupdf_{a_e,b_e}(y_{a_e},y_{b_e}) 
	\cdot \prod_{j=2}^{d-1} \prod_{e \in \edges_j}
	c_{a_e,b_e;\cndg_e} \rbr{
		\bexpmeas_{a_e|\cndg_e}(y_{a_e}|\by_{\cndg_e}), \,
		\bexpmeas_{b_e|\cndg_e}(y_{b_e}|\by_{\cndg_e})
	}.
\]
Here, $\deg_1(j)$ is the degree (number of neighbors) of node $j$ in tree $\tree_1$ and the pair copula densities can be found from $\nupdf$ by
\[
	c_{a_e,b_e;\cndg_e}(u_{a_e}, u_{b_e})
	= \frac{\lambda_{a_e,b_e|\cndg_e}(y_{a_{e}},y_{b_{e}}| \by_{\cndg_e})}{\lambda_{a_e|\cndg_e}(y_{a_{e}}|\by_{\cndg_e}) \cdot \lambda_{b_e|\cndg_e}(y_{b_{e}}|\by_{\cndg_e})}
\]
for $(u_{a_e}, u_{b_e}) \in (0, 1)^2$,
with $y_{a_e} = \bexpmeas_{a_e|\cndg_e}^{-1}(u_{a_e}|\by_{\cndg_e})$ and $y_{b_e} = \bexpmeas_{b_e|\cndg_e}^{-1}(u_{b_e}|\by_{\cndg_e})$.

\section{Estimating the copula parameters $\theta_{\cpdfset}$ in $\tree_2,\ldots,\tree_{d-1}$.}
\label{sec:estimationdetail}

Start with an edge $e = (a_e,b_e ; \cndg_e) \in \edges_2$, i.e., such that $\cndg_e$ is a singleton, $\cndg_e = \cbr{d_e}$, say. Then $f = \cbr{a_e, d_e}$ and $g = \cbr{b_e, d_e}$ are edges in $\edges_1$. By Proposition~\ref{prop:Z2c}(ii), \chng{$c_{a_e,b_e;\cndg_e}(\point,\point;\theta_e)$ is the density of the random pair
$\rbr{\xcdf_{a_e|d_e}(\Imp_{a_e}|\Imp_{d_e};\theta_f), \xcdf_{b_e|d_e}(\Imp_{b_e}|\Imp_{d_e};\theta_g)}$}.
The tail copula density parameters $\theta_f$ and $\theta_g$ have been estimated in the previous step, yielding estimates $\htheta_f$ and $\htheta_g$.
Define pseudo-observations $(\hU_{i,a_e;\cndg_e}, \hU_{i,b_e;\cndg_e})$ for $i \in \iset_{\cndg_e} = \iset_{d_e}$ from $c_{a_e,b_e;\cndg_e}$ by
$\hU_{i,a_e ; \cndg_e}
    = \xcdf_{a_e | \cndg_e} (\hZ_{i,a_e} \mid \hZ_{i,d_e} ; \; \htheta_{f})$ and $\hU_{i,b_e ; \cndg_e}  
    = \xcdf_{b_e | \cndg_e} (\hZ_{i,b_e} \mid \hZ_{i,d_e} ; \; \htheta_{g})$.
We estimate the parameter (vector) $\theta_{e} \in \Theta_{e}$ by maximising the pseudo-likelihood
\begin{equation}
\label{eq:lik2}
	\lik_{\cpdfset, e} \rbr{
		\theta_e ; \; 
		\rbr{
			\hU_{i,a_e ; \cndg_e}, 
			\hU_{i,b_e ; \cndg_e}
		}, i \in \iset_{\cndg_e} 
	}
	=  \prod_{i \in \iset_{\cndg_e}} 
	c_{a_e,b_e ; \cndg_e} \rbr{ 
		\hU_{i,a_e ; \cndg_e}, \hU_{i,b_e ; \cndg_e} ; \,
		\theta_{e} 
	}.
\end{equation}
Other estimators are possible based on parametric relationships between \chng{dependence coefficients such as Kendall's tau or Spearman's rho and} the parameters of corresponding bivariate copula densities \citep{czado2019analyzing}. 

Parameters (or parameter vectors) associated with edges $e = (a_e, b_e; \cndg_e)$ in $\edges_j$ for layers $j \in \cbr{3,\ldots,d-1}$ can be estimated similarly, proceeding recursively.
Suppose that all parameters associated with edges in $\edges_{1:(j-1)} = \edges_1 \cup \cdots \cup \edges_{j-1}$ have already been estimated; let $\theta(\edges_{1:(j-1)})$ denote the parameters and $\htheta(\edges_{1:(j-1)})$ their estimate.
By Proposition~\ref{prop:Z2c}(i), \chng{$c_{a_e,b_e; \cndg_e}$ is the density of the random pair
$\rbr{
		\xcdf_{a_e|\cndg_e} \rbr{
			\Imp_{a_e} \mid \IMP_{\cndg_e}; \;
			\theta(\edges_{1:j-1})
		}, \;
		\xcdf \rbr{
			\Imp_{b_e} \mid \IMP_{\cndg_e}; \;
			\theta(\edges_{1:j-1})
		}
	}$.} Define pseudo-observations $(\hU_{i,a_e;\cndg_e}, \hU_{i,b_e;\cndg_e})$ for $i \in \iset_{\cndg_e}$ from $c_{a_e,b_e; \cndg_e}$ by
\begin{equation*}
	\hU_{i,a_e ; \cndg_e}
	= \xcdf_{a_e | \cndg_e} \rbr{\hZ_{i,a_e} \mid \hbZ_{i,\cndg_e} ; \; \htheta(\edges_{1:j-1})}, \quad
	\hU_{i,b_e ; \cndg_e}  
	= \xcdf_{b_e | \cndg_e} \rbr{\hZ_{i,b_e} \mid \hbZ_{i,\cndg_e} ; \, \htheta(\edges_{1:j-1})}.
\end{equation*}
\chng{The conditional distribution functions} $\xcdf_{a_e|\cndg_e}$ and $\xcdf_{b_e|\cndg_e}$ are calculated recursively via Eq.~\eqref{eq:xvine:recur} with $e$ in that formula replaced by the edges $f, g \in \edges_{j-1}$ that are joined by $e$, that is, such that $e = \cbr{f, g}$.
The parameter (vector) $\theta_{e}$ is then obtained by maximising the pseudo-likelihood~\eqref{eq:lik2}.

\section{Tail dependence coefficients}
\label{app:chi}

\subsection{\chng{Tail dependence coefficients for parametric families of tail copulas}}\label{sec:chi}

Dependence measures facilitate the comparison of parameter estimates across different models on a consistent scale.
In particular, closed-form expressions for the bivariate tail dependence coefficient $\chi = \xcdf(1,1)$ are well-established for widely used parametric families of tail copula densities $\xcdf$, albeit it in the form of exponent measures. 

\begin{itemize}
	\item \HR{} model, $\xdf(x_1,x_2;\theta)=x_{1}^{-1}\frac{1}{\sqrt{2\pi\theta}}\exp\left[-\tfrac{1}{2}\left\{\frac{\log(x_{1}/x_{2})-\theta/2}{\sqrt{\theta}}\right\}\right]$ for $\theta > 0$: we have $\chi=2-2\Phi(\sqrt{\theta}/2)$ where $\Phi$ is the standard normal cumulative distribution function.
	\item Negative logistic model, $\xdf(x_1,x_2;\theta)=(1+\theta)(x_{1}x_{2})^{-(\theta+1)}(x_{1}^{-\theta}+x_{2}^{-\theta})^{-1/\theta-2}$ for $\theta > 0$: we have $\chi=2^{-1/\theta}$.
	\item Logistic model, $\xdf(x_1,x_2;\theta)=(\theta-1)(x_{1}x_{2})^{\theta-1}(x_{1}^{\theta}+x_{2}^{\theta})^{1/\theta-2}$ for $\theta > 1$: we have $\chi=2-2^{1/\theta}$.
	\item Dirichlet model, $\xdf(x_1,x_2;\theta)=\frac{2\Gamma(2\theta)}{\Gamma(\theta)^2}\left(x_{1}+x_{2}\right)^{-2\theta-1}(x_{1}x_{2})^{\theta}$ for $\theta > 0$: we have $\chi=\int_{0}^{1} I_{\frac{1}{1+x}}(\theta+1,\theta) \, \diff x$ where $I_{\frac{1}{1+x}}$ is the regularised (incomplete) Beta function. 
\end{itemize}
For the parametric relationships between Kendall's tau $\tau_e$ and pair-copula parameters, see for instance \citet[Section~3.5]{czado2019analyzing}.

\subsection{\chng{Tail dependence coefficients as goodness-of-fit measures}}
\label{sec:taildepcoefgof}

In Sections~\ref{sec:simulation-study} and \ref{sec:casestudy}, we frequently use (pairwise) tail dependence coefficients to assess the goodness-of-fit of the estimated X-vine models. The aim of this section is to explain the difference between what we call \emph{model-based}, \emph{empirical}, and \emph{fitted} tail dependence coefficients, and in particular, to explain how these are computed. First, recall the definition of the dependence coefficient $\chi_J$ for sets $J \subseteq \dset$ with two or more elements in Section~\ref{sec:background}. Note that
\begin{align*}
	\chi_{J} & = \xcdf \rbr{ \left\{\bx \in [0,\infty)^d : \max_{j \in J} x_j < 1 \right\} }  
	=  \xcdf \rbr{ \left\{ \bx \in [0,\infty)^d : \max_{j \in J} x_j < 1 \right\} \cap \LL} \\
	& = \xcdf \rbr{\LL} \prob \rbr{\max_{j \in J} Z_j < 1}  = \prob \rbr{ \max_{j \in J \setminus k} Z_j < 1 \mid Z_k < 1},
\end{align*}
since $\prob \rbr{Z_k < 1} = 1/R(\LL)$ for any $k \in J$. Empirical estimates of $\chi_J$ will be based on the right-most expression in above equation.

\subsubsection{Model-based tail dependence coefficients} 
Model-based $\chi$'s, used in Section~\ref{sec:simulation-study}, are those corresponding to the true X-vine specification. The parametric expressions given in Section~\ref{sec:chi} can only be used for pairs of variables connected by an edge in the first tree of the X-vine, i.e., for $\chi_{ab}$ with $\{a,b\} \in E_1$. 
More generally, model-based $\chi_{J}$'s for sets $J \subseteq \dset$ with two or more elements will have to be calculated via Monte Carlo simulation:
\begin{compactenum}
	\item Simulate a sample of size $N$ from the inverted multivariate Pareto distribution associated with the true X-vine specification, $\bZ_1,\ldots,\bZ_N$. In Section~\ref{sec:simulation-study}, we used $N = 5000$.
	\item 
	Use this sample to estimate $\chi_J$ for non-empty $J \subseteq \dset$ as the average of the $|J|$ ratios
	\begin{equation}\label{eq:chiab}
	\frac{|\bigcap_{j \in J} Q_j|}{|Q_k|}, \qquad k \in J,
	\end{equation}
	where $Q_j = \cbr{i = 1, \ldots,N : Z_{i,j} <  1}$ for $j \in \dset$.
\end{compactenum} 

\subsubsection{Empirical tail dependence coefficients}
In Section~\ref{sec:casestudy}, we calculate the empirical tail dependence coefficients for the large flight delays data $\hat{\bx}_1,\ldots,\hat{\bx}_n$ as follows.
\begin{compactenum}
	\item Apply the rank transformation of Eq.~\eqref{eq:margins:ranks} to obtain a sample $\hat{\bu}_1,\ldots,\hat{\bu}_n$. 
	\item For non-empty $J \subseteq \dset$, estimate $\chi_J$ by
	\[
	\hchi_J = \frac{1}{k} \sum_{i=1}^n \prod_{j\in J} \1 \cbr{\hat{u}_{i,j} < k/n}.
	\]
\end{compactenum}
Using instead \eqref{eq:chiab} applied to the sample $\hat{\bz}_i = (n/k) \hat{\bu}_i$ for $i = 1,\ldots,n$ leads to the same estimate as $\hchi_{J}$. Indeed, replacing $Z_{i,j}$ by $\hat{z}_{i,j}$ in the definition of $Q_j$ gives the set 
$K_j = \{i = 1,\ldots,n : \hat{u}_{i,j} < 1\}$ and $|K_j| = k$ for all $j \in \dset$.

\subsubsection{Fitted tail dependence coefficients} 
\label{sss:fittedtaildepcoef}
Sections~\ref{sec:simulation-study} and \ref{sec:casestudy} both use fitted tail dependence coefficients. 
In the simulation study in Section~\ref{sec:simulation-study}, these are calculated as the model-based tail dependence coefficients described above; the only difference is that in step~1, we simulate a sample of size $N = 5000$ from the \emph{fitted} X-vine model.
However, in Section~\ref{sec:casestudy}, the fitted tail dependence coefficients are calculated as follows:
\begin{compactenum}
    \item Simulate a sample of size $N = 3600$ (approximately the same sample size as that of the data) from the inverted multivariate Pareto distribution associated with the \emph{fitted} X-vine model, $\bZ_1, \ldots, \bZ_N$.
    \item Apply the procedure outlined in the paragraph ``Empirical tail dependence coefficients'' to the multivariate Pareto sample $\bZ_1, \ldots, \bZ_N$. 
\end{compactenum}
The adjusted sample size of $N=3600$ and our treatment of the multivariate Pareto sample as if stemming from an unknown distribution ensure that the fitted tail dependence coefficients are comparable to the empirical ones in terms of bias.

\section{Additional numerical results}
\label{sec:additional}

\subsection{Additional simulation results to those in Section~\ref{sec:simulation-study}}\label{sec:simu:additional}

\subsubsection{\chng{Box-plots of dependence measure estimates based on pre-specified bivariate parametric families for each edge}}
\label{sec:simu:addtional:dep}

\chng{
Similar to Fig.~\ref{fig:boxplot_Fam}, but this time assuming the same vine structure and the pre-specified bivariate parametric families in the
X-vine specification shown in Fig.~\ref{fig:5d_Xvine_a}, we create box-plots of dependence measure estimates in Fig.~\ref{fig:boxplot_dep_2}.
As expected, we see lower variability and less bias, especially in deeper trees (e.g. $\tree_3$ and $\tree_4$).}

\begin{figure}[h!]%
    \centering
  
  \includegraphics[width=0.5\textwidth]{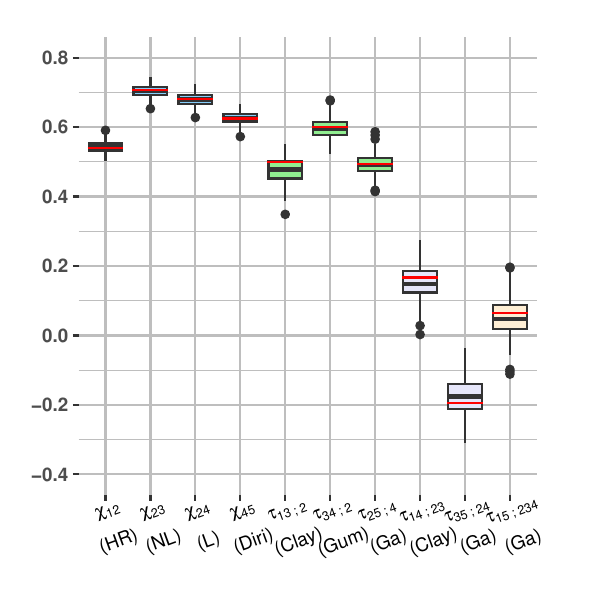}%
    \caption{\label{fig:boxplot_dep_2}
	Box-plots (\textcolor{lightskyblue}{$\blacksquare$} for $\tree_1$
    \textcolor{lightgreen}{$\blacksquare$} for $\tree_2$
    \textcolor{lavender(web)}{$\blacksquare$} for $\tree_3$
    \textcolor{papayawhip}{$\blacksquare$} for $\tree_4$) of dependence measure estimates from specified bivariate parametric families for each edge across the four trees in the X-vine specification in Fig.~\ref{fig:5d_Xvine_a}.}.
\end{figure}

\subsubsection{Box-plots of dependence measure estimates with varying $(n,k)$}
\label{sec:simu:addtional:nk}

Given the X-vine specification in Fig.~\ref{fig:5d_Xvine_a}, Fig.~\ref{fig1additional} shows box-plots of dependence measure estimates from the specified bivariate parametric families for each edge across the four trees, using various $(n,k)$, following the same set-up as in Section~\ref{sec:simu:addtional:dep}.
We consider three different sample sizes, $n \in \cbr{1\,000, 2\,000, 4\,000}$, and quantiles, $k/n \in \cbr{0.02,0.05,0.1}$. \chng{As expected, variability decreases as $k$ increases (for fixed $n$) or as $n$ increases (for fixed $k$).}

\begin{figure}[p]%
    \centering
    \subfloat[]{\includegraphics[width=0.33\textwidth]{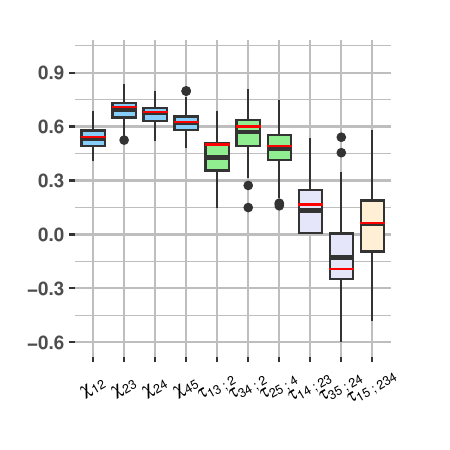}}%
    \enspace
    \subfloat[]{\includegraphics[width=0.33\textwidth]{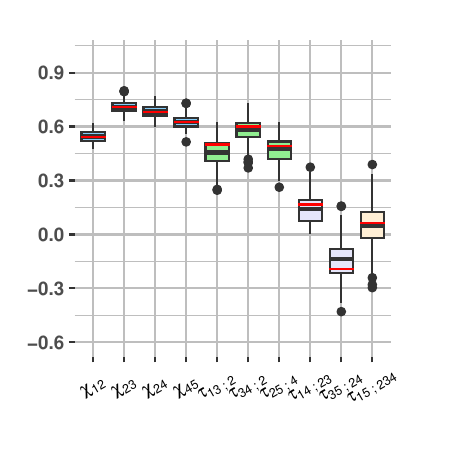}}%
    \enspace
    \subfloat[]{\includegraphics[width=0.33\textwidth]{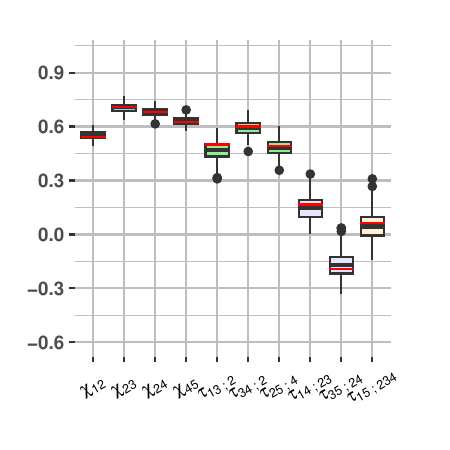}}%
    \\
    \subfloat[]{\includegraphics[width=0.33\textwidth]{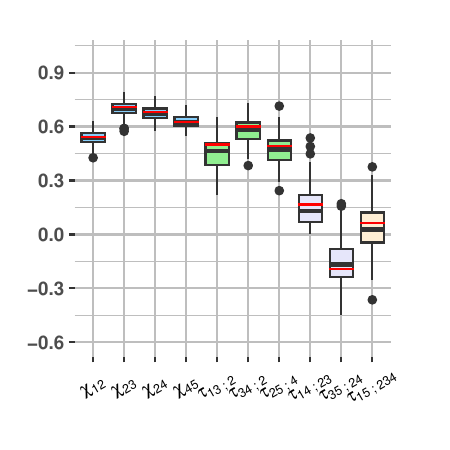}}%
    \enspace
    \subfloat[]{\includegraphics[width=0.33\textwidth]{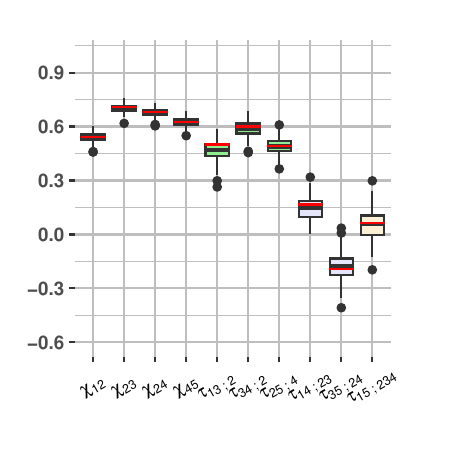}}%
    \enspace
    \subfloat[]{\includegraphics[width=0.33\textwidth]{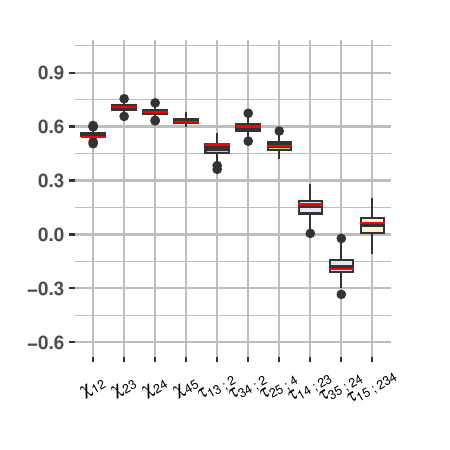}}%
    \\
    \subfloat[]{\includegraphics[width=0.33\textwidth]{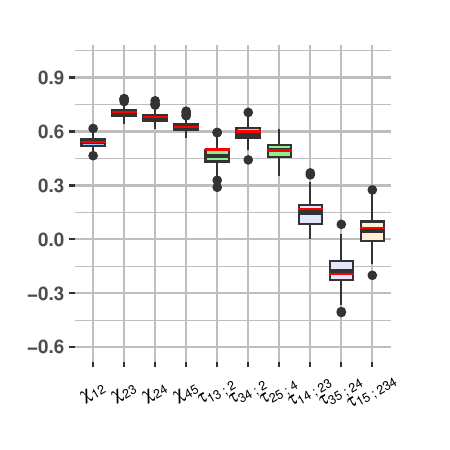}}%
    \enspace
    \subfloat[]{\includegraphics[width=0.33\textwidth]{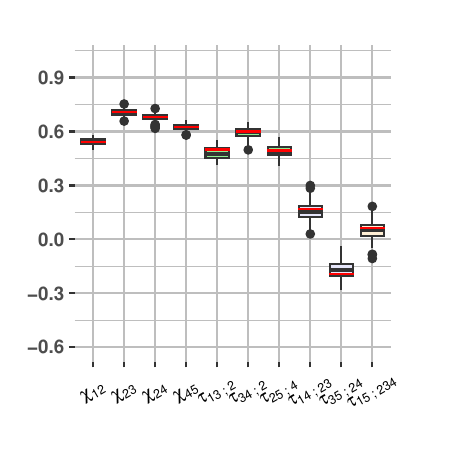}}%
    \enspace
    \subfloat[]{\includegraphics[width=0.33\textwidth]{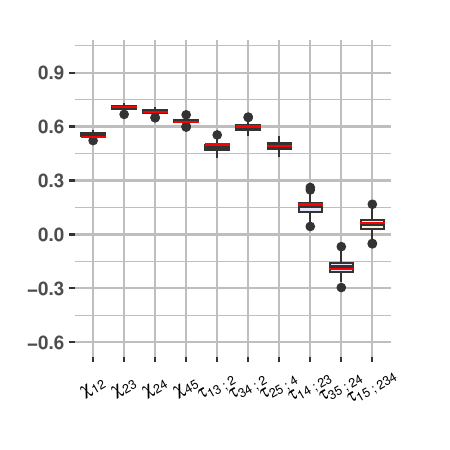}}%
    \caption{\label{fig1additional} Box-plots of dependence measure estimates from specified parametric models in the X-vine specification shown in Fig.~\ref{fig:5d_Xvine_a}, varying with $(n,k)$: (a) $(n,k)=(1\,000,20)$ (b) $(n,k)=(1\,000,50)$ (c) $(n,k)=(1\,000,100)$ (d) $(n,k)=(2\,000,40)$  (e) $(n,k)=(2\,000,100)$ (f) $(n,k)=(2\,000,200)$ (g) $(n,k)=(4\,000,80)$ (h) $(n,k)=(4\,000,200)$ (i) $(n,k)=(4\,000,400)$}%
\end{figure}

\subsubsection{Box-plots of pseudo-maximum likelihood estimates of parameters of bivariate tail copulas}
\label{sec:boxplotmle}

Using the same X-vine specification in Fig.~\ref{fig:5d_Xvine_a}, we evaluate bias and variance of maximum pseudo-likelihood parameter estimates of tail copula densities.
Specifically, Fig.~\ref{fig:bp_onepar} shows box-plots of estimates for the \HR{} tail copula density associated with the edge $e_{12} \in \edges_1$, using the same combinations of $(n,k)$ defined in Section~\ref{sec:simu:addtional:nk}.
An increase in the effective sample size $n_{D_e}$ reduces variability, while a decreasing threshold $k/n$ leads to smaller bias, as expected.
Similar results are observed with other tail copula densities (figures are omitted).

\begin{figure}[ht]%
    \centering
    \subfloat[]{\includegraphics[width=0.33\textwidth]{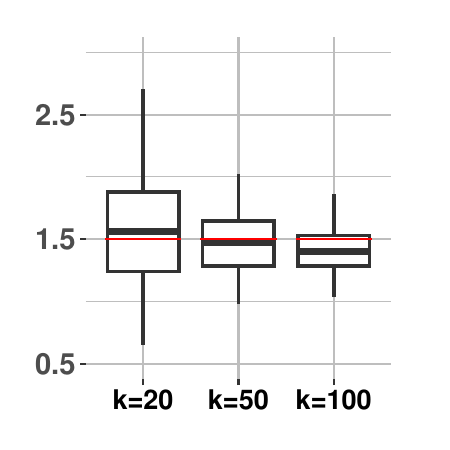}}%
    \enspace
    \subfloat[]{\includegraphics[width=0.33\textwidth]{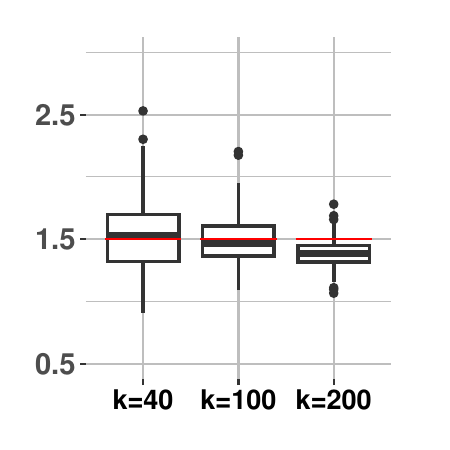}}%
    \enspace
    \subfloat[]{\includegraphics[width=0.33\textwidth]{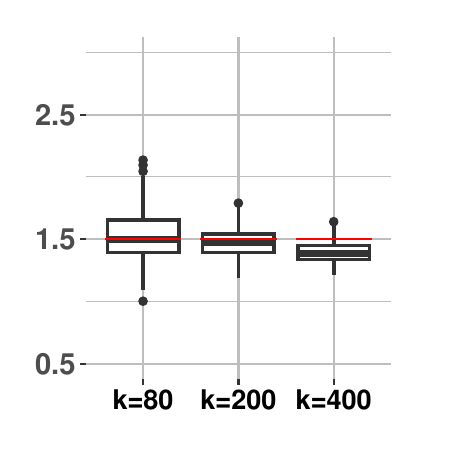}}%
    \caption{\label{fig:bp_onepar}Box-plots of maximum pseudo-likelihood parameter estimates for the bivariate \HR{} tail copula density with varying $(n,k)$: (a) $n=1\,000$ (b) $n=2\,000$, and (c) $n=4\,000$. The red line indicates the specified parameter value of $\theta=1.5$.}%
\end{figure}

\subsubsection{Averaged-effective sample sizes in different trees}\label{sec:avsamplesize}
We also investigate how effective sample sizes change over tree levels of Section~\ref{sec:simulation-study:select}.
We calculate the averaged-effective sample size over 200 repetitions for each edge in each tree and express it as a percentage of the sample size $n$:
\begin{center}
        \begin{tabular}{rrrrr}
        \midrule
        $\tree_1$ & 7.30\% & 6.47\% & 6.60\% & 6.87\% \\ 
        $\tree_2$ &  & 5.00\% & 5.00\% & 5.00\% \\ 
        $\tree_3$ &  &  & 3.52\% & 3.38\% \\ 
        $\tree_4$ &  &  &  & 3.12\% \\ 
        \bottomrule
        \end{tabular}
        \end{center}

        \chng{The order of the components from left to right and then from top to bottom} corresponds to the x-axis order of box-plots in Fig.~\ref{fig:boxplot_Fam}.
The overall averaged-effective sample size decreases in the first level tree when considering stronger extremal structures (result is omitted). 
This reduction is due to the fact that, for a high threshold $u$, the probability $\prob[X_a > u \text{ or } X_b > u]$ \emph{decreases} as the tail dependence between $X_a$ and $X_b$ \emph{increases}.
For the second level tree $\tree_2$, the averaged-effective sample is $k/n = 5\%$, since the use of the rank transform implies (in the absence of ties) the identity $\iset_j = k$ for all $j \in \dset$.

\subsection{Additional results for US flight delay data in Section~\ref{sec:casestudy}}
\label{sec:flight:additional}

\subsubsection{Estimated regular vine sequence of the X-vine model from $\tree_1$ to $\tree_7$}

Assigning $\hchi_e$ as edge weight for $e \in \edges_1$ and $\htau_e$ for $e \in \edges_j$ with $j \ge 2$, we obtain the sequence of maximum spanning trees that satisfy the proximity condition.
With the mBIC-optimal truncation level of $q^*=7$, Figure~\ref{fig:XvineFlightGraph} displays the trees of the X-vine model from $\tree_1$ to $\tree_7$ sequentially. For trees $\tree_2$ to $\tree_7$, the edges for which the independence copula is selected are not shown.

\begin{figure}[h!]%
    \centering
    \subfloat[]{\label{fig:XvineFlightGrapha}
\includegraphics[width=0.24\textwidth]{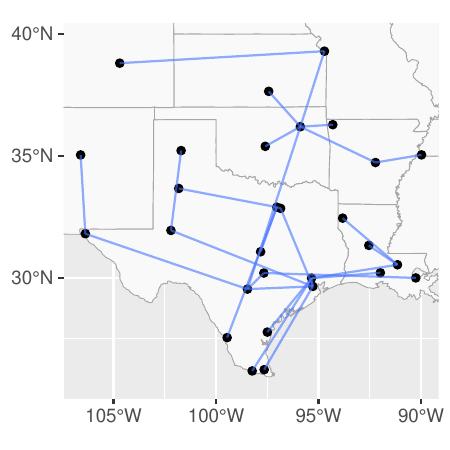}}%
    \enspace
    \subfloat[]{\includegraphics[width=0.24\textwidth]{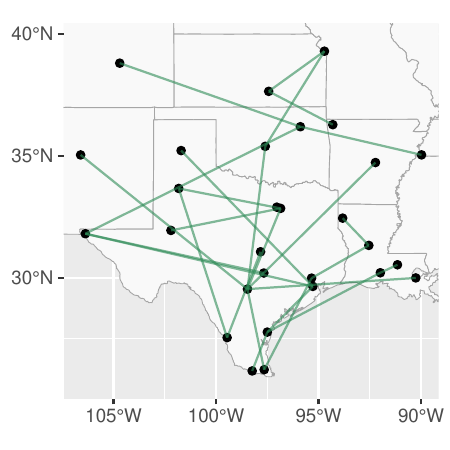}}%
    \enspace
    \subfloat[]{\includegraphics[width=0.24\textwidth]{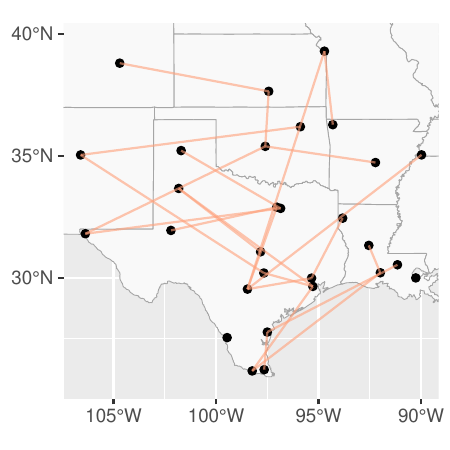}}%
    \enspace
    \subfloat[]{\includegraphics[width=0.24\textwidth]{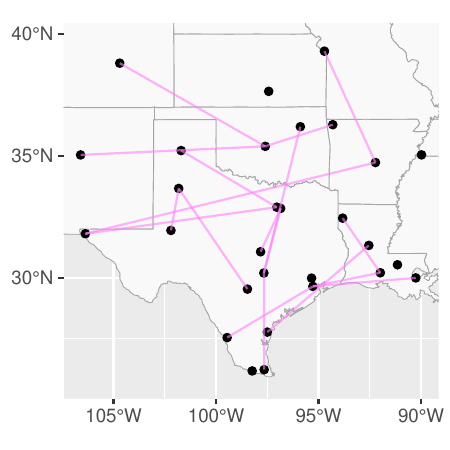}}%
    \\
    \subfloat[]{\includegraphics[width=0.24\textwidth]{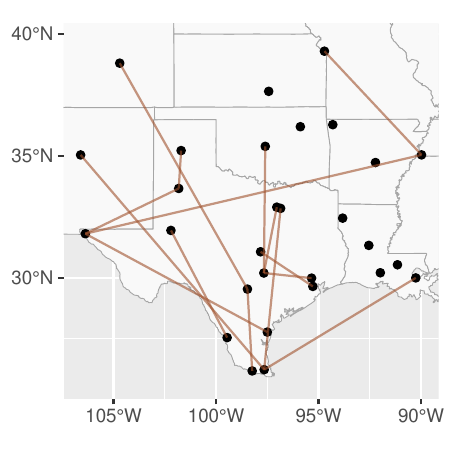}}%
    \enspace
    \subfloat[]{\includegraphics[width=0.24\textwidth]{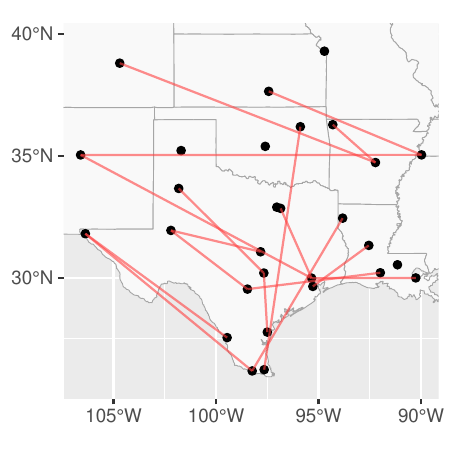}}%
    \enspace
    \subfloat[]{\includegraphics[width=0.24\textwidth]{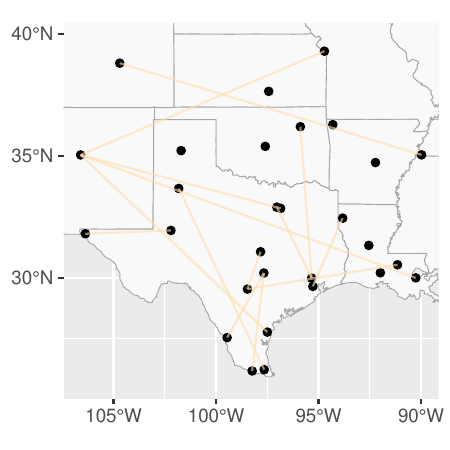}}%
    \caption{\label{fig:XvineFlightGraph}
Sequential display of subsequent trees of the truncated X-vine model. For trees $\tree_2$ to $\tree_7$, the edges for which the independence copula is selected are not shown. (a) $\tree_1$ (b) $\tree_2$ (c) $\tree_3$ (d) $\tree_4$ (e) $\tree_5$ (f) $\tree_6$ (g) $\tree_7$.}%
\end{figure}

\subsubsection{Goodness of fit for the fitted (truncated) X-vine model}
\label{sec:gofXvine}

To assess the goodness of fit for the full X-vine model, we compare empirical pairwise $\hchi$ values with those derived from the full X-vine model in Fig.~\ref{fig:ChiflightFullXvine}.
As an indicator of capturing a higher-dimensional extremal dependence, we use a trivariate (lower) tail dependence coefficient $\chi_{abc} = \xcdf_{abc}(1,1,1)$ for distinct $a,b,c\in \dset$. 
The empirical trivariate tail dependence coefficients \chng{(Section~\ref{sec:taildepcoefgof})} are compared with those derived from the full X-vine model in Fig.~\ref{fig:Chi3FlightFull} and from the truncated X-vine model in Fig.~\ref{fig:Chi3FlightTrunc}, respectively.

\begin{figure}[h!]%
    \centering
    \subfloat[\label{fig:ChiflightFullXvine}]{\includegraphics[width=0.33\textwidth]{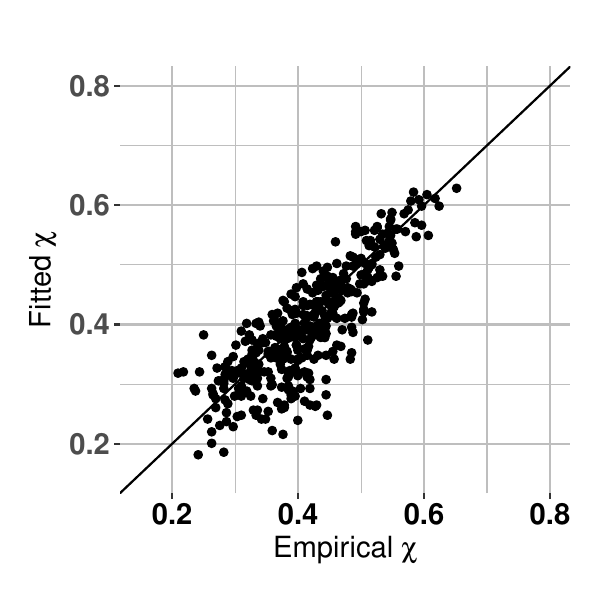}}%
    \enspace
    \subfloat[\label{fig:Chi3FlightFull}]{\includegraphics[width=0.33\textwidth]{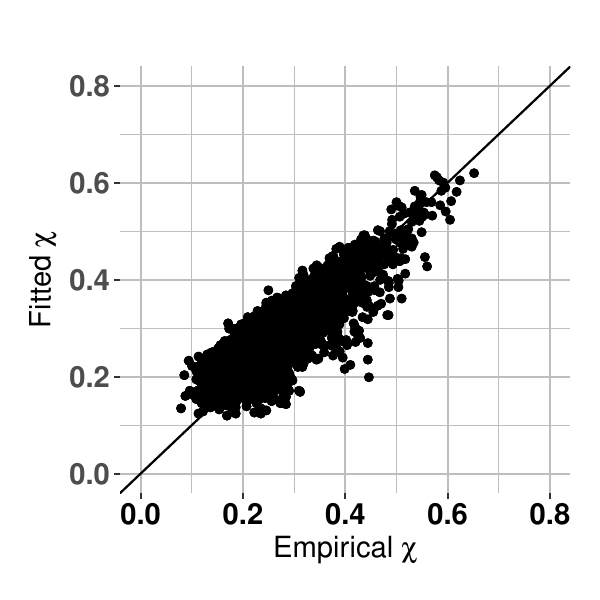}}%
    \enspace
    \subfloat[\label{fig:Chi3FlightTrunc}]{\includegraphics[width=0.33\textwidth]{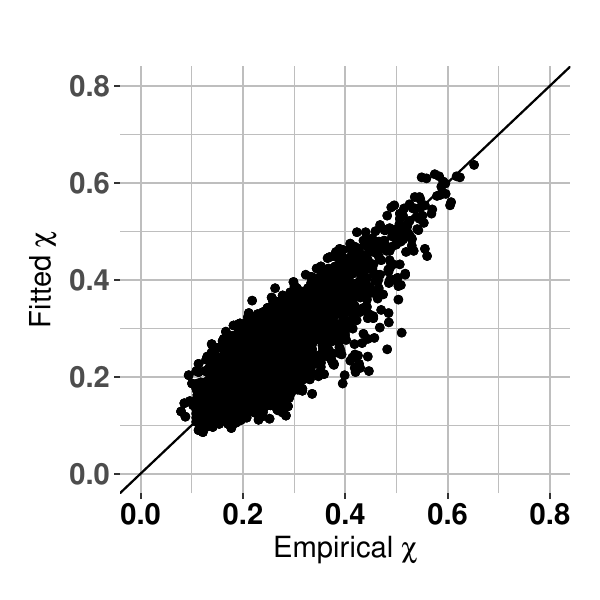}}%
    \caption{\label{fig:chifit}(a) A $\chi$-plot comparing empirical pairwise tail dependence coefficients $\hchi_{a,b}$ for distinct $a,b\in\dset$ with those derived from the full X-vine model. (b) A plot comparing empirical trivariate tail dependence coefficients $\hchi_{a,b,c}$ for distinct $a,b,c\in\dset$ with those obtained from the full X-vine model. (c) similar to (b), but for the truncated X-vine model with a truncation level of $q^*=7$.}
\end{figure}

\subsubsection{Sensitivity analysis for mBIC truncation levels over quantiles}

\chng{Fig.~\ref{fig:mBICFlight} illustrates the mBIC-optimal truncation level of $q^*=7$.
We also conduct a sensitivity analysis to observe how the mBIC-optimal truncation level changes over quantiles. Fig.~\ref{fig:mBICvsQuan} shows a} decline in the mBIC-optimal truncation level as the quantile $k/n$ decreases.
The range of mBIC-optimal truncation levels is from 4 to 8 with a median of 6.

\begin{figure}[h!]%
    \centering
    \subfloat[\label{fig:mBICFlight}]{\includegraphics[width=0.34\textwidth]{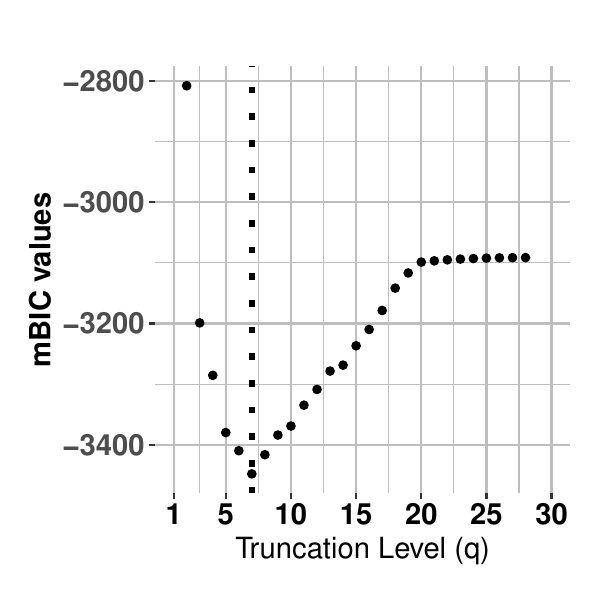}}%
    \qquad
    \subfloat[\label{fig:mBICvsQuan}]{\includegraphics[width=0.32\textwidth]{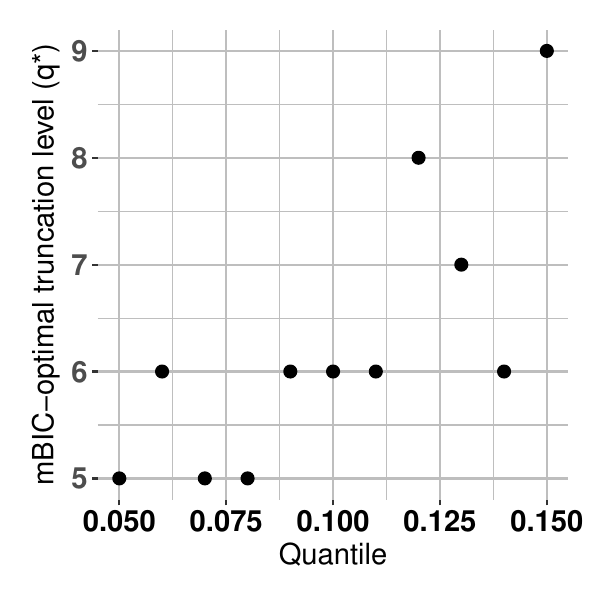}}%
    \caption{(a) mBIC-values plotted across tree levels with a dotted line indicating a selected mBIC-optimal truncation level of $q^*=7$.
    (b) Plot of the mBIC-optimal truncation level across a sequence of quantiles $k/n\in(0.05, 0.15)$.}%
\end{figure}

\subsubsection{Model comparison: X-vine \HR{} model versus \HR{} extremal graphical model from \textsf{EGlearn}}
\label{app:HRcomp}

Focusing on the \HR{} distribution, we compare model assessments between the X-vine \HR{}
model (all tail copulas in $\tree_1$ are \HR{} and all bivariate copulas in the subsequent trees are Gaussian) and the \HR{} extremal graphical model obtained with \textsf{EGlearn} algorithm \citep{engelke2021learning}.
In both models, within $\tree_1$, we use the empirical extremal variogram matrix $\hGamma$ as edge weight for the minimum spanning tree.
The resulting tree must be identical for both models, as confirmed in Fig.~\ref{fig:XvineHRgraph}.
As described in Section~\ref{sec:casestudy}, with the tuning parameter value of $\rho^*=0.1$ selected by the algorithm, the resulting sparse extremal graph has 148 edges, shown in Fig.~\ref{fig:ChiFlightEGlearn}.
For the \HR{} X-vine, 
the algorithm sets all bivariate copulas to the independence copula from tree level 19 on (i.e., the truncation level is $q=18$), selecting 217 independence copulas out of 406 edges in trees $\tree_1$ to $\tree_{28}$ (recall $d = 29$).
Using the mBIC, we further consider the truncated \HR{} X-vine with a lower truncated level.
The optimal mBIC truncation level of $q^*=7$ results in a total of 175 edges in edges $\tree_1$ to $\tree_7$, shown in Fig.~\ref{fig:mBICFlight}. In trees $\tree_2$ to $\tree_7$, there are $46$ edges out of $147$ for which the independence copula is selected.

We assess the goodness of fit between the different graphical models using the entire sub-sample, that is, we no longer split the data in a training and test set. 
\chng{As in Fig.~7 in \cite{hentschel2022statistical} and  Fig.~8(e) in \cite{engelke2024graphical}, an alternative approach to creating a $\chi$-plot for the \HR{} distribution is to convert the variogram elements $\Gamma_{ab}$ for distinct $a,b\in\dset$ to the tail dependence coefficients using the relationship $\chi_{ab}=2-2\Phi(\sqrt{\Gamma_{ab}}/2)$ for distinct $a,b\in\dset$.}

Fig.~\ref{fig:ChiComparisonHR} compares the tail dependence coefficients converted from the empirical variogram to those from the fitted variogram.
This comparision is done specifically for the extremal \HR{} graphical model based on the flight graph, the truncated \HR{} X-vine, and the extremal \HR{} graph derived from the \textsf{EGlearn} algorithm, respectively.
For the truncated \HR{} X-vine, a similar procedure to steps 1 and 2 in~Section~\ref{sss:fittedtaildepcoef} is employed to obtain the fitted variogram matrix and subsequently transform it into fitted $\chi$ values.
Similarly, like the truncated X-vine model, the truncated \HR{} X-vine exhibits more variability compared to the extremal \HR{} graphical model.

\begin{figure}[h!]%
    \centering
    \subfloat[]{\includegraphics[width=0.33\textwidth]{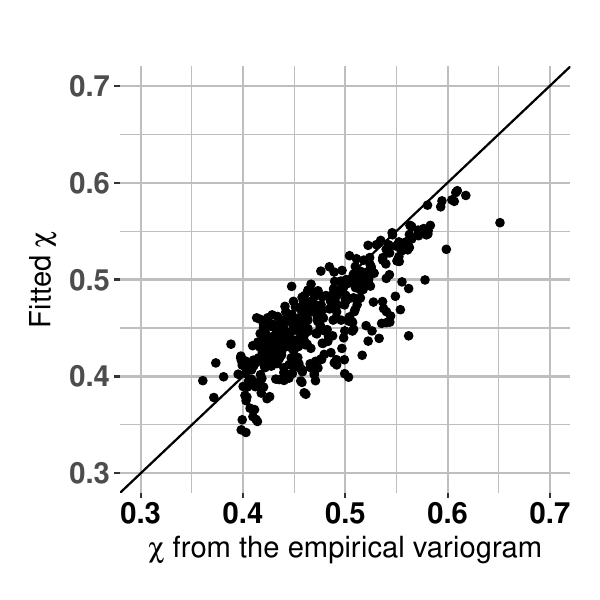}}%
    \enspace
    \subfloat[\label{fig:ChiFlightTruncXVineHR}]{\includegraphics[width=0.33\textwidth]{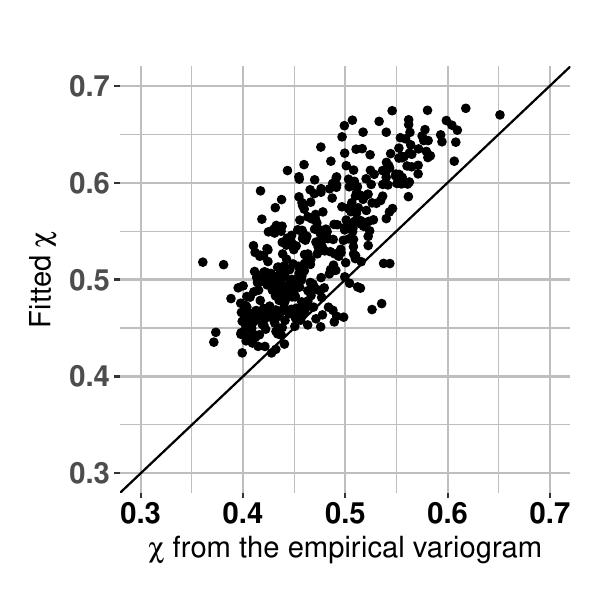}}%
    \enspace
    \subfloat[]{\includegraphics[width=0.33\textwidth]{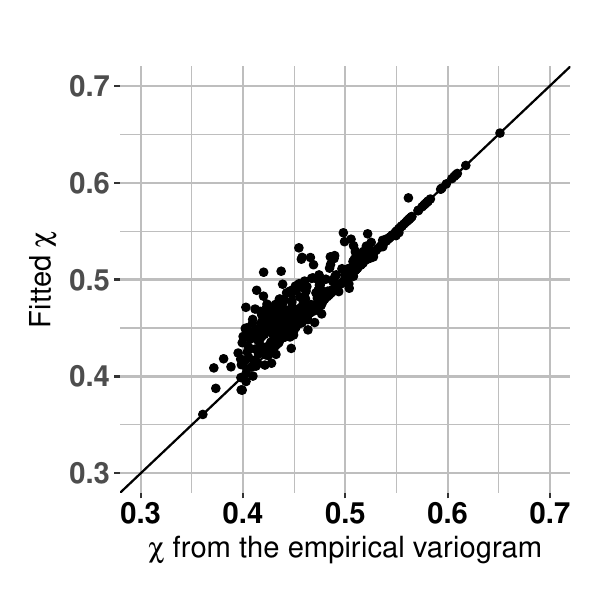}}%
\caption{\label{fig:ChiComparisonHR}
$\chi$-plots comparing \chng{pairwise $\hchi$ values converted from the empirical variogram to those from the fitted variogram for graphical models}: (a) the extremal \HR{} graphical model for the flight graph, (b) the truncated \HR{} X-vine with $q^*=7$, and (c) the extremal \HR{} graphical model with $\rho^*=0.1$ using the \textsf{EGlearn} algorithm.
}%
\end{figure}

\subsubsection{Estimated regular vine sequence of the truncated \HR{} X-vine}

In Fig.~\ref{fig:XvineHRgraph}, we sequentially display the selected trees $\tree_1$ to $\tree_7$ of the truncated \HR{} X-vine in Section~\ref{app:HRcomp}. For trees $\tree_2$ to $\tree_7$, there are $46$ out of $147$ edges for which the independence copula is selected, and these edges are not shown in the figure.

\begin{figure}[h!]%
    \centering
    \subfloat[]{\includegraphics[width=0.24\textwidth]{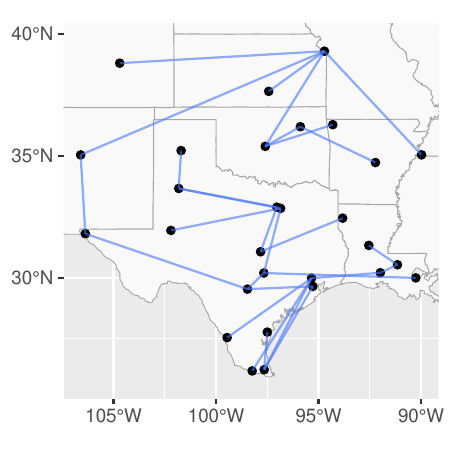}}%
    \enspace
    \subfloat[]{\includegraphics[width=0.24\textwidth]{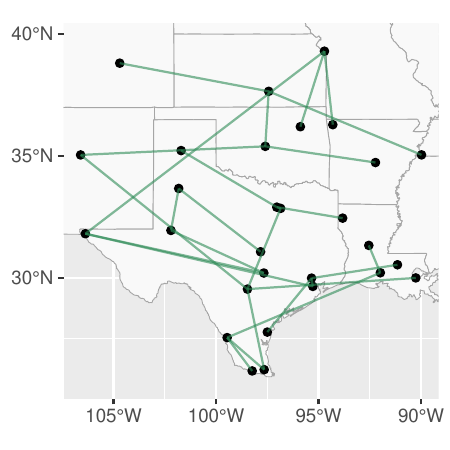}}%
    \enspace
    \subfloat[]{\includegraphics[width=0.24\textwidth]{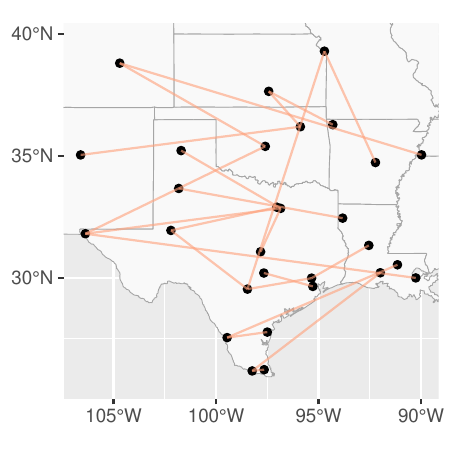}}%
    \enspace
    \subfloat[]{\includegraphics[width=0.24\textwidth]{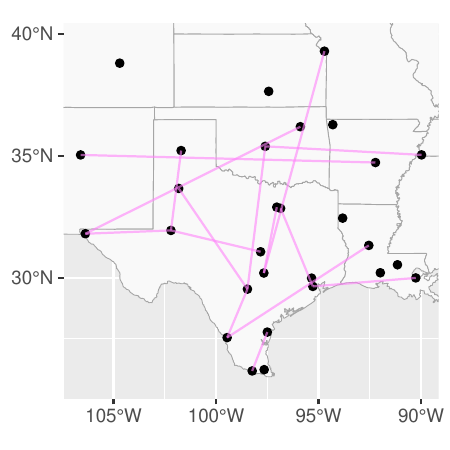}}%
    \\
    \subfloat[]{\includegraphics[width=0.24\textwidth]{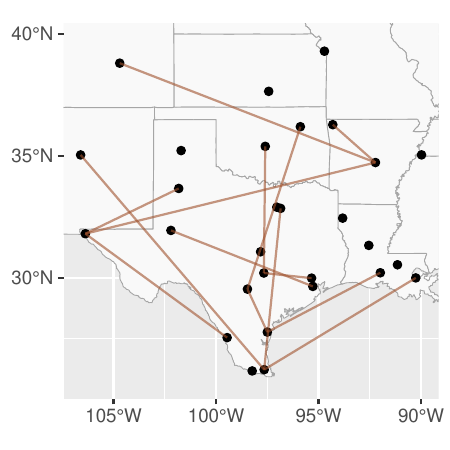}}%
    \enspace
    \subfloat[]{\includegraphics[width=0.24\textwidth]{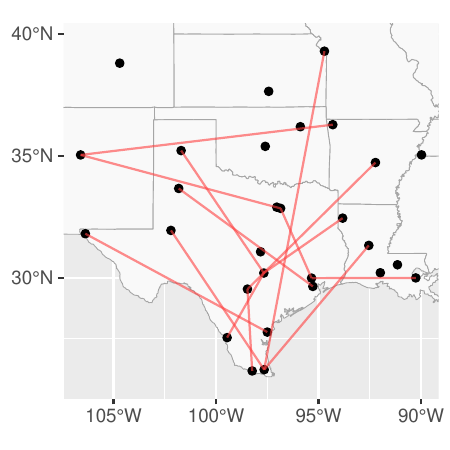}}%
    \enspace
    \subfloat[]{\includegraphics[width=0.24\textwidth]{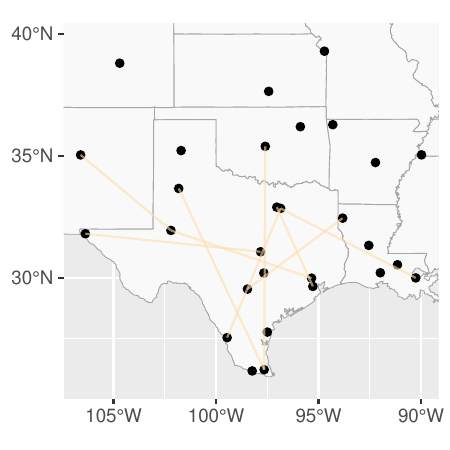}}%
    \caption{\label{fig:XvineHRgraph}
Sequential display of the first $q^* = 7$ trees of the truncated \HR{} X-vine fitted to the US flight data in Section~\ref{app:HRcomp}. For trees $\tree_2$ to $\tree_7$, the edges for which the independence copula is selected are not shown. (a) $\tree_1$ (b) $\tree_2$ (c) $\tree_3$ (d) $\tree_4$ (e) $\tree_5$ (f) $\tree_6$ (g) $\tree_7$. }%
\end{figure}

\end{document}